\documentclass[acmsmall,11pt,screen,nonacm,natbib=false]{acmart}
\usepackage[style=alphabetic,backend=biber]{biblatex}
\newcommand{\citet}{\textcite}
\usepackage{amsmath}

\usepackage{amsthm}

\usepackage{thmtools} 
\usepackage{mathtools} 
\usepackage{dsfont} 
\usepackage{subcaption}
\usepackage{siunitx}
\usepackage{longtable}
\usepackage{todonotes}

\usepackage{hyperref}
\usepackage{cleveref}

\crefname{appsec}{appendix}{appendices}
\Crefname{appsec}{Appendix}{Appendices}

\newtheorem{theorem}{Theorem}[section]
\newtheorem{lemma}[theorem]{Lemma}
\newtheorem{proposition}[theorem]{Proposition}
\newtheorem{corollary}[theorem]{Corollary}

\theoremstyle{definition}
\newtheorem{definition}[theorem]{Definition}

\newcommand{\GG}{\mathbb G}

\newcommand{\NN}{\mathbb N}
\newcommand{\PP}{\mathbb P}

\newcommand{\RR}{\mathbb R}

\newcommand{\cA}{\mathcal{A}}
\newcommand{\cB}{\mathcal{B}}
\newcommand{\cC}{\mathcal{C}}

\newcommand{\cF}{\mathcal{F}}
\newcommand{\cG}{\mathcal{G}}
\newcommand{\cH}{\mathcal{H}}

\newcommand{\cM}{\mathcal{M}}
\newcommand{\cN}{\mathcal{N}}

\newcommand{\cR}{\mathcal{R}}
\newcommand{\cS}{\mathcal{S}}
\newcommand{\cU}{\mathcal{U}}

\newcommand{\cX}{\mathcal{X}}

\newcommand{\ol}{\overline}

\newcommand{\rarr}{\rightarrow}

\newcommand{\bm}[1]{\boldsymbol{#1}}

\renewcommand{\vec}[1]{\mathbf{#1}}

\newcommand{\parens}[1]{\left(#1\right)}
\newcommand{\brackets}[1]{\left[#1\right]}
\newcommand{\braces}[1]{\left\{#1\right\}}

\newcommand{\abs}[1]{\left\lvert #1 \right\rvert}
\newcommand{\norm}[1]{\left\lVert #1 \right\rVert}

\newcommand{\opnorm}[1]{\norm{#1}_\mathrm{op}}
\newcommand{\floor}[1]{\left\lfloor #1 \right\rfloor}
\newcommand{\ceil}[1]{\left\lceil #1 \right\rceil}

\newcommand*\diff{\mathop{}\!\mathrm{d}}

\newcommand{\bra}[1]{\left\langle#1\right|}
\newcommand{\ket}[1]{\left|#1\right\rangle}

\newcommand{\selfouterprod}[1]{\ket{#1}\bra{#1}}
\newcommand{\sigmax}{\bm{\sigma}^{\mathrm{x}}}
\newcommand{\sigmay}{\bm{\sigma}^{\mathrm{y}}}
\newcommand{\sigmaz}{\bm{\sigma}^{\mathrm{z}}}

\newcommand{\iunit}{\bm{\mathrm{i}}}

\newcommand{\mixedn}{\cS^\mathrm{m}_n}
\newcommand{\probst}{p_\mathrm{st}}
\newcommand{\probf}{p_\mathrm{f}}
\newcommand{\probb}{p_\mathrm{b}}
\newcommand{\probest}{p_\mathrm{est}}
\newcommand{\dhamming}{d_\mathrm{H}}
\newcommand{\frakd}{\mathfrak{d}}
\newcommand{\cXfd}{\cX_{\mathfrak{d}}}
\newcommand{\GHZ}{\mathrm{GHZ}}

\DeclareMathOperator{\Bern}{Bern}
\DeclareMathOperator{\Conv}{Conv}

\DeclareMathOperator*{\EE}{\mathbb E}
\DeclareMathOperator{\entropy}{H}

\DeclareMathOperator{\Id}{Id}

\DeclareMathOperator{\QMC}{QMC}

\DeclareMathOperator{\supp}{supp}
\DeclareMathOperator{\Tr}{Tr}
\DeclareMathOperator{\Var}{Var}

\begin{document}

\title{The Quantum Overlap Gap Property and Algorithmic Hardness for the Quantum Hypergraph Max-Cut Problem}

\author{Mikhail Mints}
\orcid{0009-0004-4508-353X}
\affiliation{%
  \institution{California Institute of Technology}
  \city{Pasadena}
  \country{USA}
}
\email{mmints@caltech.edu}

\author{Eric R.\ Anschuetz}
\orcid{0000-0002-9825-3692}
\affiliation{%
  \institution{California Institute of Technology}
  \city{Pasadena}
  \country{USA}
}
\email{eans@caltech.edu}

\begin{abstract}
   In this work, we analyze the average-case hardness of approximation for the Quantum Hypergraph Max-Cut problem using the theoretical framework of the Quantum Overlap Gap Property (QOGP). We establish two main results. Our first result applies to a wide class of \emph{stable} quantum algorithms, satisfying a Lipschitz property with respect to the quantum Wasserstein distance of order $2$. We show a \emph{weak} hardness result, demonstrating that for any Lipschitz constant $L$, there is some $k$ such that $L$-stable algorithms cannot approximate the optimal solution to Quantum Hypergraph Max-Cut on $k$-uniform hypergraphs in the average case. Additionally, we establish a \emph{strong} hardness result where $k$ is independent of $L$, but only for a more restricted class of \emph{local} quantum algorithms defined using the quantum Wasserstein distance of order $\infty$. We apply these results to establish concrete depth lower bounds for popular quantum algorithms for preparing near-optimal states for this problem.
\end{abstract}

\maketitle

\tableofcontents

\section{Introduction}

Understanding the computational hardness of approximation is a central goal in theoretical computer science. The PCP theorem and its many consequences provide a powerful framework for establishing such hardness results for \emph{worst-case} problem instances and show that, in this setting, it is computationally difficult to find solutions even moderately close to optimal for many optimization problems~\cite{Arora_Barak_2009}. However, worst-case complexity does not necessarily reflect the behavior of the instances encountered in practice. Many problems of practical and scientific interest are naturally modeled by random instances, motivating the study of \emph{average-case complexity}: how well can efficient algorithms approximate a typical instance drawn from a natural distribution?

A manifestation of average-case hardness of approximation is the existence of a \emph{statistical-to-computational gap}: a situation where the best known approximation ratio produced by any algorithm in the average case falls below the optimum value of a typical problem instance. For instance, consider the problem of finding the largest clique in an Erd\H os-R\'enyi graph on $n$ verties where each pair of vertices is connected with probability $1/2$. It is known that statistically, the size of the largest clique is $2 \log_2 n$ with high probability in the large $n$ limit---however, the best known algorithm can only find a clique of size $\log_2 n$ in the average case~\cite{karp:M581}.

The emergence of statistical-to-computational gaps can often be linked to the \emph{Overlap Gap Property} (OGP)~\cite{Gamarnik2021_OGP}, a clustering phenomenon arising in the solution landscapes of optimization problems. The OGP can be used to show that a certain class of algorithms known as \emph{stable algorithms} cannot obtain solutions close to optimal: near-optimal solutions form small clusters separated from each other by large distances, creating a topological barrier for algorithms whose outputs are Lipschitz functions of the problem instance. While not all computationally efficient algorithms exhibit this stability property, it is widely believed that the hardness results given by the OGP extend to all efficient classical algorithms~\cite{Gamarnik2021_OGP}.

One important class of optimization problems arises from the task of preparing near-ground states of quantum many-body systems. This problem is ubiquitous across quantum science, with applications ranging from drug discovery and quantum chemistry to condensed-matter physics and materials science~\cite{feynman1982,10.1002/qua.24811}. In this setting, however, the notion of average-case approximation is complicated by the fact that quantum algorithms produce quantum state witnesses to the optimization value rather than classical bit strings: there is no direct analogue of an algorithm's overlap landscape to which the classical OGP can be straightforwardly applied. Consequently, the OGP is not, by itself, the appropriate language for studying the hardness of approximation for quantum ground-state preparation. Instead, recent work~\cite{Anschuetz2026_glassiness} extends the underlying idea of topological obstructions to quantum algorithms and the quantum ground-state problem, providing a framework for establishing hardness results through the \emph{Quantum Overlap Gap Property} (QOGP).

\subsection{Main Results}

In this work, we use these recently-developed tools to understand the average-case hardness of approximation for a $k$-local generalization of the \emph{Quantum Max-Cut} problem, which is a widely-studied model for the ground state problem~\cite{gharibian_et_al,Hwang2022_unique_games_QMC,anshu_et_al:LIPIcs.TQC.2020.7,takahashi2026su2symmetricsemidefiniteprogramminghierarchy,King2023_QMC,apte202508395approximationalgorithmeprproblem}. Recent work has demonstrated improved approximation algorithms for the Quantum Max-Cut problem~\cite{King2023_QMC,apte202508395approximationalgorithmeprproblem}, and there are results demonstrating worst-case hardness assuming a widely believed complexity theoretic conjecture~\cite{Hwang2022_unique_games_QMC}, but not much is known about the problem's average-case complexity outside of heuristic physics arguments~\cite{Zlokapa2025_glassiness}. Here, we analyze for the first time this problem's average-case complexity using the mathematical framework of the Quantum Overlap Gap Property. We first prove a prelminary result in \Cref{sec:opt_value_bound} bounding the optimal value for the Quantum Max-Cut problem on Erd\H os-R\'enyi hypergraphs. We then establish a \emph{weak hardness} result in \Cref{chp:weak_hardness}, which can be informally stated as follows:

\begin{theorem}[Weak hardness for stable algorithms, informal version of \Cref{thm:weak_hardness}]\label{thm:weak_hardness_informal}
    For any approximation ratio $0 < \gamma \leq 1$ and any constant $L > 0$, there exists some integer $k$ such that no stable algorithm with Lipschitz constant $L$ can be $\gamma$-optimal for Quantum Hypergraph Max-Cut on $k$-uniform hypergraphs.
\end{theorem}

To do this, we adapt the QOGP framework of~\citet{Anschuetz2026_glassiness}. However, we can significantly simplify the proof by observing that we do not need an OGP in the typical sense to show a result of this form. Instead of constructing interpolation paths between many independent problem instances and claiming the existence of a particular structure somewhere along these paths, we can instead resample a very small fraction hyperedges for a problem instance $m$ times. Then, we can claim that the existence of a stable algorithm implies that a ``forbidden configuration'' emerges with high probability for large enough $k$. Crucially, however, this value of $k$ is dependent on $L$.

In \Cref{chp:strong_hardness}, we establish a stronger hardness result that holds for any Lipschitz constant $L$ as long as $k$ is equal to (or larger than) some universal constant $k^\ast$. The trade-off is that we require a stronger notion of stability than that considered in \Cref{thm:weak_hardness_informal}. Instead of \emph{stable} algorithms (which are defined using the quantum Wasserstein distance of order $2$), we look at \emph{local} (or ``$W_\infty$-stable'') algorithms (which we define using the quantum Wasserstein distance of order $\infty$). The result can be informally stated as follows:

\begin{theorem}[Strong hardness for local algorithms, informal version of \Cref{thm:strong_hardness}]\label{thm:strong_hardness_informal}
For any approximation ratio $0 < \gamma \leq 1$, there exists a constant $k$ such that no local algorithm (with an arbitrary Lipschitz constant) can be $\gamma$-optimal for Quantum Hypergraph Max-Cut on $k$-uniform hypergraphs.
\end{theorem}

Our proof is structured similarly to other $m$-OGP arguments, using Ramsey theory to demonstrate the existence of a collection of $m$ near-optimal solutions whose pairwise distances fall into some small interval, under the assumption that a local algorithm exists. 
We then demonstrate that for a suitable choice of parameters, this is a ``forbidden configuration'' for the near-optimal space of quantum states of Quantum Hypergraph Max-Cut, proving that no local algorithm is near-optimal for some fixed $k$. While the structures in the space of solutions are more complicated here than in \Cref{chp:weak_hardness}, the stronger metric we require for stability gives us more control over the structure of near-optimal solution configurations and makes it possible to prove a strong hardness claim.

In \Cref{chp:applications_future_directions}, we analyze the implications of these results for specific quantum optimization algorithms. We apply the weak hardness result to show that the Quantum Approximate Optimization Algorithm (QAOA)~\cite{Farhi2014_QAOA} with depth $C \log n$ with $C$ bounded by a universal constant fails to optimize the Quantum Hypergraph Max-Cut problem for sufficiently large values of $k$. Furthermore, we apply the strong hardness result to show that there exists a smaller value of $k$ for which this algorithm fails to optimize at any constant depth. Our results demonstrate that we can generalize known algorithmic hardness results regarding the QAOA~\cite{Farhi2020_QAOA_needs_to_see_the_whole_graph,Anshu2023_concentration_bounds_QAOA} to \emph{quantum} optimization problems. We then discuss the fundamental limitations of the Quantum Overlap Gap Property we are able to show for Quantum Hypergraph Max-Cut that make it difficult to prove a stronger hardness result that holds for all stable algorithms rather than just local ones. Finally, we discuss open problems and directions for future work.

\section{Preliminaries}\label{chp:prelim}

\subsection{The Overlap Gap Property}\label{sec:prelim_ogp}

We first review the classical overlap gap property, which characterizes the average-case hardness of combinatorial optimization problems~\cite{Gamarnik2021_OGP}. We begin with the ``vanilla'' OGP, which can be viewed as a formalization of shattering phase transitions in quantum many-body systems.

\begin{definition}
Consider an optimization problem of the form $\max_{\vec{s}} V(\bm{X}, \vec{s})$, where $\bm{X}$ is a problem instance from some probability distribution $\Lambda$, and $\vec{s} \in S$ for some solution space $S$ with a metric $d$. Let $c^\ast$ be the expected optimum value, such that with high probability over $\Lambda$, we have $\max_{\vec{s}} V(\bm{X}, \vec{s}) \geq c^\ast$. Then, if $0 < \gamma \leq 1$ and $0 \leq \nu_0 < \nu_1$, we say that the problem satisfies the Overlap Gap Property with parameters $\gamma, \nu_0, \nu_1$, if for any two solutions $\vec{s}_1, \vec{s}_2$ such that $V(\bm{X}, \vec{s}) \geq \gamma c^\ast$ and $V(\bm{X}, \vec{s}') \geq \gamma c^\ast$, we have that either $d(\vec{s}, \vec{s}') \leq \nu_1$ or $d(\vec{s}, \vec{s}') \geq \nu_2$ with high probability over the distribution $\Lambda$.
\end{definition}

\begin{figure}[ht]
    \centering
    \includegraphics[scale=0.75]{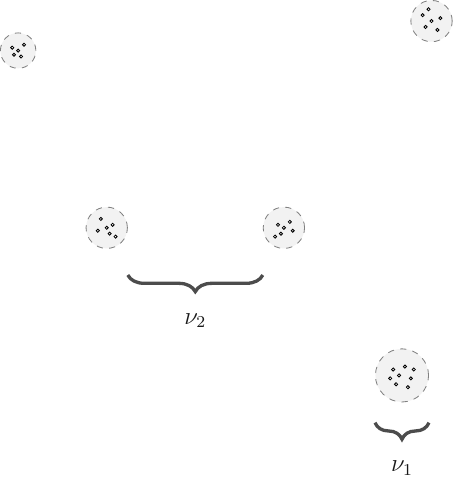}
    \caption[Illustration of the Overlap Gap Property]{Illustration of the Overlap Gap Property (OGP). The set of near-optimal solutions shatters into many disconnected clusters, where the diameter of each cluster is at most $\nu_1$ but the clusters are separated by a normalized Hamming distance of $\nu_2 > \nu_1$.}
    \Description{}
    \label{fig:ogp}
\end{figure}

In other words, for a problem satisfying the Overlap Gap Property, any two $\mu$-optimal solutions must be either closer to each other than $\nu_1$ or farther from each other than $\nu_2$ in normalized Hamming distance---the gap between $\nu_1$ and $\nu_2$ characterizes a clustering behavior in the solution space~\cite{Gamarnik2021_OGP}. Here, we will primarily be interested in proving hardness results using a stronger version of the OGP known as the \emph{ensemble OGP}:

\begin{definition}
We say that the optimization problem $\max_{\vec{s}} V(\bm{X}, \vec{s})$ satisfies the Ensemble Overlap Gap Property with parameters $\gamma, \nu_1, \nu_2$, with respect to a distribution $\Xi$ of pairs of problem instances if for $(\bm{X}, \bm{X}') \sim \Xi$, for any two solutions $\vec{s}, \vec{s}'$ such that $V(\bm{X}, \vec{s}) \geq \gamma E_\ast$ and $V(\bm{X}', \vec{s}') \geq \gamma E_\ast$, we have that either $d(\vec{s}, \vec{s}') \leq \nu_1$ or $d(\vec{s}, \vec{s}') \geq \nu_2$ with high probability over $\Xi$.
\end{definition}

To see how this can be used to establish hardness results, we consider the clique problem in the Erd\H os-R\'enyi graph $\GG(n, 1/2)$, where each of the $\binom{n}{2}$ possible edges is selected with probability $1/2$. Then, the sequence of graphs $G_1, \dots, G_{\binom{n}{2}}$ is constructed by first randomly sampling $G_1 \sim \GG(n, 1/2)$, and then, for each $1 \leq t \leq \binom{n}{2}$, resampling the $t$'th edge (using some arbitrary ordering of the edges). This sequence of graphs is an \emph{interpolation path} between two independent problem instances. Then, it can be shown~\cite{Gamarnik2021_OGP, Gamarnik2013_limits_of_local_alg} that for $1 + \frac{1}{\sqrt{2}} < \alpha < 2$, the number of pairs of cliques of size $\alpha \log_2 n$ in the problem instances $G_{t_1}, G_{t_2}$ such that the overlap of the two cliques is in some intermediate range $(\nu_1 (\alpha), \nu_2 (\alpha))$ decays exponentially with $n$. Furthermore, when considering the graphs $G_1, G_{\binom{n}{2}}$, which are completely independent from each other, the regime where the two cliques have large overlap becomes non-existent with high probability. Suppose we had a \emph{stable} algorithm $\cA$ whose output changed at most a constant number of vertices in the selected clique if a single edge is resampled, that produced cliques of size $\alpha \log_2 n$. Then, for some $t$ it must be the case that this algorithm produces a clique $\cA(G_t)$ whose overlap with $\cA(G_0)$ enters the ``forbidden'' region $(\nu_1 (\alpha), \nu_2 (\alpha))$, contradicting the ensemble OGP. Thus, no stable algorithm can produce cliques of size $\alpha \log_2 n$ for $\alpha > 1 + \frac{1}{\sqrt{2}}$. This can be further extended to $\alpha > 1$ using a generalization of the ensemble OGP known as the \emph{multi-OGP}~\cite{10.1214/16-AOP1114}. This example illustrates how the OGP creates a topological barrier for optimization algorithms. While the existence of such an obstruction does not rule out \emph{all} possible efficient algorithms, it can often be argued that a non-stable algorithm must somehow make use of the structure of the problem in a fundamentally new way in order to be able to land in the isolated clusters of near-optimal solutions. Establishing some version of an OGP for a problem when studying average-case complexity is thus analogous to establishing NP-hardness when studying worst-case complexity~\cite{Gamarnik2021_OGP}.

\subsection{Problem Setup}\label{sec:problem_setup}

The \emph{Quantum Max-Cut Hamiltonian} on a graph $G = (V, E)$ with $V = [n] = \{1, \dots, n\}$ with edge weights $w_{u,v} \in\left[-1,1\right]$ is the Hermitian operator acting on $\mathbb{C}^{2^n}$ given by~\cite{gharibian_et_al}:
\begin{equation}
\begin{aligned}
\bm{H}_2^{(n)}=& -\frac{1}{\sqrt{n}} \sum_{\{u, v\} \in E} w_{u,v} \ket{\Psi^-}\bra{\Psi^-}_{\{u,v\}} \otimes \Id_{[n] \setminus \{u, v\}} = \\
={} & -\frac{1}{4 \sqrt{n}} \sum_{\{u, v\} \in E} w_{u,v} \parens{\Id - \sigmax \otimes \sigmax - \sigmay \otimes \sigmay - \sigmaz \otimes \sigmaz}_{\{u, v\}} \otimes \Id_{[n] \setminus \{u, v\}},
\end{aligned}
\end{equation}
where the $\bm{\sigma}^{X,Y,Z}$ denote the Pauli $X,Y,Z$ operators, respectively, acting on $\mathbb{C}^2$. The Quantum Max-Cut problem is to find a quantum state (i.e., a vector $\ket{\psi}\in\mathbb{C}^{2^n}$) approximately maximizing:
\begin{equation}
    f\left(\ket{\psi}\right):=\bra{\psi}\bm{H}_2\ket{\psi},
\end{equation}
where $\bra{\psi}:=\ket{\psi}^\dagger$ denotes the Hermitian conjugate of $\ket{\psi}$.

We are here interested in ``typical'' instances of Quantum Max-Cut. 
If we randomly sample the interaction graph from an Erd\H os-R\'enyi distribution with expected degree $p$ for each vertex, and we ignore the identity term and the $1/4$ factor for convenience (which just contribute a constant shift and scale), we can write the Quantum Max-Cut Hamiltonian on $n$ vertices as:
\begin{equation}
\begin{aligned}
& \bm{H}_2^{(n)} (\bm{X}) \coloneq \frac{1}{n} \sum_{A \in \binom{[n]}{2}} X_A \parens{\sigmax \otimes \sigmax + \sigmay \otimes \sigmay + \sigmaz \otimes \sigmaz}_A \otimes \Id_{[n] \setminus A},
\end{aligned}
\end{equation}
where $\binom{[n]}{2}$ is the set of all subsets of $[n]$ of size $2$ and $\bm{X} = \{X_A\}_{A \in \binom{[n]}{2}}$ where $X_A = S_A J_A$, where the $S_A$ are i.i.d.\ Bernoulli with probability $p / (n - 1)$ and the $J_A$ are i.i.d.\ Rademacher ($1$ or $-1$ with equal probability). In this work, we consider the natural generalization of the Quantum Max-Cut Hamiltonian to hypergraphs. Specifically, we will consider the following model:

\begin{definition}[Quantum Hypergraph Max-Cut]\label{def:qmc}
We extend the Quantum Max-Cut Hamiltonian to sparse $k$-uniform hypergraphs in the following natural way:
\begin{equation}
\bm{H}^{(n, k)}(\bm{X}) \coloneq \frac{\sqrt{k}}{n} \sum_{A \in \binom{[n]}{k}} X_A \parens{\sigmax_A + \sigmay_A + \sigmaz_A} \otimes \Id_{[n] \setminus A},
\end{equation}
for $\bm{X} = \braces{X_A}_{A \in \binom{[n]}{k}}$. We define the \emph{Quantum Hypergraph Max-Cut model} as $\bm{H}^{(n,k)} (\bm{X})$ where $\bm{X} \sim \QMC(n, k, p)$, that is, each $X_A = S_A J_A$ where $S_A \sim \Bern\parens{p / \binom{n - 1}{k - 1}}$ and $J_A \sim \{-1, 1\}$ i.i.d.
\end{definition}

Note that it is possible to generalize other problems closely related to Quantum Max-Cut, such as the EPR Hamiltonian~\cite{King2023_QMC} to hypergraphs in a similar way. In \Cref{app:ghz}, we consider this EPR Hypergraph Hamiltonian, and give evidence that it does \emph{not} exhibit the QOGP like Quantum Hypergraph Max-Cut does.

\subsection{Maximal Energy Lower Bound}\label{sec:opt_value_bound}

Before proceeding by defining the Quantum Overlap Gap Property, we take the opportunity to bound the optimal value of the Quantum Hypergraph Max-Cut model.
\begin{definition}
The \emph{limiting maximal energy} for the Quantum Hypergraph Max-Cut model is given by
\begin{equation}
E_\ast^{(k, p)} \coloneq \limsup_{n \rarr \infty} E_\ast^{(n, k, p)} \coloneq \limsup_{n \rarr \infty} \EE_{\bm{X} \sim \QMC(n, k, p)} \brackets{\opnorm{\bm{H}^{(n,k)}(\bm{X})}},
\end{equation}
where
\begin{equation}
\opnorm{\bm{H}} = \sup_{\ket{\psi} \in \cS_n} \bra{\psi} \bm{H} \ket{\psi}.
\end{equation}
\end{definition}
While we are unable to compute $E_\ast^{(k,p)}$ directly, we are able to lower bound it by computing the maximum value over computational basis states; we will later see that this suffices for our purposes.
\begin{definition}
We define the energy achieved by a computational basis state $\vec{s}$ as:
\begin{equation}
V^{(n,k)} (\bm{X}, \vec{s}) \coloneq \bra{\vec{s}} \bm{H}^{(n,k)} (\bm{X}) \ket{\vec{s}}.
\end{equation}
Observe that:
\begin{equation}
\begin{aligned}
& V^{(n,k)} (\bm{X}, \vec{s}) =
\frac{\sqrt{k}}{n} \sum_{A \in \binom{[n]}{k}} X_A (-1)^{\sum_{v \in A} s_v}.
\end{aligned}
\end{equation}
This problem thus becomes equivalent to MAX-$k$-XORSAT~\cite{Marwaha2022_max_k_xorsat,Anschuetz2025_DQI}, with a penalty applied for each violated clause.
\end{definition}

We now bound the optimal value of the Quantum Hypergraph Max-Cut model.
\begin{theorem}\label{thm:lower_bound}
\begin{equation}
\Pr_{\bm{X} \sim \QMC(n,k,p)} \brackets{\max_{\vec{s}} V (\bm{X}, \vec{s}) < \alpha_L^{(k,p)}} = o_n (1),
\end{equation}
where
\begin{equation}
\alpha_L^{(k, p)} \coloneq \inf_{w \in (1/2, 1]} \sqrt{p \parens{\log 2 - \entropy(w)}\parens{1 + (2w - 1)^{-k}}},
\end{equation}
and
\begin{equation}
\lim_{k \rarr \infty} \alpha_L^{(k,p)} = \sqrt{2 p \log 2}.
\end{equation}
Consequently,
\begin{equation}
    E_\ast^{(k, p)} \geq \alpha_L^{(k,p)}.
\end{equation}
\end{theorem}

\begin{proof}
By~\citet[Theorem 6]{Anschuetz2025_bounds}, we may replace the distribution of the coefficients $X_A$ in the Hamiltonian with i.i.d. Gaussian random variables that have a mean of $0$ and the same variance of $\frac{p}{\binom{n - 1}{k - 1}}$, without changing the limiting maximal energy. We therefore assume this distribution on the coefficients WLOG in what follows.

Fix some $\alpha > 0$. Let $Y_{\vec{s}} = \mathds{1}\brackets{V\parens{\bm{X}, \vec{s}} \geq \alpha}$. Let $Y = \sum_{\vec{s} \in \{0,1\}^n} Y_{\vec{s}}$, that is, the number of configurations with energy at least $\alpha$. By the Paley-Zygmund inequality, we have that
\begin{equation}
\Pr_{\bm{X} \sim \QMC(n, k, p)} \brackets{\max_{\vec{s}} V(\bm{X}, \vec{s}) \geq \alpha} = \Pr[Y > 0] \geq
\frac{\EE[Y]^2}{\EE[Y^2]} =
\frac{\sum_{\vec{s},\vec{s}'} \EE[Y_{\vec{s}}] \EE[Y_{\vec{s}'}]}{\sum_{\vec{s},\vec{s}'} \EE[Y_{\vec{s}} Y_{\vec{s}'}]}.
\end{equation}

Now, for a given $\vec{s}, \vec{s}' \in \{0, 1\}^n$, write $w_{s,s'} \coloneq \frac{1}{n} \abs{\{v \in [n] : \vec{s}_v = s, \vec{s}'_v = s'\}}$ for $s, s' \in \{0, 1\}$. We have that
\begin{equation}
\begin{aligned}
& \sigma^2 \coloneq \EE_{\bm{X} \sim \QMC(n, k, p)} \brackets{V(\bm{X}, \vec{s})^2} =
\frac{k}{n^2} \sum_{A, A' \in \binom{[n]}{k}} \EE\brackets{X_A X_{A'}} (-1)^{\sum_{v \in A} \bm{s}_v + \sum_{v' \in A'} \bm{s}_{v'}} = \\
={} & \frac{k}{n^2} \sum_{A \in \binom{[n]}{k}} \EE\brackets{X_A^2} =
\frac{k}{n^2}\binom{n}{k}\frac{p}{\binom{n - 1}{k - 1}} =
\frac{p}{n^2} \frac{n!}{k! (n - k)!} \frac{k (k - 1)! (n - k)!}{(n - 1)!} = \frac{p}{n}.
\end{aligned}
\end{equation}
We also have that
\begin{equation}
\begin{aligned}
& \EE_{\bm{X} \sim \QMC(n, k, p)} \brackets{V(\bm{X}, \vec{s}) V(\bm{X}, \vec{s})} = \\
={} & \frac{k}{n^2} \sum_{A, A' \in \binom{[n]}{k}} \EE\brackets{X_A X_{A'}} (-1)^{\sum_{v \in A} \bm{s}_v + \sum_{v' \in A'} \bm{s}'_{v'}} = \\
={} & \frac{k}{n^2} \frac{p}{\binom{n - 1}{k - 1}} \sum_{A \in \binom{[n]}{k}} (-1)^{\sum_{v \in A} s_v + s'_v} = \\
={} & \frac{k}{n^2} \frac{p}{\binom{n - 1}{k - 1}} \sum_{a=0}^k \sum_{b=0}^{k-a} \sum_{b'=0}^{k-a-b} \binom{w_{0,0} n}{a} \binom{w_{0,1} n}{b} \binom{w_{1,0} n}{b'} \binom{w_{1,1} n}{k - a - b - b'} (-1)^{b + b'} = \\
={} & \frac{n^k k}{n^2} \frac{p (k - 1)!}{n^{k-1}} \frac{1}{k!} \sum_{a=0}^k \sum_{b=0}^{k-a} \sum_{b'=0}^{k-a-b} \\
& \hspace{8em} \frac{k! (1 + o_n (1))}{a! b! b'! (k - a - b - b')!} w_{0,0}^a (-w_{0,1})^b (-w_{1,0})^{b'} w_{1,1}^{k-a-b-b'} = \\
={} & \frac{p}{n} \parens{w_{0,0} + w_{1,1} - w_{0,1} - w_{1,0}}^k (1 + o_n (1)).
\end{aligned}
\end{equation}
So now, we have the correlation coefficient between $V\parens{\bm{X}, \vec{s}}$ and $V\parens{\bm{X}, \vec{s}'}$ (up to a factor of $1 +  o_n(1)$):
\begin{equation}
\begin{aligned}
\rho ={} & \frac{\EE_{\bm{X} \sim \QMC(n, k, p)}[V(\bm{X}, \vec{s}) V(\bm{X}, \vec{s}')]}{\sqrt{\EE_{\bm{X} \sim \QMC(n, k, p)} [V(\bm{X}, \vec{s})^2] \EE_{\bm{X} \sim \QMC(n, k, p)} [V(\bm{X}, \vec{s}')^2]}} = \\
={} & \parens{w_{0,0} + w_{1,1} - w_{0,1} - w_{1,0}}^k.
\end{aligned}
\end{equation}

Using a bivariate Gaussian tail bound, we can then establish that
\begin{equation}
\begin{aligned}
& \frac{\EE[Y_{\vec{s}} Y_{\vec{s}'}]}{\EE[Y_{\vec{s}}] \EE[Y_{\vec{s}'}]} \leq \\
\leq{} & \frac{
\frac{(1 + \rho)^2 \sigma^2}{2 \pi \alpha^2 \sqrt{1 - \rho^2}} \exp\parens{-\frac{\alpha^2}{(1 + \rho)\sigma^2}}
}{
\frac{\sigma^2}{2\pi \alpha^2} \exp\parens{-\frac{\alpha^2}{\sigma^2}} (1 + o_n(1))
} \leq \\
\leq{} & \frac{(1 + \rho)^2}{\sqrt{1 - \rho^2}} \exp\parens{\frac{\alpha^2}{\sigma^2}\parens{1 - \frac{1}{1 + \rho}}} = \\
={} & \frac{(1 + \rho)^2}{\sqrt{1 - \rho^2}} \exp\parens{\frac{\alpha^2 \rho}{\sigma^2 (1 + \rho)}} = \\
={} & \exp\parens{\frac{\frac{n}{p} \alpha^2 (w_{0,0} + w_{1,1} - w_{0,1} - w_{1,0})^k}{1 + (w_{0,0} + w_{1,1} - w_{0,1} - w_{1,0})^k} (1 + o_n(1))}.
\end{aligned}
\end{equation}

So now, we can write our lower bound in the limit of large $n$, ignoring $o_n (1)$ terms:
\begin{equation}
\begin{aligned}
& \Pr_{\bm{X} \sim \QMC(n,k,p)}[\max_{\vec{s}} V(\bm{X}, \vec{s}) \geq \alpha] \geq \\
\geq{} & \frac{\sum_{w_{0,0}, w_{0,1}, w_{1,0}, w_{1,1}} \binom{n}{w_{0,0} n} \binom{n}{w_{0,1} n} \binom{n}{w_{1,0} n} \binom{n}{w_{1,1} n}}{\sum_{w_{0,0}, w_{0,1}, w_{1,0}, w_{1,1}} \binom{n}{w_{0,0} n} \binom{n}{w_{0,1} n} \binom{n}{w_{1,0} n} \binom{n}{w_{1,1} n} \exp\parens{\frac{n \alpha^2 (w_{0,0} + w_{1,1} - w_{0,1} - w_{1,0})^k}{1 + (w_{0,0} + w_{1,1} - w_{0,1} - w_{1,0})^k}}} \approx \\
\approx{} & \frac{\sup_{w_{0,0}, w_{0,1}, w_{1,0}, w_{1,1}} \exp\parens{-n \sum_{s,s'} w_{s,s'} \log w_{s,s'}}}
{\sup_{w_{0,0}, w_{0,1}, w_{1,0}, w_{1,1}} \exp\parens{-n \sum_{s,s'} w_{s,s'} \log w_{s,s'} + \frac{\frac{n}{p} \alpha^2 (w_{0,0} + w_{1,1} - w_{0,1} - w_{1,0})^k}{1 + (w_{0,0} + w_{1,1} - w_{0,1} - w_{1,0})^k}}}.
\end{aligned}
\end{equation}

The numerator is maximized when all $w_{s,s'}$ are equal, giving a value of \\ $\exp(-n \log (1/4)) = \exp(2 n \log 2)$. Now, for the denominator, observe that it is maximized when $w_{0,0} = w_{1,1} \eqcolon w / 2$ and $w_{0,1} = w_{1,0} \eqcolon (1 - w) / 2$. So, we can rewrite it as
\begin{equation}
\sup_{w \in [0, 1]} \exp\brackets{n \parens{\log 2 - w \log w - (1 - w) \log (1 - w) + \frac{\frac{\alpha^2}{p} (2w - 1)^k}{1 + (2w - 1)^k}}},
\end{equation}
so
\begin{equation}
\begin{aligned}
& \Pr_{\bm{X} \sim \QMC(n, k, p)}[\max_{\vec{s}} V(\bm{X}, \vec{s}) \geq \alpha] \geq \\
\geq{} & \inf_{w \in [0, 1]} \exp\brackets{n \parens{\log 2 + w \log w + (1 - w) \log (1 - w) - \frac{\frac{\alpha^2}{p} (2w - 1)^k}{1 + (2w - 1)^k}}}.
\end{aligned}
\end{equation}
We want to find $\alpha$ such that the probability does not decay with $n$: so, since the numerator is $\exp(2 n \log 2)$, we want to satisfy the following for all $w \in [0, 1]$:
\begin{equation}
\begin{aligned}
& \frac{\frac{\alpha^2}{p} (2w - 1)^k}{1 + (2w - 1)^k} \leq \log 2 + w \log w + (1 - w) \log (1 - w) = \log 2 - \entropy(w).
\end{aligned}
\end{equation}
Note that if $k$ is odd and $w < 1/2$, then $-1 < 2w - 1 < 0$, so $-1 < (2w - 1)^k < 0$, so the left-hand side is negative and the inequality is satisfied trivially (since the right-hand-side is always non-negative). Furthermore, if $w = 1/2$, then both the left-hand-side and right-hand-side are equal to $0$. Also, if $k$ is even, both sides remain the same when replacing $w$ with $1 - w$. Therefore, it is sufficient to consider $w \in (1/2, 1]$. Since we have
\begin{equation}
\parens{\alpha_L^{(k,p)}}^2 \coloneq \inf_{w \in (1/2, 1]} p \parens{\log 2 - \entropy(w)}\parens{1 + (2w - 1)^{-k}},
\end{equation}
we have shown that
\begin{equation}
\Pr_{\bm{X} \sim \QMC(n, k, p)} \brackets{\max_{\vec{s}} V(\bm{X}, \vec{s}) < \alpha_L^{(k,p)}} = o_n (1).
\end{equation}
In the limit of large $k$, the value is minimized when $w$ is close to $0$ or $1$, which makes the $w \log w$ and $(1 - w) \log (1 - w)$ terms approach $0$ and the $(2 w - 1)^{-k}$ term approach $1$, which means that
\begin{equation}
\lim_{k \rarr \infty} \alpha_L^{(k,p)} = \sqrt{2 p \log 2}.
\end{equation}
\end{proof}

\subsection{The QOGP and Classical Shadows Estimators}\label{sec:prelim_shadows}

It is not immediately clear how the idea of the OGP can be extended to the quantum setting in a meaningful way that allows us to establish hardness results for quantum algorithms solving the quantum ground state problem. We can view a quantum algorithm as a function sending each choice of parameters $\bm{X}$ to a density matrix $\bm{\rho}(\bm{X})$, i.e., a positive semidefinite, trace-$1$ operator acting on $\mathbb{C}^{2^n}$. Now, the goal is to argue that a certain class of \emph{stable} quantum algorithms fails to find near-optimal solutions for this optimization problem. To properly define stability, we need to use the right distance metric on the space of density matrices. Observe that if we just use, for example, trace distance, then if we take $x \in \{0, 1\}^n$ and $\bm{\rho}(x) = \ket{x}\bra{x}$, then the distance between $\bm{\rho}(x)$ and $\bm{\rho}(y)$ will be the same for any $x \neq y$, making it difficult to define a natural notion of stability. To define a notion of stable quantum algorithms that is not too restrictive,~\citet{Anschuetz2026_glassiness} uses the quantum Wasserstein distance:

\begin{definition}\label{def:quantum_wasserstein}
The \emph{quantum Wasserstein $F$-norm of order $\alpha$} for a Hermitian, traceless observable $\bm{X}$ on $n$ qubits is
\begin{equation}
\norm{\bm{A}}_{W_\alpha} \coloneq \min_{\{\bm{A}_i\}_{i=1}^n \in \cB(\bm{A})} \sum_{i=1}^n \norm{\frac{1}{2} \bm{A}_i}_\ast^{\frac{1}{\alpha}},
\end{equation}
where $\norm{\cdot}_\ast$ is the trace norm and
\begin{equation}
\cB(\bm{A}) = \braces{\{\bm{A}_i\}_{i=1}^n : \sum_{i=1}^n \bm{A}_i = \bm{X} \wedge \bm{A}_i = \bm{A}_i^\dagger \wedge \Tr_{\{i\}} \bm{A}_i = 0}.
\end{equation}
Then, the \emph{quantum Wasserstein distance of order $p$} between density matrices $\bm{\rho}$ and $\bm{\sigma}$ is $\norm{\bm{\rho} - \bm{\sigma}}_{W_\alpha}$.

We also define the \emph{Wasserstein $\infty$-norm} as:
\begin{equation}
\norm{\bm{A}}_{W_\infty} \coloneq \lim_{\alpha \rarr \infty} \norm{\bm{A}}_{W_\alpha} = \min_{\{\bm{A}_i\}_{i=1}^n \in \cB(\bm{A})} \abs{\braces{i \in [n] : \norm{\bm{A}_i}_\ast \neq 0}}.
\end{equation}
\end{definition}

The quantum Wasserstein distance is a generalization of the Hamming distance: a quantum channel acting on $\ell$ qubits can change the quantum Wasserstein distance by at most $O(\ell)$~\cite{Anschuetz2026_glassiness}. Following this intuition, we define \emph{stability} for quantum algorithms as follows:

\begin{definition}[Stability, informal version of \Cref{def:stable_alg}]
A quantum algorithm $\bm{\rho}$ is \emph{stable} with parameters $f, L$ if
\begin{equation}
\norm{\bm{\rho}(\bm{X}) - \bm{\rho}(\bm{Y})}_{W_2} \leq f + L \norm{\bm{X} - \bm{Y}}_1
\end{equation}
with high probability over inputs $\bm{X},\bm{Y}$.
\end{definition}

We also consider a stronger version of stability that we call \emph{locality}, where we use the $W_\infty$ distance instead of $W_2$:

\begin{definition}[Locality, informal version of \Cref{def:local_alg}]
A quantum algorithm $\bm{\rho}$ is \emph{stable} with parameters $f, L$ if
\begin{equation}
\norm{\bm{\rho}(\bm{X}) - \bm{\rho}(\bm{Y})}_{W_\infty} \leq f + L \norm{\bm{X} - \bm{Y}}_1.
\end{equation}
\end{definition}

Note that in the case of deterministic classical algorithms, the notions of stability and locality are equivalent since the quantum $W_2$ and $W_\infty$ distances both reduce to the Hamming distance on computational basis states. But in general, the $W_\infty$ distance is at least the $W_2$ distance. For this reason, the notion of locality is stronger than that of stability for quantum algorithms and will allow us to prove a stronger hardness result.

The main idea of behind the Quantum Overlap Gap Property (QOGP) is that if a quantum optimization problem is converted to a classical one using a \emph{classical shadows estimator}~\cite{Huang2020_shadows} and this problem exhibits a certain form of the OGP, then it is possible to show that the quantum problem is hard for quantum algorithms that are stable (or local)~\cite{Anschuetz2026_glassiness}. In what follows, we use $\Conv(\cS)$ to denote the convex hull of $\cS$ and ``quantum channel'' to mean a completely positive trace-preserving linear map.

\begin{definition}
Let $\cB_n$ be the set of classical representations of Pauli basis states on $n$ qubits, i.e., the eigenstates of Pauli operators on $n$ qubits. We can label these by pairs $(\vec{b},\vec{s})$ where $\vec{b} \in \{1, 2, 3\}^n$ and $\vec{s} \in \{0, 1\}^n$. We can label Pauli basis states as $\ket{\vec{b}; \vec{s}}$ where for all $v \in [n]$ we have that $b_v \in \{1, 2, 3\}$ corresponds to $X, Y, Z$ respectively, and $s_v \in \{0, 1\}$ to label the $+1$ and $-1$ eigenstates. We also write $\ket{b; \vec{s}}$ if $b_v = b$ for all $v$.
\end{definition}

\begin{definition}

We say that a class of random Hamiltonians $\cH$ has a $(\delta, \probest, \probb)$-efficient local shadows estimator if there is a quantum channel $\bm{\cM} : \cS_n^\mathrm{m} \rarr \Conv(\cB)$ where $\cS_n^\mathrm{m}$ is the set of mixed states on $n$ qubits and $\cB \subseteq \cB_n$, and a linear function $\bm{\cR}$ such that:
\begin{enumerate}
    \item $\bm{\cM}$ is a convex combination of tensor product channels,
    \item For any $\bm{H}(\bm{X})$ we have
    \begin{equation}
        \Tr\parens{\bm{\cR}(\bm{H}(\bm{X})) \bm{M}(\bm{\rho})} = \Tr\parens{\bm{H}(\bm{X}) \bm{\rho}},
    \end{equation}
    \item Let $\widetilde{\bm{\cM}} : \cS_n^\mathrm{m} \times [0, 1] \rarr \cB$ be an associated pure quantum channel of $\bm{\cM}$. That is, the image of $\widetilde{\bm{\cM}}$ only consists of pure states in $\cB$, and
    \begin{equation}
        \EE_{\omega \sim \cU} \brackets{\widetilde{\bm{\cM}}(\bm{\rho}, \omega)} = \bm{\cM}(\bm{\rho}).
    \end{equation}
    Then with probability at least $1 - \probb$, we have
    \begin{equation}
        \Pr_{\omega \sim \cU} \brackets{\Tr\parens{\bm{\cR}(\bm{H}(\bm{X})) \widetilde{\bm{\cM}}(\bm{\rho})} - \Tr\parens{\bm{H}(\bm{X}) \bm{\rho}} \geq -\delta E_\ast} \geq 1 - \probest,
    \end{equation}
    for all $\bm{\rho} \in \mixedn$, where $\cU$ is the uniform distribution on $[0, 1]$ and $E_\ast$ is the limiting maximal energy.
\end{enumerate}
In other words, we can represent any quantum state $\rho$ with a \emph{classical shadows representation} consisting of $n$-qubit Pauli basis states, and this representation can be used in a linear estimator of the ground state energy with bounded error probability.
\end{definition}

In the context of our problem, we are concerned with the case where the operators $\bm{H}_i$ in our Hamiltonian are $k$-local Pauli operators. In this case, we can use the \emph{Pauli shadows estimator} as our local shadows estimator.

\begin{definition}
We define the Pauli shadows estimator (adapted to the Quantum Max-Cut Hamiltonian) ~\cite{Huang2020_shadows} by
\begin{equation}
\bm{\cM}(\bm{\rho}) \coloneq \frac{1}{3} \sum_{b \in \{1,2,3\}} \sum_{\vec{s} \in \{0,1\}^n} \bra{b; \vec{s}} \bm{\rho} \ket{b; \vec{s}} \cdot \ket{b; \vec{s}} \bra{b; \vec{s}},
\end{equation}
\end{definition}

That is, given a quantum state $\bm{\rho}$, we choose one basis out of $X$, $Y$, or $Z$ uniformly at random, and measure $\bm{\rho}$ in that basis, obtaining a state $\ket{b; \vec{s}} \bra{b; \vec{s}}$. Then, $\bm{\cM}(\bm{\rho})$ is the expectation of this state, and we let the classical shadow be $\hat{\bm{\rho}} = \bm{\cM}^{-1} (\ket{b; \vec{s}}\bra{b; \vec{s}})$. Then, taking the expectation over this procedure, we have that since $\bm{\cM}$ is linear,
\begin{equation}
\EE\brackets{\hat{\bm{\rho}}} = \frac{1}{3} \sum_{b \in \{1,2,3\}} \sum_{\vec{s} \in \{0,1\}^n} \bra{b; \vec{s}} \bm{\rho} \ket{b; \vec{s}} \cdot \bm{\cM}^{-1} \parens{\ket{b; \vec{s}} \bra{b; \vec{s}}} = \bm{\cM}^{-1} \parens{\bm{\cM} (\bm{\rho})} = \bm{\rho}.
\end{equation}

\begin{lemma}
For a $k$-local Pauli of the form $\bm{H} = \bm{P}^{\otimes k} \otimes \Id^{\otimes (n - k)}$, where $\bm{P}$ is one of $X, Y, Z$, we have that
\begin{equation}
\bm{\cM}^{-1}(\bm{H}) = 3 \bm{H}.
\end{equation}
\end{lemma}

\begin{proof}
If $b_\ast$ is $1$, $2$, or $3$ depending on whether $\bm{P}$ is $X$, $Y$, or $Z$, then
\begin{equation}
\begin{aligned}
& \bm{\cM}(3 \bm{H}) =
\sum_{b \in \{1,2,3\}} \sum_{\vec{s} \in \{0, 1\}^n} \bra{b; \vec{s}} \bm{H} \ket{b; \vec{s}} \cdot \ket{b; \vec{s}}\bra{b; \vec{s}} = \\
={} & \sum_{\vec{s} \in \{0, 1\}^n} \bra{b_\ast; \vec{s}} \bm{H} \ket{b_\ast; \vec{s}} \cdot \ket{b_\ast; \vec{s}}\bra{b_\ast; \vec{s}} = \bm{H}.
\end{aligned}
\end{equation}
\end{proof}

Now, we have the following result establishing that the Pauli shadows estimator is a valid local shadows estimator with constant $p_{\mathrm{est}}$ for local Hamiltonians, due to~\citet{Huang2020_shadows}:

\begin{lemma}\label{lem:pauli_shadows_variance}
For a fixed Hamiltonian $\bm{H} = \bm{H}(\bm{X})$, given a state $\bm{\rho}$, if we randomly sample a classical shadow $\ket{b; \vec{s}} \bra{b; \vec{s}}$ of $\bm{\rho}$, we have that
\begin{equation}
\Var \brackets{\Tr\parens{\hat{\bm{\rho}} \bm{H}}} \leq 9 \opnorm{\bm{H}}^2.
\end{equation}
\end{lemma}

\begin{proof}
We have that
\begin{equation}
\begin{aligned}
& \Var\brackets{\Tr\parens{\bm{H} \hat{\bm{\rho}}}} =
\EE\brackets{\Tr\parens{\bm{H} \hat{\bm{\rho}}}^2} - \EE\brackets{\Tr\parens{\bm{H}\hat{\bm{\rho}}}}^2 \leq \\
\leq{} & \EE\brackets{\Tr\parens{\bm{H} \hat{\bm{\rho}}}^2} =
\EE\brackets{\Tr\parens{\bm{\cM}^{-1}\parens{\bm{H} \ket{b;\vec{s}}\bra{b;\vec{s}}}}^2}.
\end{aligned}
\end{equation}
Now, observe that
\begin{equation}
\begin{aligned}
& \Tr\parens{\bm{\cM}^{-1}\parens{\ket{b;\vec{s}}\bra{b;\vec{s}}} \bm{H}} =
\Tr\parens{\bm{\cM}^{-1}\parens{\ket{b;\vec{s}}\bra{b;\vec{s}}} \bm{\cM}(\bm{\cM}^{-1}(\bm{H}))} = \\
={} & 3^{-n} \sum_{b \in \{1,2,3\}^n} \sum_{\vec{s} \in \{0,1\}^n} \bra{b; \vec{s}} \bm{\cM}^{-1}\parens{\ket{b;\vec{s}}\bra{b;\vec{s}}} \ket{b; \vec{s}} \cdot \bra{b; \vec{s}} \bm{\cM}^{-1} (\bm{H}) \ket{b; \vec{s}} = \\
={} & \Tr\parens{\parens{\ket{b;\vec{s}}\bra{b;\vec{s}}} \bm{\cM}^{-1}(\bm{H})} = \bra{b;\vec{s}} \bm{\cM}^{-1}(\bm{H}) \ket{b;\vec{s}}.
\end{aligned}
\end{equation}
So now, since $\bm{H}$ is a linear combination of $k$-local Pauli operators that each only have one of $\sigmax$, $\sigmay$, or $\sigmaz$, we have that
\begin{equation}
\begin{aligned}
\Var\brackets{\Tr\parens{\bm{H}\hat{\bm{\rho}}}} \leq{} & \EE\brackets{\parens{\bra{b;\vec{s}} \bm{\cM}^{-1}(\bm{H}) \ket{b;\vec{s}}}^2} = \\
={} & 9 \EE\brackets{\parens{\bra{b;\vec{s}} \bm{H} \ket{b;\vec{s}}^2}} \leq 9 \opnorm{\bm{H}}^2.
\end{aligned}
\end{equation}
\end{proof}

To see why this result implies $p_{\mathrm{est}}$ is constant for the Quantum Hypergraph Max-Cut model, we use the self-averaging of the model and Cantelli's inequality applied to \Cref{lem:pauli_shadows_variance}.
\begin{lemma}\label{lem:self_averaging}
The Quantum Hypergraph Max-Cut model exhibits \emph{self-averaging}, that is,
\begin{equation}
\Pr_{\bm{X} \sim \QMC(n,k,p)} \brackets{\abs{\opnorm{\bm{H}(\bm{X})} - E_\ast^{(k,p)} \geq t}} \leq e^{-\Omega(n)},
\end{equation}
\end{lemma}

\begin{proof}
The proof proceeds similarly to Proposition 58 of~\citet{Anschuetz2025_DQI}.
\end{proof}

Then, we can show that the following result holds:

\begin{lemma}\label{lem:efficient_shadows}
For both the Quantum $k$-Heisenberg Model and the Quantum Hypergraph Max-Cut Model, for any $\delta > 0$, the Pauli shadows estimator is a $(\delta, \probest, \probb)$-efficient local shadows estimator, where
\begin{equation}
\probest = \frac{1}{1 + 0.99 \cdot \delta^2 / 9},
\end{equation}
and $\probb = e^{-\Omega(n)}$.
\end{lemma}

\begin{proof}
For any $t > 0$, by applying \Cref{lem:pauli_shadows_variance} and \Cref{lem:self_averaging}, we have that
\begin{equation}
\Pr_{\bm{X}} \brackets{\sup_{\bm{\sigma} \in \mixedn} \Var\brackets{\Tr\parens{\bm{H}(\bm{X}) \hat{\bm{\sigma}}}} \geq 3 E_\ast + t} \leq \Pr_{\bm{X}} \brackets{\opnorm{\bm{H}(\bm{X})} \geq E_\ast + t / 3} = e^{-\Omega(n)}.
\end{equation}

Now, let
\begin{equation}
\widetilde{\bm{\cM}} : \mixedn \times [0, 1] \rarr \braces{\ket{b; \vec{s}}\bra{b; \vec{s}} : b \in \{1,2,3\}, \vec{s} \in \{0,1\}^n}
\end{equation}
be the associated pure algorithm for the quantum channel $\bm{\cM}$, so that
\begin{equation}
\EE_{\omega \sim \cU} \widetilde{\bm{\cM}}(\bm{\rho}, \omega) = \bm{\cM}(\bm{\rho}),
\end{equation}
and for $\hat{\bm{\rho}} = \bm{\cM}^{-1}\parens{\widetilde{\bm{\cM}}(\bm{\rho}, \omega)}$, we have
\begin{equation}
\Tr\parens{3 \bm{H}(\bm{X}) \widetilde{\bm{\cM}}(\bm{\rho}, \omega)} = \Tr\parens{\bm{H}(\bm{X}) \hat{\bm{\rho}}}.
\end{equation}
Now, take any $\bm{\rho} \in \mixedn$, and fix any $\bm{X}$ for which the probability $1 - e^{-\Omega(n)}$ event holds:
\begin{equation}
\sup_{\bm{\sigma} \in \mixedn} \Var\brackets{\Tr\parens{\bm{H}(\bm{X}) \hat{\bm{\sigma}}}} < 3 E_\ast + t.
\end{equation}
Then, for any $\delta > 0$, $t > 0$, we can apply Cantelli's inequality:
\begin{equation}
\begin{aligned}
& \Pr_{\omega \sim \cU} \brackets{\Tr\parens{3 \bm{H}(\bm{X}) \widetilde{\bm{\cM}}(\bm{\rho}, \omega)} - \Tr\parens{\bm{H}(\bm{X})\bm{\rho}} < -\delta E_\ast} \leq \\
\leq{} & \frac{\parens{3 E_\ast + t}^2}{\parens{3 E_\ast + t}^2 + \delta^2 E_\ast^2} =
\frac{1}{1 + \frac{\delta^2 E_\ast^2}{9 E_\ast^2 + 2 \cdot 3 E_\ast t + t^2}} \leq
\frac{1}{1 + 0.99 \cdot \delta^2 / 9},
\end{aligned}
\end{equation}
where the last inequality holds if we choose sufficiently small $t$. Thus, the claim follows.
\end{proof}

\subsection{Notation Reference}

Finally, we end our Preliminaries by summarizing the notation we will use in what follows for easy reference. 

\renewcommand{\arraystretch}{1.2}
\begin{longtable}{l|p{8cm}}
\toprule
\textbf{Symbol} & \textbf{Description} \\ \midrule
$n$ & System size \\
$k$ & System locality (hyperedge size) \\
$p$ & Sparsity parameter (expected degree of each vertex) \\
$\bm{X}$ & Problem instance \\
$\vec{s}$ & Solution bit string \\
$A$ & Hyperedge (of size $k$) \\
$\bm{H}^{(n,k)}(\bm{X})$ & Quantum Hypergraph Max-Cut Hamiltonian \\
$V^{(n,k)}(\bm{X}, \vec{s})$ & Classical value for Hypergraph Max-Cut \\
$E^{(k,p)}_\ast$ & Limiting maximal energy \\
$\mixedn$ & Set of mixed states on $n$ qubits \\
$\bm{\cM}$ & Shadows estimator channel \\
$\delta$ & Multiplicative error of shadows estimator \\
$\entropy$ & Binary entropy function, \newline $\entropy(x) = -x \log x - (1 - x) \log (1 - x)$ \\
$\gamma$ & Approximation ratio \\
$f$ & Additive error in Lipschitz property \\
$L$ & Lipschitz constant \\
$F$ & Fraction of qubits resampled \\
$\frakd$ & Hypergraph degree bound \\
$\cXfd \parens{\bm{X}}$ & Event that the interaction hypergraph of $\bm{X}$ has degree at most $\frakd$ \\
$m$ & Number of replicas used in the OGP statement \\
$T$ & Large number of replicas out of which $m$ satisfy the desired property \\
$R$ & Number of samples taken from the shadows estimator \\
$Q$ & Number of steps in the interpolation path \\
$t$ & Index for OGP replicas \\
$q$ & Index for interpolation steps \\
$r$ & Index for shadows estimator samples \\
$\cF(n, m, R, \xi)$ & Set of collections of bit strings $\vec{s}^{(t,r)}$ with average distance for each $t, t'$ bounded by $\xi$ \\
$\cS\parens{\alpha, m, \nu_0, \nu_1, \parens{\bm{X}^{(t)}}_{t \in [m]}}$ & Set consisting of $m$-tuples of bit strings with values at least $\alpha$ for the $\bm{X}^{(t)}$ and all distances in $[\nu_0, \nu_1]$ \\
$\dhamming$ & Hamming distance between bit strings \\
$\norm{\cdot}_{W_\alpha}$ & Quantum Wasserstein $F$-norm of order $\alpha$ \\
$\opnorm{\cdot}$ & Operator norm \\
$\norm{\cdot}_\ast$ & Trace (nuclear) norm \\
$\sigmax, \sigmay, \sigmaz$ & Pauli X, Y, Z matrices \\
$\cU$ & Uniform distribution on $[0, 1]$ \\
$\omega$ & Randomness value for randomized algorithms \\
$\probst$ & Probability of failure of stability (or locality) \\
$\probest$ & Probability of failure of shadows estimator \\
$\probb$ & Probability of the operator norm of the Hamiltonian greatly exceeding its mean \\
$\probf$ & Probability of failure of near-optimal algorithm \\
$\Lambda$ & Distribution of problem instances \\
$\Xi_{F,m}$ & Distribution of collections of $m$ problem instances with hyperedges adjacent to the first $F$ fraction of qubits resampled \\
$\Xi_F$ & Marginal distribution of pairs of instances in $\Xi_{F,m}$ \\
$\Xi_{Q,T}$ & Distribution of problem instances for interpolation path \\
$\Xi_Q^{(q)}$ & Marginal distribution of pairs of problem instances at a given point in the interpolation path \\
\bottomrule
\end{longtable}

\section{Weak Hardness for Stable Algorithms}\label{chp:weak_hardness}

In this section, we will show that under certain conditions, \emph{stable} quantum algorithms fail to find near-optimal solutions for the Quantum Hypergraph Max-Cut problem in the average case. We call this a \emph{weak hardness} result because the locality $k$ at which the problem becomes hard is dependent on the Lipschitz constant $L$ of the algorithm---thus, this result cannot be used to prove that there exists a single $k$ for which the problem is hard for all stable algorithms. The proofs in this section are largely adapted from \citet{Anschuetz2026_glassiness}, but significantly simplified.

\begin{definition}[Quantum algorithm]
A \emph{quantum algorithm} is described by a map $\bm{\cA} : \RR^D \times \Omega \rarr \mixedn$, where $\mixedn$ is the set of mixed states on $n$ qubits, $(\Omega, \PP)$ is a probability space, and in our case, $D = \binom{n}{k}$. We say that the algorithm is \emph{deterministic} if the probability space $(\Omega, \PP)$ is trivial, and it is \emph{pure} if it has codomain $\cS_n$ (the set of \emph{pure} quantum states).
\end{definition}

\begin{definition}[Associated pure quantum algorithm]
Since we can view mixed states as probability distributions over pure states, for any deterministic quantum algorithm $\bm{\cA}$, there is an \emph{associated pure quantum algorithm} $\widetilde{\bm{\cA}}$ such that
\begin{equation}
\bm{\cA}(\bm{X}) = \EE_{\omega \sim \cU} \brackets{\widetilde{\bm{\cA}}(\bm{X}, \omega)}.
\end{equation}
\end{definition}

\begin{definition}[Stable quantum algorithm]\label{def:stable_alg}
Let $\Xi$ be a probability distribution over pairs $(\bm{X}, \bm{X}')$ of problem instances that correspond to the same interaction hypergraph (i.e. $X_A = 0 \iff X'_A = 0$). Then, we say that a quantum algorithm $\bm{\cA}$ is \emph{$(f, L, \frakd, \Xi, \probst)$-stable} if we have that
\begin{equation}
\Pr_{(\bm{X}, \bm{X}', \omega) \sim \Xi \otimes \PP_\Omega} \brackets{\norm{\bm{\cA}(\bm{X}, \omega) - \bm{\cA}(\bm{X}', \omega)}_{W_2} \leq f + L \norm{\bm{X} - \bm{X}}_1 \mid \cXfd\parens{\bm{X}}} \geq 1 - \probst,
\end{equation}
where $\cXfd\parens{\bm{X}}$ is the event that the interaction hypergraph of $\bm{X}$ has degree at most $\frakd$, i.e. for all $v \in [n]$, we have that
\begin{equation}
\abs{\braces{A \in \binom{[n]}{k} : v \in A \wedge X_A \neq 0}} \leq \frakd.
\end{equation}
\end{definition}

\begin{definition}[Near-optimal quantum algorithm]
A quantum algorithm is \emph{$(\gamma, \probf)$-optimal} for a distribution $\Lambda$ of problem instances if
\begin{equation}
\Pr_{(\bm{X}, \omega) \sim \Lambda \otimes \PP_\Omega} \brackets{\Tr\parens{\bm{H}(\bm{X}) \bm{\cA}(\bm{X}, \omega)} \geq \gamma E_\ast} \geq 1 - \probf.
\end{equation}
\end{definition}

\subsection{Stable Algorithms Imply the Existence of Forbidden Configurations}

\subsubsection{Reduction to Weakly Stable Algorithms with Classical Outputs}

In the following lemmas, we will first perform a reduction where we convert a stable randomized quantum algorithm into a deterministic one, and then we will apply the classical shadows estimator to convert it into an algorithm outputting classical states that satisfies a weaker version of the stability condition.

\begin{lemma}\label{lem:deterministic_reduction}
(From \citet[Lemma 27]{Anschuetz2026_glassiness})
Let $\Xi$ be a probability distribution of pairs of problem instances on the same hypergraph, whose marginal distributions are both $\Lambda$. Suppose that there exists a quantum algorithm $\bm{\cA}(\bm{X}, \omega)$ that is $\parens{f, L, \frakd, \Xi, \probst}$-stable and $(\gamma, \probf)$-optimal for $\Lambda$. Then, there exists a deterministic quantum algorithm $\widetilde{\bm{\cA}}(\bm{X})$ that is $\parens{f, L, \frakd, \Xi, 3 \probst}$-stable and $(\gamma, 3 \probf)$-optimal for $\Lambda$.
\end{lemma}

\begin{proof}
We have that
\begin{equation}
\begin{aligned}
& \EE_{\omega \sim \PP_\Omega} \brackets{\Pr_{(\bm{X}, \bm{X}') \sim \Xi} \brackets{\norm{\bm{\cA}(\bm{X}, \omega) - \bm{\cA}(\bm{X}', \omega)}_{W_2} > f + L \norm{\bm{X} - \bm{X}}_1 \mid \cXfd(\bm{X})}} = \\
={} & \Pr_{(\bm{X}, \bm{X}', \omega) \sim \Xi \otimes \PP_\Omega} \brackets{\norm{\bm{\cA}(\bm{X}, \omega) - \bm{\cA}(\bm{X}', \omega)}_{W_2} > f + L \norm{\bm{X} - \bm{X}}_1 \mid \cXfd(\bm{X})} \leq \probst,
\end{aligned}
\end{equation}
and
\begin{equation}
\begin{aligned}
& \EE_{\omega \sim \PP_\Omega} \brackets{\Pr_{\bm{X} \sim \Lambda} \brackets{\Tr\parens{\bm{H}(\bm{X}) \bm{\cA}(\bm{X}, \omega)} < \gamma E_\ast}} = \\
={} & \Pr_{(\bm{X}, \omega) \sim \Lambda \otimes \PP_\Omega} \brackets{\Tr\parens{\bm{H}(\bm{X}) \bm{\cA}(\bm{X}, \omega)} < \gamma E_\ast} \leq \probf.
\end{aligned}
\end{equation}
So, by Markov's inequality, we have that
\begin{equation}
\begin{aligned}
& \Pr_{\omega \sim \PP_\Omega} \Bigg[\Pr_{(\bm{X}, \bm{X}') \sim \Xi} \Big[\norm{\bm{\cA}(\bm{X}, \omega) - \bm{\cA}(\bm{X}', \omega)}_{W_2} > \\
& \hspace{10em} > f + L \norm{\bm{X} - \bm{X}}_1 \mid \cXfd(\bm{X})\Big] > 3 \probst\Bigg] \leq \\
\leq{} & \frac{1}{3 \probst} \EE_{\omega \sim \PP_\Omega} \Bigg[\Pr_{(\bm{X}, \bm{X}') \sim \Xi} \Big[\norm{\bm{\cA}(\bm{X}, \omega) - \bm{\cA}(\bm{X}', \omega)}_{W_2} > \\
& \hspace{10em} > f + L \norm{\bm{X} - \bm{X}}_1 \mid \cXfd(\bm{X})\Big] \Bigg] \leq \frac{1}{3},
\end{aligned}
\end{equation}
and
\begin{equation}
\begin{aligned}
& \Pr_{\omega \sim \PP_\Omega} \brackets{\Pr_{\bm{X} \sim \Lambda} \brackets{\Tr\parens{\bm{H}(\bm{X}) \bm{\cA}(\bm{X}, \omega)} < \gamma E_\ast} > 3 \probf} \leq \\
\leq{} & \frac{1}{3 \probf} \EE_{\omega \sim \PP_\Omega} \brackets{\Pr_{\bm{X} \sim \Lambda} \brackets{\Tr\parens{\bm{H}(\bm{X}) \bm{\cA}(\bm{X}, \omega)} < \gamma E_\ast}} \leq \frac{1}{3}.
\end{aligned}
\end{equation}
So now, by the union bound, we have that
\begin{equation}
\begin{aligned}
& \Pr_{\omega \sim \PP_\Omega} \Bigg[\Pr_{\bm{X} \sim \Lambda} \brackets{\Tr\parens{\bm{H}(\bm{X}) \bm{\cA}(\bm{X}, \omega)} < \gamma E_\ast} > 3 \probf \vee \\
& \vee \Pr_{(\bm{X}, \bm{X}') \sim \Xi} \brackets{\norm{\bm{\cA}(\bm{X}, \omega) - \bm{\cA}(\bm{X}', \omega)}_{W_2} > f + L \norm{\bm{X} - \bm{X}}_1} > 3 \probst \mid \cXfd(\bm{X}) \Bigg] \leq \frac{2}{3}.
\end{aligned}
\end{equation}
Thus, there must exist some $\omega^\ast \in \Omega$ such that we have
\begin{equation}
\Pr_{(\bm{X}, \bm{X}') \sim \Xi} \brackets{\norm{\bm{\cA}(\bm{X}, \omega^\ast) - \bm{\cA}(\bm{X}', \omega^\ast)}_{W_2} \leq f + L \norm{\bm{X} - \bm{X}}_1 \mid \cXfd(\bm{X})} \geq 1 - 3 \probst,
\end{equation}
and
\begin{equation}
\Pr_{\bm{X} \sim \Lambda} \brackets{\Tr\parens{\bm{H}(\bm{X}) \bm{\cA}(\bm{X}, \omega^\ast)} \geq \gamma E_\ast} \geq 1 - 3 \probf.
\end{equation}
Thus, if we define $\widetilde{\bm{\cA}}(\bm{X}) \coloneq \bm{\cA}(\bm{X}, \omega^\ast)$, then $\widetilde{\bm{\cA}}$ is a deterministic quantum algorithm that is $\parens{f, L, \Xi, 3 \probst}$-stable and $(\gamma, 3 \probf)$-optimal.
\end{proof}

Now, if we have a stable and near-optimal deterministic quantum algorithm, we can apply the shadows estimator to it, obtaining a randomized algorithm that only outputs classical states. This algorithm will be close to $1/3$-optimal with high probability and it will stable \emph{in expectation} (which is a weaker notion than stability as we have previously defined it).

\begin{lemma}\label{lem:nondeterministic_shadows_reduction}
(Modified version of \citet[Lemma 28]{Anschuetz2026_glassiness}) Let $\Xi$ be a probability distribution of pairs of problem instances on the same hypergraph whose marginal distributions are both $\Lambda$. Suppose that the Pauli shadows estimator is $(\delta, \probest, \probb)$-efficient for the class of Hamiltonians $\bm{H}(\Lambda)$. Suppose there exists a deterministic quantum algorithm $\widetilde{\bm{\cA}}$ that is $\parens{f, L, \frakd, \Xi, 3 \probst}$-stable and $(\gamma, 3 \probf)$-optimal for $\Lambda$. Then, there exists a pure quantum algorithm $\bm{\cG} : \RR^{\binom{n}{k}} \times [0, 1] \rarr \{0, 1\}^n$ (that is, a randomized algorithm outputting only classical states), such that the following weaker notion of stability holds:
\begin{equation}\label{eq:shadows_reduction_weak_stability}
\begin{aligned}
& \Pr_{(\bm{X}, \bm{X}') \sim \Xi} \Bigg[\norm{\EE_{\omega \sim \cU}\brackets{\selfouterprod{\bm{\cG}(\bm{X}, \omega)} - \selfouterprod{\bm{\cG}(\bm{X}', \omega)}}}_{W_2} \leq \\
& \hspace{15em} \leq{} f + L\norm{\bm{X} - \bm{X}'}_1 \mid \cXfd(\bm{X}) \Bigg] \geq 1 - 3 \probst.
\end{aligned}
\end{equation}
Furthermore, consider some $\bm{X}$ that satisfies the following two conditions (which is true with probability at least $1 - 3 \probf - \probb$ over $\bm{X} \sim \Lambda$):
\begin{equation}\label{eq:shadows_reduction_condition_1}
\Tr\parens{\bm{H}(\bm{X}) \widetilde{\bm{\cA}}(\bm{X})} \geq \gamma E_\ast
\end{equation}
\begin{equation}\label{eq:shadows_reduction_condition_2}
\Pr_{\omega \sim \cU} \brackets{\Tr\parens{3 \bm{H}(\bm{X}) \widetilde{\bm{\cM}}(\bm{\rho}, \omega)} - \Tr\parens{\bm{H}(\bm{X})\bm{\rho}} \geq -\delta E_\ast} \geq 1 - \probest \quad \forall \bm{\rho} \in \mixedn.
\end{equation}
For such $\bm{X}$, we must have that
\begin{equation}\label{eq:shadows_reduction_near_optimality}
\Pr_{\omega \sim \cU} \brackets{V \parens{\bm{X}, \bm{\cG}(\bm{X}, \omega)} \geq \frac{1}{3}(\gamma - \delta) E_\ast} \geq 1 - \probest.
\end{equation}
\end{lemma}

\begin{proof}
We define $\bm{\cG}(\bm{X}, \omega)$ as follows. First, run the deterministic quantum algorithm $\bm{\cA}(\bm{X})$ to obtain a quantum state $\bm{\rho}$. Then randomly choose one of the $X$, $Y$, and $Z$ bases and measure $\bm{\rho}$ in that basis, outputting the resulting bit string. Consider the related algorithm $\bm{\cG}'$ defined as
\begin{equation}
\bm{\cG}'(\bm{X}, \omega) \coloneq \widetilde{\bm{\cM}}(\widetilde{\bm{\cA}}(\bm{X}), \omega).
\end{equation}
That is, if the result of measuring the qubits in basis $b$ was the bit string $\vec{s}$, the output of $\bm{\cG}$ is $\vec{s}$ while the output of $\bm{\cG'}$ is $\ket{b; \vec{s}}$. Observe that since our model is symmetric with respect to exchanging $X$, $Y$, and $Z$ bases, we have that
\begin{equation}
\bra{\bm{\cG}(\bm{X}, \omega)} \bm{H}(\bm{X}) \ket{\bm{\cG}(\bm{X}, \omega)} = \Tr\parens{\bm{H}(\bm{X})\bm{\cG}'(\bm{X}, \omega)}.
\end{equation}

Now, we have the following:
\begin{equation}
\begin{aligned}
& \norm{\EE_{\omega \sim \cU}\brackets{\selfouterprod{\bm{\cG}(\bm{X}, \omega)} - \selfouterprod{\bm{\cG}(\bm{X}', \omega)}}}_{W_2} \leq \\
\leq{} & \norm{\EE_{\omega \sim \cU}\brackets{\bm{\cG}'(\bm{X}, \omega)} - \bm{\cG}'(\bm{X}', \omega)}_{W_2} = \\
={} & \norm{\EE_{\omega \sim \cU}\brackets{\widetilde{\bm{\cM}}(\widetilde{\bm{\cA}}(\bm{X}), \omega) - \widetilde{\bm{\cM}}(\widetilde{\bm{\cA}}(\bm{X}'), \omega)}}_{W_2} = \\
={} & \norm{\bm{\cM}(\widetilde{\bm{\cA}}(\bm{X})) - \bm{\cM}(\widetilde{\bm{\cA}}(\bm{X}'))}_{W_2} \leq
\norm{\widetilde{\bm{\cA}}(\bm{X}) - \widetilde{\bm{\cA}}(\bm{X}')}_{W_2},
\end{aligned}
\end{equation}
where the last inequality is due to the contractivity of quantum Wasserstein distance under tensor product channels \cite{Anschuetz2026_glassiness}. So now, we have that
\begin{equation}
\begin{aligned}
& \Pr_{(\bm{X}, \bm{X}') \sim \Xi} \Bigg[\norm{\EE_{\omega \sim \cU}\brackets{\selfouterprod{\bm{\cG}(\bm{X}, \omega)} - \selfouterprod{\bm{\cG}(\bm{X}', \omega)}}}_{W_2} \leq \\
& \hspace{15em} \leq f + L\norm{\bm{X} - \bm{X}'}_1 \mid \cXfd(\bm{X})\Bigg] \geq \\
\geq{} & \Pr_{(\bm{X}, \bm{X}') \sim \Xi} \brackets{\norm{\widetilde{\bm{\cA}}(\bm{X}) - \widetilde{\bm{\cA}}(\bm{X}')}_{W_2} \leq f + L\norm{\bm{X} - \bm{X}'}_1 \mid \cXfd(\bm{X})} \geq 1 - 3 \probst,
\end{aligned}
\end{equation}
which proves \Cref{eq:shadows_reduction_weak_stability}. Now, we want to show the second claim. By \Cref{eq:shadows_reduction_condition_1} and \Cref{eq:shadows_reduction_condition_2}, we have that
\begin{equation}
\begin{aligned}
& \Pr_{\omega \sim \cU} \brackets{V \parens{\bm{X}, \bm{\cG}(\bm{X}, \omega)} \geq \frac{1}{3}(\gamma - \delta) E_\ast} = \\
={} & \Pr_{\omega \sim \cU} \brackets{\Tr\parens{3 \bm{H}(\bm{X}) \bm{\cG}'(\bm{X}, \omega)} \geq (\gamma - \delta) E_\ast} \geq \\
\geq{} & \Pr_{\omega \sim \cU} \brackets{\Tr\parens{3 \bm{H}(\bm{X}) \widetilde{\bm{\cM}}(\widetilde{\bm{\cA}}(\bm{X}), \omega)} - \Tr\parens{\bm{H}(\bm{X}) \widetilde{\bm{\cA}}(\bm{X})} \geq -\delta E_\ast} \geq 1 - \probest.
\end{aligned}
\end{equation}
So, \Cref{eq:shadows_reduction_near_optimality} holds.
\end{proof}

\subsubsection{Constructing Collectively Stable Solution Configurations}

Now that we have an algorithm satisfying the weak stability condition, the next step of the proof proceeds as follows. We sample many independent outputs from the algorithm, taking a large enough number so that at least one is near-optimal with high probability, and we argue that these outputs satisfy a notion of ``collective stability'': that is, they obey the Lipschitz property \emph{on average}. We will first need this result on the properties of the $W_2$ distance:

\begin{proposition}[{From~\citet[Proposition 29]{Anschuetz2026_glassiness}}]\label{prop:expected_wasserstein_distance}
For each problem instance $\bm{X}$, let $p_{\bm{X}}(\vec{s})$ be the probability distribution of $\bm{\cG}(\bm{X}, \omega) \in \{0, 1\}^n$ over $\omega \sim \cU$. Then, for every pair of problem instances $\bm{X}, \bm{X}'$, there exists some probability distribution $\pi_{\bm{X}, \bm{X}'} (\vec{s}, \vec{s}')$ such that
\begin{equation}
\begin{aligned}
& \EE_{\vec{s}, \vec{s}' \sim \pi_{\bm{X}, \bm{X}'}} \brackets{\dhamming(\vec{s}, \vec{s}')^2} \leq \\
\leq{} & \norm{\EE_{\omega \sim \cU} \brackets{\selfouterprod{\bm{\cG}(\bm{X}, \omega)} - \selfouterprod{\bm{\cG}(\bm{X}', \omega)}}}_{W_2}^2,
\end{aligned}
\end{equation}
\begin{equation}
\sum_{\vec{s}'} \pi_{\bm{X}, \bm{X}'}(\vec{s}, \vec{s}') = p_{\bm{X}}(\vec{s}),
\end{equation}
\begin{equation}
\sum_{\vec{s}} \pi_{\bm{X}, \bm{X}'}(\vec{s}, \vec{s}') = p_{\bm{X}'}(\vec{s}').
\end{equation}
\end{proposition}
\begin{proof}
We can write
\begin{equation}
\EE_{\omega \sim \cU} \brackets{\selfouterprod{\bm{\cG}(\bm{X}, \omega)} - \selfouterprod{\bm{\cG}(\bm{X}', \omega)}} = \sum_{i=1}^n \bm{A}_i,
\end{equation}
where the $\bm{A}_i$ are Hermitian such that $\Tr_{\{i\}} \parens{\bm{A}_i} = 0$, picked to be the optimal in \Cref{def:quantum_wasserstein}. Now, we can write each $\bm{A}_i$ as $\bm{A}_i^+ - \bm{A}_i^-$, where $\bm{A}_i^+$ and $\bm{A}_i^-$ are positive semidefinite with the same trace. So, we can write $\bm{A}_i = c_i \parens{\bm{\rho}_i - \bm{\sigma}_i}$, where $\bm{\rho}_i$ and $\bm{\sigma}_i$ are density matrices diagonal in the computational basis, where
\begin{equation}
c_i = \norm{\frac{1}{2} \bm{A}_i}_\ast = \Tr\parens{\bm{A}_i^+} = \Tr\parens{\bm{A}_i^-}.
\end{equation}
Then, we have that
\begin{equation}
\begin{aligned}
& \sum_{i=1}^n \sqrt{c_i} =
\sum_{i=1}^n \sqrt{\norm{\frac{1}{2} \bm{A}_i}_\ast} = \\
={} & \norm{\EE_{\omega \sim \cU} \brackets{\selfouterprod{\bm{\cG}(\bm{X}, \omega)} - \selfouterprod{\bm{\cG}(\bm{X}', \omega)}}}_{W_2}.
\end{aligned}
\end{equation}
Since the density matrices $\bm{\rho}_i$ and $\bm{\sigma}_i$ are diagonal in the computational basis, we can write $p_i, q_i$ for their diagonals (representing probability distributions over bit strings), respectively. Similarly, write $p, q$ for the diagonal elements of $\bm{\rho}$, $\bm{\sigma}$. Observe that the marginals of $p_i$ and $q_i$ on $[n] \setminus \{i\}$ are the same, which means that $W_2 (p_i, q_i) \leq 1$. Thus, we have that
\begin{equation}
\begin{aligned}
& \inf_{\pi_{\bm{X},\bm{X}'} \in \cC\parens{p_{\bm{X}}, p_{\bm{X'}}}} \EE_{\vec{s}, \vec{s}' \sim \pi_{\bm{X}, \bm{X}'}} \brackets{\dhamming(\vec{s}, \vec{s}')^2} = W_2\parens{p, q}^2 \leq \\
\leq{} & \parens{\sum_{i=1}^n \sqrt{c_i} W_2 (p_i, q_i)}^2 \leq \parens{\sum_{i=1}^n \sqrt{c_i}}^2 = \\
={} & \norm{\EE_{\omega \sim \cU} \brackets{\selfouterprod{\bm{\cG}(\bm{X}, \omega)} - \selfouterprod{\bm{\cG}(\bm{X}', \omega)}}}_{W_2}^2.
\end{aligned}
\end{equation}
\end{proof}

\begin{lemma}[{From~\citet[Lemma 30]{Anschuetz2026_glassiness}}]\label{lem:deterministic_shadows_reduction}
Let $\Xi$ be a probability distribution of pairs of problem instances on the same hypergraph whose marginal distributions are both $\Lambda$. Fix $\beta > 0$ and $R \in \NN$ such that
\begin{equation}
\frac{1}{\beta^2} + \probest^R < 1.
\end{equation}
Then, with probability at least $1 - 3 \probst - 3 \probf - \probb$ over $(\hat{\bm{X}}, \bm{X}) \sim \Xi$, conditioned on the event $\cXfd\parens{\hat{\bm{X}}}$, there exists a set of bit strings $\parens{\hat{\vec{s}}^{(r)}, \vec{s}^{(r)} \in \{0, 1\}^n}_{r \in [R]}$ such that we have a ``collective stability'' condition:
\begin{equation}
\frac{1}{R} \sum_{r=1}^R \dhamming\parens{\hat{\vec{s}}^{(r)}, \vec{s}^{(r)}} \leq \beta f + \beta L \norm{\hat{\bm{X}} - \bm{X}}_1,
\end{equation}
and
\begin{equation}
\max_{r \in [R]} V\parens{\bm{X}, \vec{s}^{(r)}} \geq \frac{1}{3}(\gamma - \delta) E_\ast.
\end{equation}
\end{lemma}

\begin{proof}
Consider some $\hat{\bm{X}},\bm{X}$ that satisfies the following conditions:
\begin{equation}
\norm{\widetilde{\bm{\cA}}(\hat{\bm{X}}) - \widetilde{\bm{\cA}}(\bm{X})}_{W_2} \leq f + L \norm{\hat{\bm{X}} - \bm{X}}_1,
\end{equation}
\begin{equation}
\Tr\parens{\bm{H}(\bm{X}) \widetilde{\bm{\cA}}(\bm{X})} \geq \gamma E_\ast,
\end{equation}
\begin{equation}
\Pr_{\omega \sim \cU} \brackets{\Tr\parens{3 \bm{H}(\bm{X}) \widetilde{\bm{\cM}}(\bm{\rho}, \omega)} - \Tr\parens{\bm{H}(\bm{X})\bm{\rho}} \geq -\delta E_\ast} \geq 1 - \probest \quad \forall \bm{\rho} \in \mixedn.
\end{equation}
By the union bound, the probability over $\Xi$ that all of these conditions occur is at least
\begin{equation}
1 - 3 \probst - 3 \probf - \probb.
\end{equation}
Let $p_{\hat{\bm{X}}}, p_{\bm{X}}$ be the probability distributions of $\bm{\cG}(\hat{\bm{X}}, \omega)$ and $\bm{\cG}(\bm{X}, \omega)$, respectively, over $\omega \sim \cU$, and let $\pi_{\hat{\bm{X}}, \bm{X}}$ be the joint distribution as in \Cref{prop:expected_wasserstein_distance}. We then have that when $\hat{\bm{X}}, \bm{X}$ satisfy the above conditions, by \Cref{lem:nondeterministic_shadows_reduction} and \Cref{prop:expected_wasserstein_distance}, we have that
\begin{equation}
\begin{aligned}
& \EE_{\parens{\hat{\vec{s}}, \vec{s}} \sim \pi} \brackets{\dhamming\parens{\hat{\vec{s}}, \vec{s}}^2} \leq \\
\leq{} & \norm{\EE_{\omega \sim \cU} \brackets{\selfouterprod{\bm{\cG}(\hat{\bm{X}}, \omega)} - \selfouterprod{\bm{\cG}(\bm{X}, \omega)}}}_{W_2}^2 \leq \\
\leq{} & \parens{f + L\norm{\hat{\bm{X}} - \bm{X}}_1}^2.
\end{aligned}
\end{equation}

We will now use this distribution to sample $R$ many independent replicas from the algorithm $\bm{\cG}$ (the shadows estimator applied to $\widetilde{\bm{\cA}}$), and show that they satisfy the desired conditions.

By Markov's inequality,
\begin{equation}
\begin{aligned}
& \Pr_{\parens{\hat{\vec{s}}^{(r)}, \vec{s}^{(r)}}_{r \in [R]} \sim \pi_{\hat{\bm{X}}, \bm{X}}^{\times R}} \brackets{\frac{1}{R} \sum_{r=1}^R \dhamming\parens{\hat{\vec{s}}^{(r)}, \vec{s}^{(r)}} > \beta \parens{f + L\norm{\hat{\bm{X}} - \bm{X}}_1}} = \\
={} & \Pr_{\parens{\hat{\vec{s}}^{(r)}, \vec{s}^{(r)}}_{r \in [R]} \sim \pi_{\hat{\bm{X}},\bm{X}}^{\times R}} \brackets{\parens{\frac{1}{R} \sum_{r=1}^R \dhamming\parens{\hat{\vec{s}}^{(r)}, \vec{s}^{(r)}}}^2 > \beta^2 \parens{f + L\norm{\hat{\bm{X}} - \bm{X}}_1}^2} \leq \\
\leq{} & \frac{\EE_{\parens{\hat{\vec{s}}^{(r)}, \vec{s}^{(r)}}_{r \in [R]} \sim \pi_{\hat{\bm{X}},\bm{X}}^{\times R}} \brackets{\parens{\frac{1}{R} \sum_{r=1}^R \dhamming\parens{\hat{\vec{s}}^{(r)}, \vec{s}^{(r)}}}^2}}{\beta^2 \parens{f + L\norm{\hat{\bm{X}} - \bm{X}}_1}^2} \leq \\
\leq{} & \frac{\EE_{\parens{\hat{\vec{s}}, \vec{s}} \sim \pi_{\hat{\bm{X}},\bm{X}}} \brackets{\dhamming\parens{\hat{\vec{s}}, \vec{s}}^2}}{\beta^2 \parens{f + L\norm{\hat{\bm{X}} - \bm{X}}_1}^2} \leq
\frac{1}{\beta^2}.
\end{aligned}
\end{equation}
Now, from \Cref{lem:nondeterministic_shadows_reduction}, for our choice of $\hat{\bm{X}}, \bm{X}$, we also have that
\begin{equation}
\Pr_{\vec{s} \sim p_{\bm{X}}} \brackets{V\parens{\bm{X}, \vec{s}} < \frac{1}{3}(\gamma - \delta) E_\ast} =
\Pr_{\omega \sim \cU} \brackets{V \parens{\bm{X}, \bm{\cG}(\bm{X}, \omega)} < \frac{1}{3}(\gamma - \delta) E_\ast} \leq \probest.
\end{equation}
So, since the replicas are sampled independently, we have that
\begin{equation}
\Pr_{\parens{\hat{\vec{s}}^{(r)}, \vec{s}^{(r)}}_{r \in [R]} \sim \pi_{\hat{\bm{X}}, \bm{X}}^{\times R}} \brackets{\max_{r \in [R]} V\parens{\bm{X}, \vec{s}^{(r)}} < \frac{1}{3}(\gamma - \delta) E_\ast} \leq \probest^R.
\end{equation}
Since we chose $\beta, R$ such that
\begin{equation}
\frac{1}{\beta^2} + \probest^R < 1,
\end{equation}
there must exist $\parens{\vec{s}^{(r)} \in \{0, 1\}^n}_{r \in [R]}$ such that the desired conditions are satisfied.
\end{proof}

\subsubsection{Taking Many Correlated Problem Instances}

We now extend the above results to many problem instances that are highly correlated with each other, with each one only resampling hyperedges adjacent to a small fraction $F$ of the qubits. This is a simplified version of the \emph{interpolation path} typically seen in OGP proofs (and which we will see in \Cref{chp:strong_hardness}). Here, instead of constructing a full interpolation path, we are effectively only taking a single step along it. It is important to note that we do not lose much by doing this: because of the way the collectively stable configurations are constructed, using a full interpolation path does not allow us to prove stronger results. We explain this limitation in more detail in \Cref{sec:limitations}.

\begin{definition}\label{def:single_step_interpolation_path}
For $m \in \NN$ and $F \in (0, 1]$, we define the distribution $\Xi_{F,m}$ of collections \\ $\parens{\hat{\bm{X}}, \bm{X}^{(1)}, \dots, \bm{X}^{(m)}}$ as follows. First, sample $\bm{S}$ by setting each $S_A \sim \Bern\parens{p / \binom{n - 1}{p - 1}}$ i.i.d. Then, sample $\hat{\bm{J}}, \widetilde{\bm{J}}^{(1)}, \dots, \widetilde{\bm{J}}^{(m)}$ with i.i.d. entries in $\{-1, 1\}$. Then, for all and $A \in \binom{[n]}{k}$, set $\bm{X}$
\begin{equation}
X^{(t)}_A = \begin{cases}
S_A \widetilde{J}^{(t)}_A & \text{ if } A \cap \brackets{\ceil{F n}} \neq \varnothing \\
S_A \hat{J}_A & \text{ otherwise. }
\end{cases}
\end{equation}
In other words, we are interpolating between $\hat{\bm{X}}$ and $m$ other problem instances on the same hypergraph with independent weights by resampling the weights of hyperedges intersecting with the first $F$ fraction of qubits. We call the marginal distribution of $\parens{\hat{\bm{X}}, \bm{X}^{(t)}}$ for a fixed $t$ by the name $\Xi_F$. Observe that these marginal distributions are independent of $t$.
\end{definition}

\begin{lemma}\label{lem:m_replicas}
Suppose that there exists a quantum algorithm $\bm{\cA}$ that is \\ $\parens{f, L, \frakd, \Xi_F, \probst}$-stable and $(\gamma, \probf)$-optimal for $\QMC(n, k, p)$. Let $\beta$ be such that
\begin{equation}
\frac{m}{\beta^2} + m \probest^R < 1.
\end{equation}
Then, with probability at least
\begin{equation}
1 - 3 m \probst - 3 m \probf - m \probb
\end{equation}
over $\parens{\hat{\bm{X}}, \bm{X}^{(1)}, \dots, \bm{X}^{(m)}} \sim \Xi_{F,m}$, conditioned on $\cXfd\parens{\hat{\bm{X}}}$, there exist bit strings $\hat{\vec{s}}^{(r)}$, $\vec{s}^{(t,r)}$ for $t \in [m]$, $r \in [R]$ such that for all $t \in [m]$ we have that
\begin{equation}
\frac{1}{R} \sum_{r=1}^R \dhamming\parens{\hat{\vec{s}}, \vec{s}^{(t,r)}} \leq \beta f + \beta L \norm{\hat{\bm{X}} - \bm{X}^{(t)}}_1,
\end{equation}
and
\begin{equation}
\max_{r \in [R]} V \parens{\bm{X}^{(t)}, \vec{s}^{(t,r)}} \geq \frac{1}{3}\parens{\gamma - \delta}E_\ast.
\end{equation}
\end{lemma}

\begin{proof}
Fix $\hat{\bm{X}}, \bm{X}^{(1)}, \dots, \bm{X}^{(m)}$ such that for each $t \in [m]$, $\parens{\hat{\bm{X}}, \bm{X}^{(t)}}$ satisfies the condition in the proof of \Cref{lem:deterministic_shadows_reduction}. Then, by the union bound, this occurs with probability of at least $1 - 3 m \probst - 3 m \probf - m \probb$ over $\Xi_{F,m}$. Then, applying a union bound to the argument in \Cref{lem:deterministic_shadows_reduction}, we have that
\begin{equation}
\begin{aligned}
& \Pr_{\parens{\hat{\vec{s}}^{(r)}, \vec{s}^{(0,r)}, \dots, \vec{s}^{(m,r)}}_{r \in [R]} \sim \parens{\Xi_{F,m}}^{\times R}} \\
& \hspace{10em} \brackets{\max_{t \in [m]} \frac{1}{R} \sum_{r=1}^R \dhamming\parens{\hat{\vec{s}}^{(r)}, \vec{s}^{(r)}} > \beta \parens{f + L\norm{\hat{\bm{X}} - \bm{X}}_1}} \leq \frac{m}{\beta^2}
\end{aligned}
\end{equation}
and
\begin{equation}
\Pr_{\parens{\hat{\vec{s}}^{(r)}, \vec{s}^{(r)}}_{r \in [R]} \sim \pi_{\hat{\bm{X}}, \bm{X}}^{\times R}} \brackets{\min_{t \in [m]} \max_{r \in [R]} V\parens{\bm{X}, \vec{s}^{(r)}} < \frac{1}{3}(\gamma - \delta) E_\ast} \leq m \probest^R.
\end{equation}
Since $m / \beta^2 + m \probest^R < 1$, this completes the proof.
\end{proof}

\begin{lemma}\label{lem:l1_bound}
If $\bm{X}, \bm{X}'$ are $d$-dimensional vectors sampled by first setting each entry to $0$ in both vectors with probability $1 - p$ and then setting the remaining entries to $-1$ or $1$ independently, we have that
\begin{equation}
\Pr\brackets{\norm{\bm{X} - \bm{X}'}_1 \geq 3 p d} \leq e^{-p d (3 \log 2 - 2)}.
\end{equation}
\end{lemma}

\begin{proof}
\begin{equation}
\begin{aligned}
& \Pr\brackets{\norm{\bm{X} - \bm{X}'}_1 > t} =
\Pr\brackets{e^{\alpha \norm{\bm{X} - \bm{X}'}_1} > e^{\alpha t}} \leq \\
\leq{} & e^{-\alpha t} \EE\brackets{e^{\alpha \norm{\bm{X} - \bm{X}'}_1}} =
e^{-\alpha t} \parens{\frac{p}{2} e^{2\alpha} + 1 - \frac{p}{2}}^d \leq \\
={} & e^{-\alpha t} \parens{1 + \frac{p}{2} e^{2\alpha}}^d \leq
e^{-\alpha t} e^{\frac{1}{2} p d e^{2\alpha}}.
\end{aligned}
\end{equation}
If we take $\alpha = \log 2$ and $t = 3 p d$, this becomes
\begin{equation}
\Pr\brackets{\norm{\bm{X} - \bm{X}'}_1 > 3 p d} \leq e^{2 p d - 3 p d \log 2} = e^{-p d (3 \log 2 - 2)}.
\end{equation}
\end{proof}

\begin{lemma}\label{lem:ogp_existence}
Suppose that there exists a quantum algorithm $\bm{\cA}$ that is \\ $\parens{f, L, \Xi_F, \probst}$-stable and $(\gamma, \probf)$-optimal for $\QMC(n, k, p)$. Let $\beta$ be such that
\begin{equation}
\frac{m}{\beta^2} + m \probest^R < 1.
\end{equation}
Then, with probability at least
\begin{equation}
1 - 3 m \probst - 3 m \probf - m \probb - m e^{-p F n (3 \log 2 - 2)}
\end{equation}
over $\parens{\hat{\bm{X}}, \bm{X}^{(1)}, \dots, \bm{X}^{(T)}} \sim \Xi_{F,m}$, conditioned on $\cXfd\parens{\hat{\bm{X}}}$, there exist bit strings $\vec{s}^{(t,r)}$ for $t \in [m]$, $r \in [R]$ such that for all $t \neq t' \in [m]$,
\begin{equation}
\frac{1}{R n} \sum_{r=1}^R \dhamming\parens{\vec{s}^{(t,r)}, \vec{s}^{(t',r)}} \leq \frac{2 \beta f}{n} + 6 \beta L p F,
\end{equation}
and for all $t \in [m]$,
\begin{equation}
\max_{r \in [R]} V \parens{\bm{X}^{(t)}, \vec{s}^{(t,r)}} \geq \frac{1}{3}\parens{\gamma - \delta}E_\ast.
\end{equation}
\end{lemma}

\begin{proof}
Observe that $\hat{X}^{(t)}_A$ can be different from $X^{(t)}_A$ only if we have $A \cap [\ceil{F n}] \neq \varnothing$. Since $A$ must contain one of $F n$ many qubits, the number of such $A$ is at most
\begin{equation}
F n\binom{n - 1}{k - 1}.
\end{equation}
Then,
\begin{equation}
\frac{p}{\binom{n - 1}{k - 1}} \cdot F n \binom{n - 1}{k - 1} = p F n.
\end{equation}
So, by \Cref{lem:l1_bound}, we have that
\begin{equation}
\Pr_{\parens{\hat{\bm{X}}, \bm{X}^{(1)}, \dots, \bm{X}^{(m)}} \sim \Xi_{F,m}} \brackets{\max_{t \in [m]} \norm{\hat{\bm{X}} - \bm{X}^{(t)}}_1 \leq 3 p F n} \geq 1 - m e^{-p F n (3 \log 2 - 2)}
\end{equation}
Suppose that this event and the event from \Cref{lem:m_replicas} (which occurs with probability $\geq 1 - 3 m \probst - 3 m \probf - m \probb$ conditioned on $\cXfd\parens{\hat{\bm{X}}}$. Then, for all $t \in [m]$, we have that
\begin{equation}
\max_{r \in [R]} V \parens{\bm{X}^{(t)}, \vec{s}^{(t,r)}} \geq \frac{1}{3}\parens{\gamma - \delta}E_\ast,
\end{equation}
and for $t \neq t'$, we have that
\begin{equation}
\begin{aligned}
& \frac{1}{R} \sum_{r=1}^R \dhamming\parens{\vec{s}^{(t,r)}, \vec{s}^{(t',r)}} \leq \\
\leq{} & 2 \beta f + \beta L \parens{\norm{\hat{\bm{X}} - \bm{X}^{(t)}}_1 + \norm{\hat{\bm{X}} - \bm{X}^{(t')}}_1} \leq \\
\leq{} & 2 \beta f + 6 \beta L p F n,
\end{aligned}
\end{equation}
so
\begin{equation}
\frac{1}{R n} \sum_{r=1}^R \dhamming\parens{\vec{s}^{(t,r)}, \vec{s}^{(t',r)}} \leq \frac{2 \beta f}{n} + 6 \beta L p F.
\end{equation}
\end{proof}

\subsubsection{Lower Bound on the Probability of Bounded Degree}

In the previous calculations, we have been conditioning on the event $\cXfd(\bm{S})$ that our hypergraph has maximum degree bounded by a constant $\frakd$. In this section, we lower-bound the probability of this event for our distribution of problem instances (where each vertex in the hypergraph has expected degree $p$).

\begin{proposition}\label{prop:constant_degree_prob}
Suppose that each $S_A$ for $A \in \binom{[n]}{k}$ is i.i.d. $\Bern\parens{p / \binom{n - 1}{k - 1}}$. Then,
\begin{equation}
\Pr\brackets{\cXfd(\bm{S})} \geq \parens{1 - \parens{\frac{e \frakd}{p}}^\frakd e^{-p}}^n.
\end{equation}
\end{proposition}

\begin{proof}
For $v \in [n]$, let $D_v$ be the degree of $v$ in the hypergraph. We have that $\EE\brackets{D_v} = p$. Then, we have that
\begin{equation}
\begin{aligned}
& \Pr\brackets{D_v > d} =
\Pr\brackets{e^{\alpha D_v} > e^{\alpha d}} \leq
e^{-\alpha \frakd} \EE\brackets{e^{\alpha D_v}} = \\
={} & e^{-\alpha \frakd} \parens{\frac{p}{\binom{n - 1}{k - 1}} e^\alpha + 1 - \frac{p}{\binom{n - 1}{k - 1}}}^{\binom{n-1}{k-1}} =
e^{-\alpha \frakd} \parens{1 + \frac{p}{\binom{n - 1}{k - 1}} (e^\alpha - 1)}^{\binom{n-1}{k-1}} \leq \\
\leq{} & e^{-\alpha \frakd + p (e^\alpha - 1)}.
\end{aligned}
\end{equation}
Optimizing for $\alpha$, we must have $-\frakd + p e^\alpha = 0$, so $\alpha = \log(\frakd / p)$, so
\begin{equation}
\Pr\brackets{D_v > \frakd} \leq e^{-\frakd \log(\frakd / p) + \frakd - p} = \parens{\frac{p}{\frakd}}^\frakd e^{\frakd - p} = \parens{\frac{e p}{\frakd}}^\frakd e^{-p}.
\end{equation}
Now, observe that
\begin{equation}
\Pr\brackets{\cXfd(\bm{S})} = \Pr\brackets{D_v \leq \frakd \; \forall v \in [n]} \geq \prod_{v \in [n]} \Pr\brackets{D_v \leq \frakd} \geq \parens{1 - \parens{\frac{e p}{\frakd}}^\frakd e^{-p}}^n.
\end{equation}
\end{proof}

\begin{corollary}\label{cor:constant_degree_prob}
For any $0 < \epsilon \leq \log 2$, if we let
\begin{equation}
\frakd \geq \ceil{\max\parens{e^2 p, 2 \log(1 / \epsilon)}},
\end{equation}
then
\begin{equation}
\Pr\brackets{\cXfd\parens{\bm{S}}} \geq e^{-\epsilon n}.
\end{equation}
\end{corollary}

\begin{proof}
By \Cref{prop:constant_degree_prob}, we have that
\begin{equation}
\Pr\brackets{\cXfd(\bm{S})} \geq
\parens{1 - \parens{\frac{e p}{\frakd}}^\frakd e^{-p}}^n \geq
\parens{1 - e^{-2\log(1/\epsilon)} e^{-p}}^n \geq (1 - \epsilon^2)^n \geq e^{-\epsilon n}.
\end{equation}
\end{proof}

\subsection{Non-Existence of Forbidden Configurations}\label{sec:weak_hardness_QOGP}

In this section, we will demonstrate a (simplified version of) the QOGP. We will argue that the set of ``forbidden configurations'' of solution bit strings is empty with high probability, contradicting the above argument showing the existence of these configurations if a stable algorithm exists. This ultimately allows us to prove the weak hardness result for stable algorithms.

\subsubsection{Entropy Bound}

\begin{definition}\label{def:F_set}
Let
\begin{equation}
\cF(n, m, R, \xi) \coloneq \braces{\parens{\vec{s}^{(t,r)} \in \{0, 1\}^n}_{t \in [m], r \in [R]} : \frac{1}{R n} \sum_{r=1}^R \dhamming\parens{\vec{s}^{(t,r)}, \vec{s}^{(t',r)}} \leq \xi}.
\end{equation}
\end{definition}

\begin{lemma}\label{lem:F_set_size}
\begin{equation}
\abs{\cF(n, m, R, \xi)} \leq \exp\brackets{n \cdot \parens{R \log 2 + (m - 1) R \entropy(\xi)}(1 + o_n(1))}.
\end{equation}
\end{lemma}
\begin{proof}
Let $\vec{s}^{(t)} \coloneq \parens{\vec{s}^{(t,r)}}_{r \in [R]}$. Then,
\begin{equation}
\frac{1}{R} \sum_{r=1}^R \dhamming\parens{\vec{s}^{(t,r)}, \vec{s}^{(t',r)}} \leq \xi n
\iff
\dhamming\parens{\vec{s}^{(t)}, \vec{s}^{(t')}} \leq \xi R n.
\end{equation}
If we fix $\vec{s}^{(t)}$, then the number of $\vec{s}^{(t')}$ that satisfy this is at most
\begin{equation}
\sum_{j=0}^{\floor{\xi R n}} \binom{R n}{j} \leq
R n \binom{R n}{\xi R n} =
\exp\brackets{R n \entropy(\xi) (1 + o_n(1))}.
\end{equation}
Now, considering all $m - 1$ choices of $t'$ and the $2^{R n} = \exp\parens{R n \log 2}$, we have that
\begin{equation}
\abs{\cF(n, m, R, \xi)} \leq \exp\brackets{n \cdot \parens{R \log 2 + (m - 1) R \entropy(\xi)}(1 + o_n(1))}.
\end{equation}
\end{proof}

\subsubsection{Covariance Bound}

Now, take some $\parens{\vec{s}^{(t,r)}}_{t \in [m], r \in [R]} \in \cF(n, m, R, \xi)$. We can write
\begin{equation}
\begin{aligned}
& \Pr_{\bm{X}^{(1)}, \dots, \bm{X}^{(m)} \sim \Xi_{F,m}} \brackets{\min_{t \in [m]} \max_{r \in [R]} V \parens{\bm{X}^{(t)}, \vec{s}^{(t,r)}} \geq \alpha} \leq \\
\leq{} & R^m \max_{(r_t)_{t \in [m]}} \Pr_{\bm{X}^{(1)}, \dots, \bm{X}^{(m)} \sim \Xi_{F,m}} \brackets{\min_{t \in [m]} V \parens{\bm{X}^{(t)}, \vec{s}^{(t,r_t)}} \geq \alpha}.
\end{aligned}
\end{equation}

Now, we have that
\begin{equation}
\EE\brackets{V\parens{\bm{X}^{(t)}, \vec{s}^{(t,r_t)}}} = 0,
\end{equation}

and also, as in the proof of \Cref{thm:lower_bound}, we have
\begin{equation}
\begin{aligned}
& \EE\brackets{V\parens{\bm{X}^{(t)}, \vec{s}^{(t,r_t)}}^2} = \frac{p}{n},
\end{aligned}
\end{equation}
and for $t \neq t'$, we have that
\begin{equation}
\begin{aligned}
& \EE\brackets{V \parens{\bm{X}^{(t)}, \vec{s}^{(t,r_t)}} V \parens{\bm{X}^{(t')}, \vec{s}^{(t',r_{t'})}}} = \\
={} & \frac{k}{n^2} \sum_{A, A' \in \binom{[n]}{k}} \EE\brackets{X^{(t)}_A \bm{X}^{(t')}_{A'}} (-1)^{\sum_{v \in A} \bm{s}^{(t,r_t)}_v + \sum_{v' \in A'} \bm{s}^{(t',r_{t'})}_{v'}} = \\
={} & \frac{k}{n^2} \sum_{A \in \binom{[n]}{k}} \EE\brackets{X^{(t)}_A X^{(t')}_A} (-1)^{\sum_{v \in A} \parens{\bm{s}^{(t,r_t)}_v + \bm{s}^{(t',r_{t'})}_v}} = \\
={} & \frac{k p}{n^2 \binom{n - 1}{k - 1}} \sum_{\substack{A \in \binom{[n]}{k} \\ A \cap [\ceil{F n}] = \varnothing}} (-1)^{\sum_{v \in A} \parens{\bm{s}^{(t,r_t)}_v + \bm{s}^{(t',r_{t'})}_v}} \leq \\
\leq{} & \frac{k p}{n^2 \binom{n - 1}{k - 1}} \binom{\floor{(1 - F)n}}{k} = \\
={} & \frac{k p}{n^2} \frac{(k - 1)!}{n^{k-1}} \frac{(1 - F)^k n^k}{k!} (1 + o_n(1)) = \\
={} & \frac{p}{n} (1 - F)^k (1 + o_n(1)).
\end{aligned}
\end{equation}

Now, consider the $m \times m$ covariance matrix
\begin{equation}
\bm{\Sigma} = \frac{p}{n} \brackets{\parens{1 - (1 - F)^k} \Id_m + (1 - F)^k \bm{J}_m}
\end{equation}
where $\bm{J}_m$ is the $m \times m$ matrix consisting of all ones. We know that the eigenvalues of $\bm{J}_m$ are $m$ (with multiplicity $1$) and $0$ (with multiplicity $m - 1$). So, the eigenvalues of $\bm{\Sigma}$ are:
\begin{equation}
\sigma_1 \coloneq \frac{p}{n} \parens{1 + (m - 1) (1 - F)^k} \qquad \text{(with multiplicity 1)},
\end{equation}
and
\begin{equation}
\sigma_2 \coloneq \frac{p}{n} \parens{1 - (1 - F)^k} \leq \sigma_1 \qquad \text{(with multiplicity $m - 1$)}.
\end{equation}
So, $\bm{\Sigma}$ is positive definite and invertible. We then have that
\begin{equation}
\opnorm{\bm{\Sigma}^{-1/2}} \leq \sigma_2^{-1/2} =
\sqrt{\frac{n}{p\parens{1 - (1 - F)^k}}}.
\end{equation}

We can think of $\braces{\ol{X}_A}_{A \in \binom{[n]}{k}}$ as a set of $\binom{n}{k}$ $m$-dimensional independent vectors with an entry for each $t \in [m]$ corresponding to the contribution to $V \parens{\bm{X}^{(t)}, \vec{s}^{(t,r_t)}}$ from the $X^{(t)}_A$ term. Then, $\bm{\Sigma}$ is the covariance matrix of the sum of these vectors (which is the $m$-dimensional vector of $V \parens{\bm{X}^{(t)}, \vec{s}^{(t,r_t)}}$ over $t \in [m]$). By the Berry-Esseen Theorem~\cite{Rai2019_berry_esseen}, if we take an $m$-dimensional Gaussian vector $G = \parens{G^{(t)}}_{t \in [m]}$ with mean $0$ and covariance matrix $\bm{\Sigma}$, then we have that
\begin{equation}
\begin{aligned}
& \Pr_{\bm{X}^{(1)}, \dots, \bm{X}^{(m)} \sim \Xi_{F,m}} \brackets{\min_{t \in [m]} V \parens{\bm{X}^{(t)}, \vec{s}^{(t,r_t)}} \geq \alpha} \leq \\
\leq{} & \Pr \brackets{\min_{t \in [m]} G^{(t)} \geq \alpha} + C m^{1/4} \sum_{A \in \binom{[n]}{k}} \EE\brackets{\norm{\bm{\Sigma}^{-1/2} \ol{X}_A}_2^3},
\end{aligned}
\end{equation}
where $C$ is some universal constant. Now, we have that
\begin{equation}
\begin{aligned}
& \EE\brackets{\norm{\bm{\Sigma}^{-1/2} \ol{X}_A}_2^3} \leq
\parens{\frac{n}{p \parens{1 - (1 - F)^k}}}^{3/2} \EE\brackets{\norm{\ol{X}_A}_2^3} \leq \\
\leq{} & \parens{\frac{n}{p \parens{1 - (1 - F)^k}}}^{3/2} \frac{\sqrt{k}}{n} \EE\brackets{\norm{\ol{X}_A}_2^2} = \\
={} & \frac{\sqrt{k n}}{\parens{p \parens{1 - (1 - F)^k}}^{3/2}} \sum_{t \in [m]} \EE\brackets{\parens{\ol{X_A}}_{t}^2} = \\
={} & \frac{\sqrt{k n}}{\parens{p\parens{1 - p (1 - F)^k}}^{3/2}} \frac{m p}{\binom{n - 1}{k - 1}} \frac{k}{n^2} = \\
={} & \parens{\frac{k}{n \parens{1 - (1 - F)^k}}}^{3/2} \frac{m}{\sqrt{p} \binom{n - 1}{k - 1}}.
\end{aligned}
\end{equation}
Thus,
\begin{equation}
\begin{aligned}
& \Pr_{\bm{X}^{(1)}, \dots, \bm{X}^{(m)} \sim \Xi_{F,m}} \brackets{\min_{t \in [m]} V \parens{\bm{X}^{(t)}, \vec{s}^{(t,r_t)}} \geq \alpha} \leq \\
\leq{} & \Pr \brackets{\min_{t \in [m]} G^{(t)} \geq \alpha} + C m^{1/4} \binom{n}{k} \parens{\frac{k}{n \parens{1 - (1 - F)^k}}}^{3/2} \frac{m}{\sqrt{p} \binom{n - 1}{k - 1}} = \\
={} & \Pr \brackets{\min_{t \in [m]} G^{(t)} \geq \alpha} + \frac{C m^{5/4}}{\sqrt{p}} \frac{n}{k} \parens{\frac{k}{n \parens{1 - (1 - F)^k}}}^{3/2} = \\
={} & \Pr \brackets{\min_{t \in [m]} G^{(t)} \geq \alpha} + \frac{C m^{5/4}}{\parens{1 - (1 - F)^k}^{3/2}} \sqrt{\frac{k}{p n}} = \\
={} & \Pr \brackets{\min_{t \in [m]} G^{(t)} \geq \alpha} + o_n (1).
\end{aligned}
\end{equation}

Now that we have a bound in terms of Gaussians, we can write, up to $o_n (1)$ terms,
\begin{equation}
\begin{aligned}
& \Pr_{\bm{X}^{(1)}, \dots, \bm{X}^{(m)} \sim \Xi_{F,m}} \brackets{\min_{t \in [m]} \max_{r \in [R]} V \parens{\bm{X}^{(t)}, \vec{s}^{(t,r)}} \geq \alpha} \leq \\
\leq{} & R^m \Pr \brackets{\min_{t \in [m]} G^{(t)} \geq \alpha} \leq \\
\leq{} & R^m \Pr \brackets{\sum_{t \in [m]} G^{(t)} \geq m \alpha} \leq \\
\leq{} & \exp\brackets{m \log R - \frac{\frac{1}{2} m^2 \alpha^2}{\frac{m p}{n} + \frac{m (m - 1) p}{n} (1 - F)^k}} = \\
={} & \exp\brackets{-\frac{\frac{1}{2 p} n m \alpha^2}{1 + (m - 1)(1 - F)^k}}(1 + o_n (1)).
\end{aligned}
\end{equation}

So now, by the first moment method, we have that
\begin{equation}
\begin{aligned}
& \Pr_{\bm{X}^{(1)}, \dots, \bm{X}^{(m)} \sim \Xi_{F,m}} \brackets{\max_{\parens{\vec{s}^{(t,r)}}_{t \in [m], r \in [R]} \in \cF(n, m, R, \xi)} \min_{t \in [m]} \max_{r \in [R]} V \parens{\bm{X}^{(t)}, \vec{s}^{(t,r)}} \geq \alpha} \leq \\
\leq{} & \exp\brackets{n \parens{R \log 2 + (m - 1) R \entropy(\xi) - \frac{\frac{1}{2 p} m \alpha^2}{1 + (m - 1) (1 - F)^k}} (1 + o_n(1))}.
\end{aligned}
\end{equation}

If the expression above is $\geq 2^{-n}$, we must have that
\begin{equation}
\alpha^2 \leq \frac{2 p}{m} \parens{(R + 1) \log 2 + (m - 1) R \entropy(\xi)} \parens{1 + (m - 1) (1 - F)^k}.
\end{equation}
Note that an extra $\log 2$ term is added to account for the desired $2^{-n}$ bound. Now, if we pick
\begin{equation}
\beta = \sqrt{\frac{m}{0.99 - m \probest^R}}
\end{equation}
to satisfy $m / \beta^2 + m \probest^R < 1$, and
\begin{equation}
\xi = \frac{2 \beta f}{n} + 6 \beta L p F =
\sqrt{\frac{m}{0.99 - m \probest^R}}\parens{\frac{2 f}{n} + 6 L p F},
\end{equation}
then, substituting $\probest$ from \Cref{lem:efficient_shadows}, we have (assuming that $f = o(n)$),
\begin{equation}
\entropy(\xi) \leq \sqrt{2\xi} =
2 \parens{\frac{m}{0.99 - \frac{m}{\parens{1 + 0.99 \cdot \delta^2 / 9}^R}}}^{1/4} \sqrt{3 L p F} (1 + o_n (1)).
\end{equation}

\subsubsection{Completing the Proof}

We now want to get a contradiction with \Cref{lem:ogp_existence}. To do this, we can show that for any $0 < \gamma \leq 1$ and $L \geq 0$, there exists some choice of $R, m, F, k, \delta$ such that
\begin{equation}
\begin{aligned}
& \frac{2 p}{m} \parens{(R + 1) \log 2 + 2 (m - 1) R \parens{\frac{m}{0.99 - \frac{m}{\parens{1 + 0.99 \cdot \delta^2 / 9}^R}}}^{1/4} \sqrt{3 L p F}} \cdot \\
& \hspace{20em} \cdot \parens{1 + (m - 1) (1 - F)^k} < \\
<{} & \frac{2 p}{9} (\gamma - \delta)^2 \log 2 \leq \parens{\frac{1}{3}(\gamma - \delta) E_\ast^{(k)}}^2,
\end{aligned}
\end{equation}
that is,
\begin{equation}
\begin{aligned}
& \frac{1}{m} \parens{(R + 1) \log 2 + \frac{2 (m - 1) R \sqrt{3 L p F}}{\parens{\frac{0.99}{m} - \frac{1}{\parens{1 + 0.99 \cdot \delta^2 / 9}^R}}^{1/4}}} \parens{1 + (m - 1) (1 - F)^k} < \\
<{} & \frac{1}{9} (\gamma - \delta)^2 \log 2.
\end{aligned}
\end{equation}










To accomplish this, it is sufficient to first take $\delta = \gamma / 2$ and pick $R, m, F, k$ to satisfy the following conditions:

\begin{equation}
\frac{1}{m} (R + 1) \log 2 \leq \frac{1}{108} \gamma^2 \log 2
\end{equation}
\begin{equation}
\frac{2 R \sqrt{3 L p F}}{\parens{\frac{0.99}{m} - \frac{1}{\parens{1 + 0.99 \cdot \gamma^2 / 36}^R}}^{1/4}} \leq \frac{1}{108} \gamma^2 \log 2
\end{equation}
\begin{equation}
1 + (m - 1) (1 - F)^k \leq \frac{3}{2}
\end{equation}
Note that $(\gamma - \delta)^2 = \gamma^2 / 4$ and we remove the factor of $(m - 1) / m$ in the second condition.

To satisfy the first condition, we must have
\begin{equation}
m \geq \ceil{\frac{108}{\gamma^2} (R + 1)}.
\end{equation}
Suppose we set $m$ to be exactly that value. Now, we want to set $R$ so that the denominator in the second condition is well-defined and positive. We want
\begin{equation}
\frac{0.99}{m} > \frac{1}{(1 + 0.99 \cdot \gamma^2 / 36)^R}
\end{equation}
\begin{equation}
\parens{1 + \frac{0.99 \cdot \gamma^2}{36}}^R > \frac{m}{0.99},
\end{equation}
so it is sufficient to show that
\begin{equation}
\parens{1 + \frac{0.99 \cdot \gamma^2}{36}}^R > \frac{\frac{108}{\gamma^2} (R + 1) + 1}{0.99} =
\frac{108}{0.99 \cdot \gamma^2} (R + 1) + 1 / 0.99
\end{equation}
If we let $z = 0.99 \cdot \gamma^2 / 36$, this is the same as
\begin{equation}
(1 + z)^R > \frac{3 (R + 1)}{z} + 1/0.99.
\end{equation}
We can check that this is satisfied when $R \geq 1 / z^2$.

Now, we can pick $F$ to be sufficiently small such that the second condition is satisfied. Lastly, we can pick $k$ to be sufficiently large so that the third condition is satisfied.

In particular, if we set $\gamma = 1$, we can pick the following values. We can set $R = 500$, so $m = \num{54108}$. Then, we must take $F \leq \frac{\num{3.97e-15}}{L p}$. If we set $k = 12 / F$, then we have $(1 - F)^k \leq e^{-12} \leq \frac{1}{2 (m - 1)}$.

Thus, we have the following result:

\begin{theorem}[Weak hardness for stable algorithms]\label{thm:weak_hardness}
Take any constants $p > 0$, $L > 0$, $0 < \gamma \leq 1$ and $f = o(n)$. Then, for some sufficiently small $F$ and sufficiently large $k$, there exists no $\parens{f, L, \frakd, \Xi_F, \probst}$-stable and $(\gamma, \probf)$-optimal quantum algorithm for the Quantum Hypergraph Max-Cut problem for sufficiently small $\probst, \probf \geq 0$ (depending only on $\gamma$) and $\frakd \geq \max\parens{\ceil{e^2 p}, 1}$. In particular, when $\gamma = 1$, we can take $\probst = \probf = \num{1e-6}$, set $F = \frac{\num{3.97e-15}}{L p}$ and $k = \num{3.03e15} L p \geq 12 / F$.
\end{theorem}

\begin{proof}
From \Cref{lem:efficient_shadows}, we know that $\probb = e^{-\Omega(n)}$. Since $m$ in the calculation above can be chosen to be sufficiently large depending only on $\gamma$, we can take sufficiently small $\probst, \probf \geq 0$ such that
\begin{equation}
\begin{aligned}
& 1 - 3 m \probst - 3 m \probf - m \probb - m e^{-3 p F n (3 \log 2 - 2)} \geq 0.
\end{aligned}
\end{equation}
We can check that for $\gamma = 1$, this is satisfied when $\probst = \probf = \num{1e-6}$.

By \Cref{cor:constant_degree_prob}, if we set $\epsilon = \log 2$, we have that $2 \log(1 / \epsilon) < 1$, so if we have $\frakd \geq \max\parens{\ceil{e^2 p}, 1}$, then we must have
\begin{equation}
\Pr\brackets{\cXfd(\bm{S})} \geq 2^{-n}.
\end{equation}
So, assuming that a stable and near-optimal algorithm as in the statement exists, by \Cref{lem:ogp_existence}, with probability $(1 + o_n (1)) 2^{-n}$  there exist bit strings $\vec{s}^{(t,r)}$ for $t \in [m]$, $r \in [R]$ such that for all $t \neq t' \in [m]$
\begin{equation}
\frac{1}{R n} \sum_{r=1}^R \dhamming\parens{\vec{s}^{(t,r)}, \vec{s}^{(t',r)}} \leq \frac{2 \beta f}{n} + 9 \beta L p F,
\end{equation}
and for all $t \in [m]$, if $\gamma = 1/2$ and $\delta = 1/4$, then
\begin{equation}
\max_{r \in [R]} V \parens{\bm{X}^{(t)}, \vec{s}^{(t,r)}} \geq \frac{1}{3}\parens{\gamma - \delta}E_\ast \geq \frac{1}{12} \sqrt{2 p \log 2},
\end{equation}
but by the calculation above, this can only happen with probability less than \\ $(1 + o_n (1)) 2^{-n}$. Thus, we have a contradiction and so no such algorithm can exist.
\end{proof}

\section{Strong Hardness for Local Algorithms}\label{chp:strong_hardness}

In the previous section, we were able to show a statement roughly of the form ``for any $L$, there exists a $k$ such that all $L$-stable algorithms fail''. However, ideally, we would like to make a stronger statement, of the form ``there exists a fixed value of $k$ such that \emph{no} stable algorithm is optimal''. Unfortunately, we are not able to show a sufficiently strong version of the Quantum Overlap Gap property for Quantum Hypergraph Max-Cut to establish a statement of this form. However, such a statement \emph{can} be made if we strengthen our notion of stability and consider ``local algorithms.'' Instead of using the quantum Wasserstein distance of order $2$ ($W_2$) as our distance metric, this notion of stability will use $W_\infty$:

\begin{definition}[Local algorithm]\label{def:local_alg}
Let $\Xi$ be a probability distribution over pairs $(\bm{X}, \bm{X}')$ of problem instances that correspond to the same interaction hypergraph. Then, we say that a quantum algorithm $\bm{\cA}$ is \emph{$(f, L, \frakd, \Xi, \probst)$-local} if
\begin{equation}
\Pr_{(\bm{X}, \bm{X}', \omega) \sim \Xi \otimes \PP_\Omega} \brackets{\norm{\bm{\cA}(\bm{X}, \omega) - \bm{\cA}(\bm{X}', \omega)}_{W_\infty} \leq f + L \norm{\bm{X} - \bm{X}}_1 \mid \cXfd(\bm{X})} \geq 1 - \probst.
\end{equation}
\end{definition}

\subsection{Local Algorithms Imply the Existence of Forbidden Configurations}

\subsubsection{Reduction to Weakly Local Algorithms with Classical Outputs}

The beginning of this proof proceeds similarly to \Cref{chp:weak_hardness}. We do a reduction to deterministic algorithms and then algorithms with classical outputs that are ``weakly local''. However, the properties of the $W_\infty$ distance allow us to derive much stronger results from the ``weakly local'' condition than from the ``weakly stable'' one.

\begin{lemma}\label{lem:local_alg_deterministic_reduction}
Let $\parens{\Xi^{(q)}}_{q \in [Q]}$ be a set of probability distributions of pairs of problem instances whose marginal distributions are all $\Lambda$. Suppose that there exists a quantum algorithm $\bm{\cA}(\bm{X}, \omega)$ that is $\parens{f, L, \frakd, \Xi^{(q)}, \probst}$-local for all $q \in [Q]$ and $(\gamma, \probf)$-optimal for $\Lambda$. Then, there exists a deterministic quantum algorithm $\widetilde{\bm{\cA}}(\bm{X})$ that is $\parens{f, L, \frakd, \Xi, 3 Q \probst}$-local and $(\gamma, 3 \probf)$-optimal for $\Lambda$.
\end{lemma}

\begin{proof}
The proof parallels that of \Cref{lem:deterministic_reduction}, with an additional union bound. We have that
\begin{equation}
\begin{aligned}
& \EE_{\omega \sim \PP_\Omega} \brackets{\Pr_{(\bm{X}, \bm{X}') \sim \Xi^{(q)}} \brackets{\norm{\bm{\cA}(\bm{X}, \omega) - \bm{\cA}(\bm{X}', \omega)}_{W_\infty} > f + L \norm{\bm{X} - \bm{X}}_1 \mid \cXfd(\bm{X})}} = \\
={} & \Pr_{(\bm{X}, \bm{X}', \omega) \sim \Xi \otimes \PP_\Omega} \brackets{\norm{\bm{\cA}(\bm{X}, \omega) - \bm{\cA}(\bm{X}', \omega)}_{W_2} > f + L \norm{\bm{X} - \bm{X}}_1 \mid \cXfd(\bm{X})} \leq \probst,
\end{aligned}
\end{equation}
and
\begin{equation}
\begin{aligned}
& \EE_{\omega \sim \PP_\Omega} \brackets{\Pr_{\bm{X} \sim \Lambda} \brackets{\Tr\parens{\bm{H}(\bm{X}) \bm{\cA}(\bm{X}, \omega)} < \gamma E_\ast}} = \\
={} & \Pr_{(\bm{X}, \omega) \sim \Lambda \otimes \PP_\Omega} \brackets{\Tr\parens{\bm{H}(\bm{X}) \bm{\cA}(\bm{X}, \omega)} < \gamma E_\ast} \leq \probf.
\end{aligned}
\end{equation}
So, by Markov's inequality, we have that
\begin{equation}
\begin{aligned}
& \Pr_{\omega \sim \PP_\Omega} \Bigg[\Pr_{(\bm{X}, \bm{X}') \sim \Xi^{(q)}} \Big[\norm{\bm{\cA}(\bm{X}, \omega) - \bm{\cA}(\bm{X}', \omega)}_{W_2} > \\
& \hspace{10em} > f + L \norm{\bm{X} - \bm{X}}_1 \mid \cXfd(\bm{X})\Big] > 3 Q \probst \Bigg] \leq \\
\leq{} & \frac{1}{3 Q \probst} \EE_{\omega \sim \PP_\Omega} \Bigg[\Pr_{(\bm{X}, \bm{X}') \sim \Xi^{(q)}} \Big[\norm{\bm{\cA}(\bm{X}, \omega) - \bm{\cA}(\bm{X}', \omega)}_{W_\infty} > \\
& \hspace{10em} > f + L \norm{\bm{X} - \bm{X}}_1 \mid \cXfd(\bm{X})\Big] \Bigg] \leq \frac{1}{3 Q},
\end{aligned}
\end{equation}
and
\begin{equation}
\begin{aligned}
& \Pr_{\omega \sim \PP_\Omega} \brackets{\Pr_{\bm{X} \sim \Lambda} \brackets{\Tr\parens{\bm{H}(\bm{X}) \bm{\cA}(\bm{X}, \omega)} < \gamma E_\ast} > 3 \probf} \leq \\
\leq{} & \frac{1}{3 \probf} \EE_{\omega \sim \PP_\Omega} \brackets{\Pr_{\bm{X} \sim \Lambda} \brackets{\Tr\parens{\bm{H}(\bm{X}) \bm{\cA}(\bm{X}, \omega)} < \gamma E_\ast}} \leq \frac{1}{3}.
\end{aligned}
\end{equation}
So now, by the union bound, we have that
\begin{equation}
\begin{aligned}
& \Pr_{\omega \sim \PP_\Omega} \Bigg[\Pr_{\bm{X} \sim \Lambda} \brackets{\Tr\parens{\bm{H}(\bm{X}) \bm{\cA}(\bm{X}, \omega)} < \gamma E_\ast} > 3 \probf \vee \\
& \vee \bigvee_{q \in [Q]} \Pr_{(\bm{X}, \bm{X}') \sim \Xi^{(q)}} \Big[\norm{\bm{\cA}(\bm{X}, \omega) - \bm{\cA}(\bm{X}', \omega)}_{W_\infty} > \\
& \hspace{10em} > f + L \norm{\bm{X} - \bm{X}}_1 \mid \cXfd(\bm{X})\Big] > 3 \probst \Bigg] \leq \frac{2}{3}.
\end{aligned}
\end{equation}
Thus, there must exist some $\omega^\ast \in \Omega$ such that for all $q \in [Q]$ we have
\begin{equation}
\begin{aligned}
& \Pr_{(\bm{X}, \bm{X}') \sim \Xi^{(q)}} \brackets{\norm{\bm{\cA}(\bm{X}, \omega^\ast) - \bm{\cA}(\bm{X}', \omega^\ast)}_{W_\infty} \leq f + L \norm{\bm{X} - \bm{X}}_1 \mid \cXfd(\bm{X})} \geq \\
\geq{} & 1 - 3 Q \probst,
\end{aligned}
\end{equation}
and
\begin{equation}
\Pr_{\bm{X} \sim \Lambda} \brackets{\Tr\parens{\bm{H}(\bm{X}) \bm{\cA}(\bm{X}, \omega^\ast)} \geq \gamma E_\ast} \geq 1 - 3 \probf.
\end{equation}
Thus, if we define $\widetilde{\bm{\cA}}(\bm{X}) \coloneq \bm{\cA}(\bm{X}, \omega^\ast)$, then $\widetilde{\bm{\cA}}$ is a deterministic quantum algorithm that is $\parens{f, L, \Xi, 3 \probst}$-stable and $(\gamma, 3 \probf)$-optimal.
\end{proof}

\begin{lemma}\label{lem:local_alg_nondeterministic_shadows_reduction}
Let $\parens{\Xi^{(q)}}_{q \in [Q]}$ be a probability distribution of pairs of problem instances whose marginal distributions are all $\Lambda$. Suppose that the Pauli shadows estimator is $(\delta, \probest, \probb)$-efficient for the class of Hamiltonians $\bm{H}(\Lambda)$. Suppose there exists a deterministic quantum algorithm $\widetilde{\bm{\cA}}$ that is $\parens{f, L, \frakd, \Xi^{(q)}, 3 Q \probst}$-local for all $q \in [Q]$ and $(\gamma, 3 \probf)$-optimal for $\Lambda$. Then, there exists a pure quantum algorithm $\bm{\cG} : \RR^{\binom{n}{k}} \times [0, 1] \rarr \{0, 1\}^n$ (that is, a randomized algorithm outputting only classical states), such that the following weaker notion of locality holds for all $q \in [Q]$:
\begin{equation}\label{eq:local_alg_shadows_reduction_weak_stability}
\begin{aligned}
& \Pr_{(\bm{X}, \bm{X}') \sim \Xi} \Bigg[\norm{\EE_{\omega \sim \cU}\brackets{\selfouterprod{\bm{\cG}(\bm{X}, \omega)} - \selfouterprod{\bm{\cG}(\bm{X}', \omega)}}}_{W_\infty} \leq \\
& \hspace{10em} \leq f + L\norm{\bm{X} - \bm{X}'}_1 \mid \cXfd(\bm{X})\Bigg] \geq 1 - 3 Q \probst.
\end{aligned}
\end{equation}
Furthermore, consider some $\bm{X}$ that satisfies the following two conditions (which is true with probability at least $1 - 3 \probf - \probb$ over $\bm{X} \sim \Lambda$):
\begin{equation}\label{eq:local_alg_shadows_reduction_condition_1}
\Tr\parens{\bm{H}(\bm{X}) \widetilde{\bm{\cA}}(\bm{X})} \geq \gamma E_\ast
\end{equation}
\begin{equation}\label{eq:local_alg_shadows_reduction_condition_2}
\Pr_{\omega \sim \cU} \brackets{\Tr\parens{3 \bm{H}(\bm{X}) \widetilde{\bm{\cM}}(\bm{\rho}, \omega)} - \Tr\parens{\bm{H}(\bm{X})\bm{\rho}} \geq -\delta E_\ast} \geq 1 - \probest \quad \forall \bm{\rho} \in \mixedn.
\end{equation}
For such $\bm{X}$, we must have that
\begin{equation}\label{eq:local_alg_shadows_reduction_near_optimality}
\Pr_{\omega \sim \cU} \brackets{V \parens{\bm{X}, \bm{\cG}(\bm{X}, \omega)} \geq \frac{1}{3}(\gamma - \delta) E_\ast} \geq 1 - \probest.
\end{equation}
\end{lemma}

\begin{proof}
The proof is identical to that of \Cref{lem:nondeterministic_shadows_reduction}.
\end{proof}

\subsubsection{Solution Configurations over Interpolation Paths}

\begin{proposition}[Stronger version of \Cref{prop:expected_wasserstein_distance} for $W_\infty$]\label{prop:max_wasserstein_infinity_distance}
For each problem instance $\bm{X}$, let $p_{\bm{X}}(\vec{s})$ be the probability distribution of $\bm{\cG}(\bm{X}, \omega) \in \{0, 1\}^n$ over $\omega \sim \cU$. Then, for every pair of problem instances $\bm{X}, \bm{X}'$, there exists a joint distribution $\pi_{\bm{X}, \bm{X}'}$ such that
\begin{equation}
\begin{aligned}
& \max_{(\vec{s}, \vec{s}') \in \supp\parens{\pi_{\bm{X}, \bm{X}'}}} \dhamming(\vec{s}, \vec{s}') \leq \\
\leq{} & \norm{\EE_{\omega \sim \cU} \brackets{\selfouterprod{\bm{\cG}(\bm{X}, \omega)} - \selfouterprod{\bm{\cG}(\bm{X}', \omega)}}}_{W_\infty}.
\end{aligned}
\end{equation}
\end{proposition}

\begin{proof}
Write
\begin{equation}
\bm{A} \coloneq \EE_{\omega \sim \cU} \brackets{\selfouterprod{\bm{\cG}(\bm{X}, \omega)} - \selfouterprod{\bm{\cG}(\bm{X}', \omega)}} = \sum_{i=1}^n \bm{A}_i,
\end{equation}
where the $\bm{A}_i$ are Hermitian such that $\Tr_{\{i\}} \parens{\bm{A}_i} = 0$, picked to be the optimal in \Cref{def:quantum_wasserstein}. That is, suppose (without loss of generality) that $\norm{\bm{A}_i}_\ast \neq 0$ only for $i \leq W$, where
\begin{equation}
W = \norm{\EE_{\omega \sim \cU} \brackets{\selfouterprod{\bm{\cG}(\bm{X}, \omega)} - \selfouterprod{\bm{\cG}(\bm{X}', \omega)}}}_{W_\infty}.
\end{equation}
Since, as in \Cref{prop:expected_wasserstein_distance}, we can write
\begin{equation}
\norm{\frac{1}{2} \bm{A}_i}_\ast = \Tr\parens{\bm{A}_i^+} = \Tr\parens{\bm{A}_i^-},
\end{equation}
we must have that $\bm{A}_i = 0$ for $i > W$. Now, we can think of $\bm{A}$ as specifying a flow on a hypercube, with each $\bm{A}_i$ representing a flow along the $i$'th dimension, with the positive eigenvalues corresponding to sinks and negative eigenvalues to sources. Then, we can construct the joint distribution $\pi$ to describe the net flow from the sum of these contributions, and observe that $\pi(\bm{s}, \bm{s}') = 0$ for any $\bm{s}, \bm{s}'$ with $\bm{s}_i \neq \bm{s}'_i$ for any $i > W$ since there is no flow along dimension $i$ of the hypercube for $i > W$. Thus,
\begin{equation}
\begin{aligned}
& \max_{(\vec{s}, \vec{s}') \in \supp\parens{\pi_{\bm{X}, \bm{X}'}}} \dhamming(\vec{s}, \vec{s}') \leq \\
\leq{} & W = \norm{\EE_{\omega \sim \cU} \brackets{\selfouterprod{\bm{\cG}(\bm{X}, \omega)} - \selfouterprod{\bm{\cG}(\bm{X}', \omega)}}}_{W_\infty}.
\end{aligned}
\end{equation}
\end{proof}

\begin{definition}[Distribution of problem instances for interpolation path]\label{def:local_alg_interpolation_path}
\phantom{} \newline For $T, Q \in \NN$, we define the distribution $\Xi_{Q,T}$ of collections $\parens{\bm{X}^{(q,t)}}_{0 \leq q \leq Q, t \in [T]}$ as follows. First, sample $\bm{S}$ by setting each $S_A \sim \Bern\parens{p / \binom{n - 1}{p - 1}}$ i.i.d. Then, sample $\hat{\bm{J}}, \widetilde{\bm{J}}^{(1)}, \dots, \widetilde{\bm{J}}^{(T)}$ with i.i.d. entries in $\{-1, 1\}$. Then, for all and $A \in \binom{[n]}{k}$, set
\begin{equation}
X^{(q,t)}_A = \begin{cases}
S_A \widetilde{J}^{(t)}_A & \text{ if } A \cap \brackets{\ceil{q n / Q}} \neq \varnothing \\
S_A \hat{J}_A & \text{ otherwise. }
\end{cases}
\end{equation}
In other words, we are interpolating between $\hat{\bm{X}}$ and $T$ other independent problem instances by resampling the weights of hyperedges intersecting with the first $q / Q$ fraction of qubits. In contrast to \Cref{def:single_step_interpolation_path}, where we only considered a single interpolation step, here we are considering the whole path between two independent instances. We call the marginal distribution of $\parens{\bm{X}^{(q-1,t)}, \bm{X}^{(q,t)}}$ for a fixed $q,t$ by the name $\Xi_Q^{(q)}$. Observe that these marginal distributions are independent of $t$. We also have that the marginal distribution of $\parens{\bm{X}^{(q,t)}}_{t \in [T]}$ for a fixed $q$ is the same as $\Xi_{\frac{q}{Q},T}$ as in \Cref{def:single_step_interpolation_path}. Also, observe that when sampling from this distribution, we will always have $\bm{X}^{(0,t)} = \bm{X}^{(0,t')} = \hat{\bm{X}}$ for all $t, t' \in [T]$.
\end{definition}

\begin{figure}[htpb]
    \centering
    \includegraphics{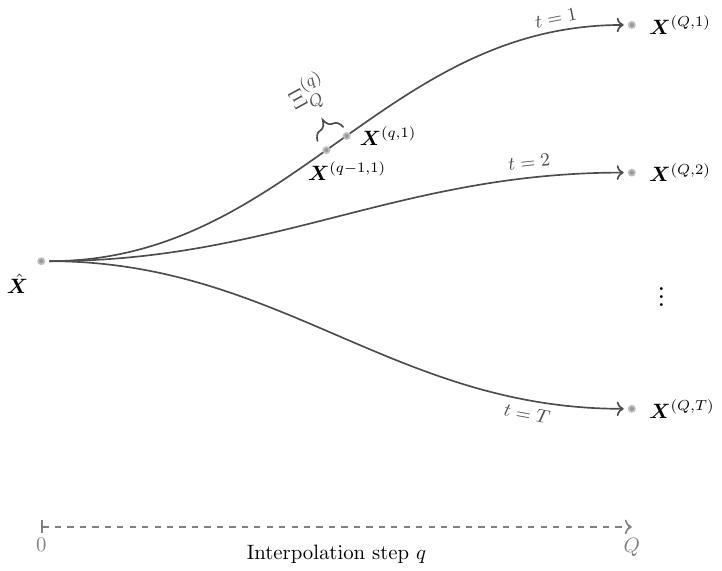}
    \caption[Illustration of the interpolation path]{Diagram illustrating the distribution $\Xi_{Q,T}$, showing interpolation paths branching from a single initial problem instance $\hat{\bm{X}}$ into $T$ independent instances. The diagram also shows the marginal distribution $\Xi_Q^{(q)}$ corresponding to a single step along the interpolation path.}
    \Description{}
    \label{fig:interpolation}
\end{figure}

\begin{lemma}\label{lem:local_alg_deterministic_shadows_reduction}
Suppose that there exists a quantum algorithm $\bm{\cA}$ that is \\ $\parens{f, L, \frakd, \Xi_Q^{(q)}, \probst}$-local for all $0 \leq q \leq Q$ and $(\gamma, \probf)$-optimal for $\QMC(n, k, p)$. Consider some $T \in \NN$. Then, with probability at least $1 - 3 Q^2 T \probst - 3 Q T \probf - Q T \probb$ over $\Xi_{Q,T}$, conditioned on $\cXfd\parens{\hat{\bm{X}}}$, there exists a set of bit strings \\ $\parens{\vec{s}^{(q,t)} \in \{0, 1\}^n}_{0 \leq q \leq Q, t \in [T]}$ such that $\bm{s}^{(0,t)} = \bm{s}^{(0,t')}$ for $t, t' \in [T]$ and for all $q \in [Q]$, $t \in [T]$ we have that
\begin{equation}
\dhamming\parens{\vec{s}^{(q-1, t)}, \vec{s}^{(q, t)}} \leq f + L \norm{\bm{X}^{(q-1, t)} - \bm{X}^{(q, t)}}_1,
\end{equation}
and
\begin{equation}
V\parens{\bm{X}^{(q,t)}, \vec{s}^{(q,t)}} \geq \frac{1}{3}(\gamma - \delta) E_\ast.
\end{equation}
\end{lemma}

\begin{proof}
Consider some $\parens{\bm{X}^{(q,t)}}_{0 \leq q \leq Q, t \in [T]}$ that satisfy the following conditions for all $q \in [Q], t \in [T]$:
\begin{equation}
\norm{\widetilde{\bm{\cA}}\parens{\bm{X}^{(q-1,t)}} - \widetilde{\bm{\cA}}\parens{\bm{X}^{(q,t)}}}_{W_\infty} \leq f + L \norm{\bm{X}^{(q-1,t)} - \bm{X}^{(q,t)}}_1,
\end{equation}
\begin{equation}
\Tr\parens{\bm{H}(\bm{X}^{(q,t)}) \widetilde{\bm{\cA}}\parens{\bm{X}^{q,t}}} \geq \gamma E_\ast,
\end{equation}
\begin{equation}
\Pr_{\omega \sim \cU} \brackets{\Tr\parens{3 \bm{H}\parens{\bm{X}^{(q,t)}} \widetilde{\bm{\cM}}(\bm{\rho}, \omega)} - \Tr\parens{\bm{H}\parens{\bm{X}^{(q,t)}}\bm{\rho}} \geq -\delta E_\ast} \geq 1 - \probest \; \forall \bm{\rho} \in \mixedn.
\end{equation}
By the union bound, the probability over $\Xi$ that all of these conditions occur is at least
\begin{equation}
1 - 3 Q^2 T \probst - 3 Q T \probf - Q T \probb.
\end{equation}

Let $p^{(q,t)}$ be the probability distribution of $\bm{\cG}\parens{\bm{X}^{(q,t)}, \omega}$, over $\omega \sim \cU$. Let $\pi^{(q,t)}$ be the joint distribution with marginals $p^{(q-1,t)}, p^{(q,t)}$ from \Cref{prop:max_wasserstein_infinity_distance}.

Then, consider the distribution $\Pi$ of collections of bit strings $\parens{\vec{s}^{(q, t)}}_{0 \leq q \leq Q, t \in [T]}$ formed by the following sampling procedure. First, pick $\vec{s}^{(0, 1)} \sim p^{(0, 1)}$ and set $\vec{s}^{(0, t)} = \vec{s}^{(0, 1)}$ for all $t \in [T]$ (note that $p^{(0, t)} = p^{(0, t')}$ for all $t, t'$). Then, for each $q, t$ pick $\vec{s}^{(q, t)}$ according to the distribution $\pi^{(q,t)}$ conditioned on $\vec{s}^{(q-1,t)}$. Keep resampling $\vec{s}^{(q-1,t)}$ until $V\parens{\bm{X}^{(q,t)}, \vec{s}^{(q,t)}} \geq \frac{1}{3}(\gamma - \delta) E_\ast$. Observe that when sampling from $\Pi$, we have that each $\parens{\vec{s}^{(q-1,t)}, \vec{s}^{(q,t)}} \in \supp\parens{\pi^{(q,t)}}$.

From \Cref{lem:local_alg_nondeterministic_shadows_reduction}, for our choice of $\parens{\bm{X}^{(q,t)}}_{0 \leq q \leq Q, t \in [T]}$, we know that
\begin{equation}
\begin{aligned}
& \Pr_{\vec{s}^{(q,t)} \sim p^{(q,t)}} \brackets{V\parens{\bm{X}^{(q,t)}, \vec{s}^{(q,t)}} < \frac{1}{3}(\gamma - \delta) E_\ast} = \\
={} & \Pr_{\omega \sim \cU} \brackets{V \parens{\bm{X}^{(q,t)}, \bm{\cG}(\bm{X}^{(q,t)}, \omega)} < \frac{1}{3}(\gamma - \delta) E_\ast} \leq \probest < 1,
\end{aligned}
\end{equation}
and since the marginal distribution of $\pi^{(q,t)}$ conditioned on $\vec{s}^{(q-1,t)}$ is $p^{(q,t)}$, the above sampling procedure must eventually produce a valid $\vec{s}^{(q,t)}$ satisfying
\begin{equation}
V\parens{\bm{X}^{(q,t)}, \vec{s}^{(q,t)}} \geq \frac{1}{3}(\gamma - \delta) E_\ast.
\end{equation}

Now, by \Cref{prop:max_wasserstein_infinity_distance} and \Cref{lem:local_alg_nondeterministic_shadows_reduction}, we have that
\begin{equation}
\begin{aligned}
& \dhamming\parens{\vec{s}^{(q-1, t)}, \vec{s}^{(q, t)}} \leq \\
\leq{} & \norm{\EE_{\omega \sim \cU} \brackets{\selfouterprod{\bm{\cG}(\bm{X}, \omega)} - \selfouterprod{\bm{\cG}(\bm{X}', \omega)}}}_{W_\infty} \\
\leq{} & f + L \norm{\bm{X}^{(q-1, t)} - \bm{X}^{(q, t)}}_1,
\end{aligned}
\end{equation}
which proves the desired claim.
\end{proof}

\begin{lemma}\label{lem:local_alg_stability_between_replicas}
Suppose that there exists a quantum algorithm $\bm{\cA}$ that is \\ $\parens{f, L, \Xi_Q^{(q)}, \probst}$-local for all $0 \leq q \leq Q$ and $(\gamma, \probf)$-optimal for $\QMC(n, k, p)$.
Then, with probability at least
\begin{equation}
1 - 3 Q^2 T \probst - 3 Q T \probf - Q T \probb - T \exp\parens{-3 (3 \log 2 - 2) \frac{p n}{Q}}
\end{equation}
over $\parens{\bm{X}^{(q,t)}}_{0 \leq q \leq Q, t \in [T]} \sim \Xi_{Q,T}$, conditioned on $\cXfd\parens{\hat{\bm{X}}}$, there exist bit strings $\vec{s}^{(q,t)}$ for $0 \leq q \leq [Q], t \in [m]$ such that for all $q \in [Q]$, $t \neq t' \in [m]$ we have
\begin{equation}
\abs{\dhamming\parens{\vec{s}^{(q-1,t)} - \vec{s}^{(q-1,t')}} - \dhamming\parens{\vec{s}^{(q,t)} - \vec{s}^{(q,t')}}} \leq \frac{2 f}{n} + \frac{18 L p}{Q},
\end{equation}
and for all $q \in [Q]$, $t \in [m]$,
\begin{equation}
V \parens{\bm{X}^{(q,t)}, \vec{s}^{(q,t)}} \geq \frac{1}{3}\parens{\gamma - \delta}E_\ast.
\end{equation}
\end{lemma}

\begin{proof}
Observe that $X^{(q-1,t)}_A$ can be different from $X^{(q,t)}_A$ only if we have \\ $A \cap [\ceil{qn/Q}] \neq \varnothing$ but $A \cap [\ceil{(q-1)n/Q}] = \varnothing$. Since $A$ must contain one of $n/Q$ many qubits, the number of such $A$ is at most
\begin{equation}
\frac{n}{Q} \binom{n - 1}{k - 1}.
\end{equation}
Then,
\begin{equation}
\frac{p}{\binom{n - 1}{k - 1}} \frac{n}{Q} \binom{n - 1}{k - 1} = \frac{p n}{Q}.
\end{equation}
So, by \Cref{lem:l1_bound}, we have that
\begin{equation}
\begin{aligned}
& \Pr_{\parens{\bm{X}^{(q,t)}}_{0 \leq q \leq Q, t \in [T]} \sim \Xi_{Q,T}} \brackets{\max_{t \in [T], q \in [Q]} \norm{\bm{X}^{(q-1,t)} - \bm{X}^{(q,t)}}_1 \leq \frac{3 p n}{Q}} \geq \\
\geq{} & 1 - T \exp\parens{-(3 \log 2 - 2) \frac{p n}{Q}}.
\end{aligned}
\end{equation}

Now, take the $\parens{\vec{s}^{(q,t)}}_{0 \leq q \leq Q, t \in [T]}$ in \Cref{lem:local_alg_deterministic_shadows_reduction}, which exist with probability at least $1 - 3 Q^2 T \probst - 3 Q T \probf - Q T \probb$. Then, assuming that the condition above also holds, we have that

\begin{equation}
\begin{aligned}
& \abs{\dhamming\parens{\vec{s}^{(q-1,t)} - \vec{s}^{(q-1,t')}} - \dhamming\parens{\vec{s}^{(q,t)} - \vec{s}^{(q,t')}}} \leq \\
\leq{} & \dhamming\parens{\vec{s}^{(q-1,t)}, \vec{s}^{q,t}} + \dhamming\parens{\vec{s}^{(q-1,t)}, \vec{s}^{q,t}} \leq \\
\leq{} & 2 f + L \norm{\bm{X}^{(q-1,t)} - \bm{X}^{(q,t)}}_1 + L \norm{\bm{X}^{(q-1,t')} - \bm{X}^{(q,t')}}_1 \leq \\
\leq{} & 2 f + \frac{6 L p n}{Q}.
\end{aligned}
\end{equation}
\end{proof}

\begin{definition}
We define the set $\cS\parens{\alpha, m, \nu_0, \nu_1, \parens{\bm{X}^{(t)}}_{t \in [m]}}$ as the set consisting of $m$-tuples of bit strings $\parens{\vec{s}^{(t)}}_{t \in [m]}$ such that for all $t \in [m]$ we have
\begin{equation}
V \parens{\bm{X}^{(t)}, \vec{s}^{(t)}} \geq \alpha,
\end{equation}
and for all $t \neq t' \in [m]$ we have
\begin{equation}
\frac{1}{n} \dhamming\parens{\vec{s}^{(t)}, \vec{s}^{(t')}} \in [\nu_0, \nu_1].
\end{equation}
\end{definition}

\begin{lemma}\label{lem:m_admissible}
Consider some collection of problem instances $\parens{\bm{X}^{(q,t)}}_{0 \leq q \leq Q, t \in [T]}$ satisfying the ``with high probability'' events in \Cref{lem:local_alg_deterministic_shadows_reduction} and consider the corresponding collection of bit strings $\vec{s}^{(q,t)}$. Assume that
\begin{equation}
\frac{2 f}{n} + \frac{6 L p}{Q} < \nu_1 - \nu_0.
\end{equation}
Furthermore, suppose that for any $\cM \in \binom{[T]}{m}$ we have that
\begin{equation}
\cS\parens{\frac{1}{3}\parens{\gamma - \delta}E_\ast, m, 0, \nu_1, \parens{\bm{X}^{(Q,t)}}_{t \in \cM}} = \varnothing.
\end{equation}
Then, for any $\cM \in \binom{[T]}{m}$ there must exist some $t, t' \in \cM$ and $q \in [Q]$ such that
\begin{equation}
\frac{1}{n} \dhamming\parens{\vec{s}^{(q,t)}, \vec{s}^{(q,t')}} \in [\nu_0, \nu_1].
\end{equation}
\end{lemma}

\begin{proof}
Fix any $\cM \in \binom{[T]}{m}$. By assumption, we know that there must exist some $t, t' \in \cM$ such that
\begin{equation}
\frac{1}{n} \dhamming\parens{\vec{s}^{(Q,t)}, \vec{s}^{(Q,t')}} > \nu_1.
\end{equation}
But we also know that
\begin{equation}
\frac{1}{n} \dhamming\parens{\vec{s}^{(0,t)}, \vec{s}^{(0,t')}} = 0,
\end{equation}
and, by \Cref{lem:local_alg_stability_between_replicas} and the assumption on $\nu_1 - \nu_0$, we have that for all $q \in [Q]$
\begin{equation}
\abs{\dhamming\parens{\vec{s}^{(q-1,t)} - \vec{s}^{(q-1,t')}} - \dhamming\parens{\vec{s}^{(q,t)} - \vec{s}^{(q,t')}}} < \nu_1 - \nu_0.
\end{equation}
Since each step between $q - 1$ and $q$ is smaller than the ``gap'' between $\nu_0$ and $\nu_1$ that needs to be crossed going from $q = 0$ to $q = Q$, there must exist some $q$ such that
\begin{equation}
\frac{1}{n} \dhamming\parens{\vec{s}^{(q,t)}, \vec{s}^{(q,t')}} \in [\nu_0, \nu_1].
\end{equation}
\end{proof}

\begin{lemma}\label{lem:strong_hardness_existence}
Make the same assumptions as in \Cref{lem:m_admissible} and suppose that $T = (Q + 1)^{(Q + 1) m}$. Then, there must exist some $\cM \in \binom{[T]}{m}$ and some $q \in [Q]$ such that
\begin{equation}
\cS\parens{\frac{1}{3}(\gamma - \delta)E_\ast, m, \nu_0, \nu_1, \parens{\bm{X}^{(q,t)}}_{t \in \cM}} \neq \varnothing.
\end{equation}
\end{lemma}

\begin{proof}
Consider a graph $G$ with $T$ vertices where we draw an edge of color $q$ between vertices $t$ and $t'$ if $\frac{1}{n} \dhamming\parens{\vec{s}^{(t)}, \vec{s}^{(t')}} \in [\nu_0, \nu_1]$. Then, under the assumption of \Cref{lem:m_admissible}, for any subset $\cM$ of $m$ vertices in $G$, the subgraph corresponding to $\cM$ contains at least one edge. Thus, we can directly apply \Cref{cor:ramsey_m_clique}, which tells us that $G$ must contain a monochromatic $m$-clique. This is equivalent to the condition we want to show.
\end{proof}

\subsection{Non-Existence of Forbidden Configurations}\label{sec:strong_hardness_QOGP}

In this section, we prove the QOGP and demonstrate that the set of configurations described above is empty with high probability. The structure of the proof is similar to \Cref{sec:weak_hardness_QOGP}, but this time we must use the structure of the problem itself to make the argument work instead of relying on the interpolation step being large enough. The covariance bound will apply even if we are looking at the \emph{same} problem instance, which ensures that the value of $k$ at which we can say the problem is hard can be independent of $L$.

First, the entropy bound is exactly the same as before. Consider the set $\cF(n, m, 1, \nu_1)$ (as in \Cref{def:F_set}, with $R = 1$). By \Cref{lem:F_set_size}, we have that
\begin{equation}
\abs{\cF(n, m, 1, \nu_1)} \leq \exp\brackets{n \cdot \parens{\log 2 + (m - 1) \entropy(\nu_1)} (1 + o_n (1))}.
\end{equation}

\subsubsection{Covariance Bound}

Now, consider some $\parens{\vec{s}^{(t)}}_{t \in [m]} \in \cF(n, m, 1, \nu_1)$. We know that if we sample $\parens{\bm{X}^{(q,t)}}_{0 \leq q \leq Q, t \in [T]} \sim \Xi_{Q,T}$, then 
\begin{equation}
\EE\brackets{V \parens{\bm{X}^{(q,t)}, \vec{s}^{(t)}}^2} = \frac{p}{n},
\end{equation}
and for $t \neq t'$,
\begin{equation}
\begin{aligned}
& \EE\brackets{V \parens{\bm{X}^{(q,t)}, \vec{s}^{(q,t)}} V \parens{\bm{X}^{(q,t')}, \vec{s}^{(t',r_{t'})}}} = \\
={} & \frac{k}{n^2} \sum_{A, A' \in \binom{[n]}{k}} \EE\brackets{X^{(q,t)}_A \bm{X}^{(q,t')}_{A'}} (-1)^{\sum_{v \in A} \bm{s}^{(t)}_v + \sum_{v' \in A'} \bm{s}^{(t')}_{v'}} = \\
={} & \frac{k}{n^2} \sum_{A \in \binom{[n]}{k}} \EE\brackets{X^{(q,t)}_A X^{(q,t')}_A} (-1)^{\sum_{v \in A} \parens{\bm{s}^{(t)}_v + \bm{s}^{(t')}_v}} = \\
={} & \frac{k p}{n^2 \binom{n - 1}{k - 1}} \sum_{\substack{A \in \binom{[n]}{k} \\ A \cap \brackets{\ceil{q n / Q}} = \varnothing}} (-1)^{\sum_{v \in A} \parens{\bm{s}^{(t)}_v + \bm{s}^{(t')}_v}}.
\end{aligned}
\end{equation}

For $i,j \in \{0, 1\}$, let
\begin{equation}
w^{(q)}_{i,j} \coloneq
\abs{\braces{\ceil{q n / Q} < v \leq n : \vec{s}^{(t)}_v = i \wedge s^{(t')}_v = j}}.
\end{equation}
Note that $w_{0,0} + w_{0,1} + w_{1,0} + w_{1,1} = 1 - \frac{q}{Q}$.

Then, proceeding as in \Cref{thm:lower_bound}, we can write
\begin{equation}
\begin{aligned}
& \EE\brackets{V \parens{\bm{X}^{(q,t)}, \vec{s}^{(q,t)}} V \parens{\bm{X}^{(q,t')}, \vec{s}^{(t',r_{t'})}}} = \\
={} & \frac{k p}{n^2 \binom{n - 1}{k - 1}} \sum_{a=0}^{k} \sum_{b=0}^{k-a} \sum_{b'=0}^{k-a-b} \binom{w^{(q)}_{0,0} n}{a} \binom{w^{(q)}_{0,1} n}{b} \binom{w^{(q)}_{1,0} n}{b'} \binom{w^{(q)}_{1,1} n}{k - a - b - b'} (-1)^{b + b'} = \\
={} & \frac{p}{n} \parens{w^{(q)}_{0,0} + w^{(q)}_{1,1} - w^{(q)}_{0,1} - w^{(q)}_{1,0}}^k (1 + o_n (1)) = \\
={} & \frac{p}{n} \parens{1 - \frac{q}{Q} - \frac{2}{n} \dhamming\parens{\vec{s}^{(t)}_{> \ceil{q n / Q}}, \vec{s}^{(t')}_{> \ceil{q n / Q}}}}^k (1 + o_n (1)),
\end{aligned}
\end{equation}
that is, we are looking at the Hamming distance between the bit strings ignoring the first $q n / Q$ bits.

Now, just as in \Cref{thm:weak_hardness}, we can apply the Berry-Esseen theorem and view the distributions as Gaussian up to an $o_n (1)$ error. We then have that

\begin{equation}
\begin{aligned}
& \Pr_{\bm{X}^{(q,1)}, \dots, \bm{X}^{(q,m)} \sim \Xi_{\frac{q}{Q},m}} \brackets{\min_{t \in [m]} V \parens{\bm{X}^{(q,t)}, \vec{s}^{(t)}} \geq \alpha} \leq \\
\leq{} & \Pr \brackets{\min_{t \in [m]} G^{(t)} \geq \alpha} \leq \\
\leq{} & \Pr \brackets{\sum_{t \in [m]} G^{(t)} \geq m \alpha} \leq \\
\leq{} & \exp\brackets{- \frac{\frac{1}{2} m^2 \alpha^2}{\frac{m p}{n} + \frac{p}{n} \sum_{t \neq t' \in [m]} \parens{1 - \frac{q}{Q} - \frac{2}{n} \dhamming\parens{\vec{s}^{(t)}_{> \ceil{q n / Q}}, \vec{s}^{(t')}_{> \ceil{q n / Q}}}}^k}} = \\
={} & \exp\brackets{- \frac{\frac{1}{2 p} n m \alpha^2}{1 + \frac{1}{m} \sum_{t \neq t' \in [m]} \parens{1 - \frac{q}{Q} - \frac{2}{n} \dhamming\parens{\vec{s}^{(t)}_{> \ceil{q n / Q}}, \vec{s}^{(t')}_{> \ceil{q n / Q}}}}^k}}.
\end{aligned}
\end{equation}

So now, by the first moment method, we have that up to $o_n (1)$ factors,
\begin{equation}
\begin{aligned}
& \Pr_{\parens{\bm{X}^{(q,t)}}_{t \in [m]} \sim \Xi_{\frac{q}{Q},m}} \brackets{\cS\parens{\alpha, m, \nu_0, \nu_1, \parens{\bm{X}^{(q,t)}}_{t \in [m]}} \neq \varnothing} \leq \\
& \Pr_{\parens{\bm{X}^{(q,t)}}_{t \in [m]} \sim \Xi_{\frac{q}{Q},m}} \brackets{\max_{\parens{\vec{s}^{(t)}}_{t \in [m]} \in \cF(n, m, 1, \nu_1)} \min_{t \in [m]} V \parens{\bm{X}^{(q,t)}, \vec{s}^{(t)}} \geq \alpha} \leq \\
\leq{} & \exp\left[n \left(\log 2 + (m - 1) \entropy(\nu_1) - \rule{0cm}{1cm} \right. \right. \\
& \hspace{5em} \left. \left. \rule{0cm}{1cm} - \frac{\frac{1}{2 p} m \alpha^2}{1 + \frac{1}{m} \sum_{t \neq t' \in [m]} \parens{1 - \frac{q}{Q} - \frac{2}{n} \dhamming\parens{\vec{s}^{(t)}_{> \ceil{q n / Q}}, \vec{s}^{(t')}_{> \ceil{q n / Q}}}}^k}\right)\right].
\end{aligned}
\end{equation}

If this is $> 2^{-n}$, then we must have
\begin{equation}\label{eq:local_alg_general_energy_bound}
\begin{aligned}
& \alpha^2 \leq \frac{2 p}{m} \parens{2 \log 2 + (m - 1) \entropy(\nu_1)} \cdot \\
& \hspace{10em} \cdot \parens{1 + \frac{1}{m} \sum_{t \neq t' \in [m]} \parens{1 - \frac{q}{Q} - \frac{2}{n} \dhamming\parens{\vec{s}^{(t)}_{> \ceil{q n / Q}}, \vec{s}^{(t')}_{> \ceil{q n / Q}}}}^k}.
\end{aligned}
\end{equation}

Now, we know that
\begin{equation}
\frac{1}{n} \dhamming\parens{\vec{s}^{(t)}, \vec{s}^{(t')}} \geq \nu_0,
\end{equation}
so
\begin{equation}
\frac{2}{n} \dhamming\parens{\vec{s}^{(t)}_{> \ceil{q n / Q}}, \vec{s}^{(t')}_{> \ceil{q n / Q}}} \geq \nu_0 - \frac{q}{Q}.
\end{equation}

If $k$ is even, we must minimize the distance term to maximize $\alpha$, so we can write

\begin{equation}
\alpha^2 \leq \frac{2 p}{m} \parens{2 \log 2 + (m - 1) \entropy(\nu_1)} \parens{1 + (m - 1) \parens{1 - \nu_0}^k}.
\end{equation}

\subsubsection{Completing the Proof}

Now, we want to show that there exists some choice of $k, m, \nu_0, \nu_1$ with \\ $0 < \nu_0 < \nu_1 < 1$ such that
\begin{equation}
\begin{aligned}
\frac{2 p}{m} \parens{2 \log 2 + (m - 1) \entropy(\nu_1)} \parens{1 + (m - 1) \parens{1 - \nu_0}^k} & < \frac{2 p}{9} \parens{\gamma - \delta}^2 \log 2,
\end{aligned}
\end{equation}
for which it is sufficient to show that
\begin{equation}
\begin{aligned}
\frac{1}{m} \parens{2 \log 2 + (m - 1) \sqrt{2 \nu_1}} \parens{1 + (m - 1) \parens{1 - \nu_0}^k} & < \frac{1}{9} \parens{\gamma - \delta}^2 \log 2.
\end{aligned}
\end{equation}

Now, we select the parameters as follows:
\begin{equation}\label{eq:strong_hardness_param_regime}
\begin{aligned}
m &= \ceil{\frac{54}{(\gamma - \delta)^2}} \\
\nu_1 &= \frac{1}{2 \cdot 27^2} (\gamma - \delta)^4 \log^2 2 \\
\nu_0 &= (1 - \eta) \nu_1 \\
k &\geq \ceil{\frac{\log(2m - 2)}{\nu_0}}
\end{aligned}
\end{equation}
for an arbitrary constant $0 < \eta < 1$.

Then, we have that
\begin{equation}
\begin{aligned}
\frac{1}{m} 2 \log 2 &\leq \frac{1}{27} (\gamma - \delta)^2 \log 2 \\
\frac{m - 1}{m} \sqrt{2 \nu_1} &< \frac{1}{27} (\gamma - \delta)^2 \log 2 \\
1 + (m - 1)(1 - \nu_0)^k &\leq
1 + (m - 1)(1 - \nu_0)^{\frac{\log(2 m - 2)}{\nu_0}} \leq \\
&\leq 1 + (m - 1) e^{-\log(2 m - 2)} = 1 + \frac{m - 1}{2 m - 2} = \frac{3}{2}.
\end{aligned}
\end{equation}

So, for this choice of parameters, we have that
\begin{equation}\label{eq:strong_hardness_ogp}
\begin{aligned}
& \Pr_{\parens{\bm{X}^{(q,t)}}_{t \in [m]} \sim \Xi_{\frac{q}{Q},m}} \brackets{\cS\parens{\frac{1}{3}\parens{\gamma - \delta}E_\ast, m, \nu_0, \nu_1, \parens{\bm{X}^{(q,t)}}_{t \in [m]}} \neq \varnothing} < 2^{-n}.
\end{aligned}
\end{equation}

Furthermore, observe that if we set $q = Q$, then the covariance term in \Cref{eq:local_alg_general_energy_bound} disappears, which means that the above bound will still hold even if we set $\nu_0 = 0$. Thus, we have
\begin{equation}\label{eq:strong_hardness_chaos_property}
\begin{aligned}
& \Pr_{\parens{\bm{X}^{(Q,t)}}_{t \in [m]} \sim \Xi_{1,m}} \brackets{\cS\parens{\frac{1}{3}\parens{\gamma - \delta}E_\ast, m, 0, \nu_1, \parens{\bm{X}^{(Q,t)}}_{t \in [m]}} \neq \varnothing} < 2^{-n}.
\end{aligned}
\end{equation}

Additionally, note that if we set $\delta$ to be sufficiently small, we can effectively set $\delta$ to $0$ in \Cref{eq:strong_hardness_param_regime} and the above bounds will still hold since they are sufficiently loose (e.g. because of the $(m - 1) / m$ term).

\begin{theorem}[Strong hardness for local algorithms]\label{thm:strong_hardness}
Take any constant $0 < \gamma \leq 1$. Then, fix $m, \nu_0, \nu_1, k$ as in \Cref{eq:strong_hardness_param_regime}, using $\delta = 0$ and any $\eta > 0$, and making $k$ even. Then, for $p > 0$ and $L > 0$, let
\begin{equation}
Q = \ceil{\frac{3 L p}{\eta \nu_1}} + 1
\end{equation}
and
\begin{equation}
T = (Q + 1)^{(Q + 1) m}.
\end{equation}
Let $f$ be such that $f \leq \frac{\eta \nu_1}{4} n$ (for sufficiently large $n$), let $\probst, \probb$ be such that
\begin{equation}
3 Q^2 T \probst + 3 Q T \probf < 1,
\end{equation}
and let $\frakd \geq \max\parens{\ceil{e^2 p}, 1}$. Then, there exists no algorithm that is \\ $\parens{f, L, \frakd, \Xi_Q^{(q)}, \probst}$-local for all $0 \leq q \leq Q$ and $(\gamma, \probf)$-optimal for the Quantum Hypergraph Max-Cut model with parameters $k, p$.
\end{theorem}

\begin{proof}
First, observe that assuming the given conditions on $f$ and $Q$, we have that
\begin{equation}
\frac{2 f}{n} + \frac{6 L p}{Q} < \eta \nu_1 = \nu_1 - \nu_0.
\end{equation}
Since $\probb = e^{-\Omega(n)}$ for the shadows estimator and $3 Q^2 T \probst + 3 Q T \probf < 1$, we have that the probability of the events in \Cref{lem:local_alg_deterministic_shadows_reduction} is bounded away from zero. Furthermore, since \Cref{eq:strong_hardness_chaos_property} holds, we have that
\begin{equation}
\begin{aligned}
& \Pr_{\parens{\bm{X}^{(Q,t)}}_{t \in [T]} \sim \Xi_{1,T} \parens{\QMC(n,k,p)}} \Bigg[\cS\parens{\frac{1}{3}\parens{\gamma - \delta}E_\ast, m, 0, \nu_1, \parens{\bm{X}^{(q,t)}}_{t \in \cM}} = \varnothing \\
& \hspace{14em} \forall \cM \in \binom{[T]}{m}\Bigg] \geq
1 - \binom{T}{m} e^{-\Omega(n)} = 1 - e^{-\Omega(n)}.
\end{aligned}
\end{equation}
By \Cref{cor:constant_degree_prob}, if we set $\epsilon = \log 2$, we have that $2 \log(1 / \epsilon) < 1$, so if we have $\frakd \geq \max\parens{\ceil{e^2 p}, 1}$, then we must have
\begin{equation}
\Pr\brackets{\cXfd(\bm{S})} \geq 2^{-n}.
\end{equation}
Thus, with probability of at least $(1 + o_n(1)) 2^{-n}$ over $\Xi_{Q,T}$, all the conditions of \Cref{lem:strong_hardness_existence} are satisfied, which means that for some $q$,
\begin{equation}
\cS\parens{\frac{1}{3}(\gamma - \delta)E_\ast, m, \nu_0, \nu_1, \parens{\bm{X}^{(q,t)}}_{t \in \cM}} \neq \varnothing.
\end{equation}
However, by \Cref{eq:strong_hardness_ogp}, this probability is less than $2^{-n}$, so we have a contradiction. Thus, we have shown that no such local algorithm can exist.
\end{proof}

\begin{corollary}\label{cor:strong_hardness_explicit_bound}
For even values of $k \geq \num{14210}$, the Quantum Hypergraph Max-Cut problem exhibits a statistical-to-computational gap for local algorithms of the form described above.
\end{corollary}

\begin{proof}
This follows by using the values $\gamma = 1 - \num{e-5}$ and $\eta = \num{e-5}$ in \Cref{eq:strong_hardness_param_regime}.
\end{proof}

\section{Applications and Future Directions}\label{chp:applications_future_directions}

In this section, we analyze the stability properties of some prominent algorithms for quantum optimization problems and apply the hardness results of \Cref{chp:weak_hardness,chp:strong_hardness}. We discuss the current barriers to using the framework of the Quantum Overlap Gap Property we are able to prove for Quantum Hypergraph Max-Cut to producing \emph{strong} hardness results for \emph{stable} algorithms, and discuss potential future research directions to address this.

\subsection{Hardness Results for Quantum Algorithms}

We will illustrate the use of both of the hardness results through the example of the Quantum Approximate Optimization Algorithm (QAOA)~\cite{Farhi2014_QAOA}. Similar complexity lower bounds hold for other standard quantum optimization algorithms such as quantum phase estimation, quantum annealing, and Decoded Quantum Interferometry (DQI) as these algorithms are also known to be stable~\cite{Anschuetz2026_glassiness} or, in DQI's case, local under certain assumptions on the decoder~\cite{Anschuetz2025_DQI}.

\begin{definition}
The \emph{depth-$\ell$ QAOA} for the Quantum Hypergraph Max-Cut Hamiltonian $\bm{H}(\bm{X})$ is defined by
\begin{equation}
\bm{\cA}_\ell \parens{\bm{X}, \omega} = \selfouterprod{\psi_\ell \parens{\bm{X}, \omega}},
\end{equation}
where
\begin{equation}
\begin{aligned}
& \ket{\psi_\ell \parens{\bm{X}, \omega}} \coloneq
\prod_{j=\ell}^1 \exp\parens{-\iunit \beta_j \bm{H}_M (\omega)} \exp\parens{-\iunit \gamma_j^{(\mathrm{z})} \bm{H}^{(\mathrm{z})} (\bm{X})} \exp\parens{-\iunit \gamma_j^{(\mathrm{y})} \bm{H}^{(\mathrm{y})} (\bm{X})} \cdot \\
& \hspace{10em} \cdot \exp\parens{-\iunit \gamma_j^{(\mathrm{x})} \bm{H}^{(\mathrm{x})} (\bm{X})} \ket{\psi_0 (\omega)},
\end{aligned}
\end{equation}
where $\bm{H}^{(\mathrm{x})} (\bm{X}), \bm{H}^{(\mathrm{y})} (\bm{X}), \bm{H}^{(\mathrm{z})} (\bm{X})$ are the $\sigmax, \sigmay, \sigmaz$ parts of $\bm{H} (\bm{X})$; the state $\ket{\psi_0 (\omega)}$ is chosen to be the maximal-energy eigenstate of the mixing Hamiltonian $\bm{H}_M (\omega)$ (which is a $1$-local Hamiltonian that does not commute with $\bm{H}(\bm{X})$), and the parameters $\beta_j, \gamma_j^{(x)}, \gamma_j^{(y)}, \gamma_j^{(z)}$ (which we assume to be in $[0, 2\pi]$) are optimized over.
\end{definition}

\begin{figure}[h!]
    \centering
    \includegraphics{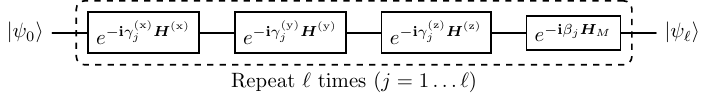}
    \caption{Circuit diagram illustrating the depth-$\ell$ QAOA ansatz.}
    \Description{}
    \label{fig:QAOA}
\end{figure}

\begin{proposition}[{From \citet[Corollary 64]{Anschuetz2026_glassiness}}]\label{prop:QA_stability}
The depth-$\ell$ QAOA algorithm is \\ $\parens{\sqrt{n}, L, \frakd, \Xi_F, 0}$-stable for any $F \geq 0$, where
\begin{equation}
L = \frac{2 \pi}{4 \sqrt{2 n}} \parens{\frac{3}{2} k \frakd}^{4 \ell}.
\end{equation}
\end{proposition}

\begin{corollary}\label{cor:QA_stability}
Suppose that $\ell = C \log n$ and $\frakd = \ceil{e^2 p} \geq 1$. Then, if $C \leq \frac{1}{8 \log 3}$, we must have that the depth-$\ell$ QAOA algorithm is $\parens{\sqrt{n}, \frac{\sqrt{2} \pi}{4}, \frakd, \Xi_F, 0}$-stable for any $F \geq 0$.
\end{corollary}

\begin{proof}
We have that
\begin{equation}
\begin{aligned}
L \leq{} & \frac{2 \pi}{4 \sqrt{2 n}} \parens{\frac{3}{2} e^2 p k}^{4 C \log n} = \\
={} & \frac{2 \pi}{4 \sqrt{2 n}} \exp\parens{4 C \log n \log \parens{\frac{3}{2} e^2 p k}} = \\
={} & \frac{\sqrt{2} \pi}{4} n^{4C \log\parens{\frac{3}{2} e^2 p k} - \frac{1}{2}}
\end{aligned}
\end{equation}
If $L = O_n(1)$, we must have
\begin{equation}
C \leq \frac{1}{8 \log\parens{\frac{3}{2} e^2 p k}}.
\end{equation}
Since $e^2 p > 1$ and $k \geq 2$, we must have
\begin{equation}
C \leq \frac{1}{8 \log 3}.
\end{equation}
\end{proof}

\begin{corollary}\label{cor:QA_hardness}
The $C \log n$ depth QAOA algorithm with
\begin{equation}
C \leq \frac{1}{8 \log 3}
\end{equation}
cannot be $(1, \num{1e-6})$-optimal for the Quantum Hypergraph MaxCut problem if $p > 1 / e^2$ and
\begin{equation}
k \geq \num{4.56e14} \geq \num{3.03e15} p \frac{\sqrt{2} \pi}{4}.
\end{equation}
\end{corollary}

\begin{proof}
This follows immediately by applying \Cref{thm:weak_hardness} to the result of \Cref{cor:QA_stability}.
\end{proof}

Now, if we instead consider \emph{constant}-depth QAOA, we can get a significantly better bound.

\begin{proposition}\label{prop:QA_locality}
The depth-$\ell$ QAOA algorithm is $\parens{0, k^{\ell + 1} \frakd^\ell, \frakd, \Xi_Q^{(q)}, 0}$-local for any $1 \leq q \leq Q$ and any $\frakd$.
\end{proposition}

\begin{proof}
Suppose that the weight of one hyperedge $A \in \binom{[n]}{k}$ is changed. This hyperedge is adjacent to $k$ qubits, each of which is adjacent to at most $\frakd$ hyperedges. Applying time evolution by $\bm{H}\parens{\bm{X}}$ propagates changes to $k \frakd$ hyperedges. Applying time evolution by the mixing Hamiltonian does not propagate to other qubits since the Hamiltonian is $1$-local. So, after $\ell$ layers, a change to a single hyperedge propagates to $k \parens{k \frakd}^\ell = k^{\ell + 1} \frakd^\ell$ qubits. Thus, we have that
\begin{equation}
\norm{\bm{\cA}(\bm{X}, \omega) - \bm{\cA}(\bm{X}', \omega)}_{W_\infty} \leq k^{\ell + 1} \frakd^\ell \norm{\bm{X} - \bm{X}}_1,
\end{equation}
and so, depth-$\ell$ QAOA is $\parens{0, k^{\ell + 1} \frakd^\ell, \frakd, \Xi_F, 0}$-local.
\end{proof}

\begin{corollary}
There exists $0 < \gamma < 1$ such that QAOA of constant depth $\ell$ cannot be $(\gamma, \probf)$-optimal for the Quantum Hypergraph Max-Cut problem with even values of $k \geq \num{14210}$, for any sufficiently small $\ell$-dependent value of $\probf \geq 0$.
\end{corollary}

\begin{proof}
This follows from \Cref{cor:strong_hardness_explicit_bound}.
\end{proof}

Observe that while the bound needed for $\probf$ is $\ell$-dependent, it is independent of $n$. In practice, the error probability of any near-optimal algorithm may be reduced to $o_n (1)$ by repetition, as long as the original algorithm has error probability bounded away from $1$.





\subsection{Limitations of Our Technique}\label{sec:limitations}

If we are working with stable algorithms and the $W_2$ metric, we must make use of \Cref{prop:expected_wasserstein_distance} to bound expected Hamming distances between pairs of solutions generated by a quantum algorithm. To bound the error from the classical shadows estimator, we then need to take $R$ replicas and consider them as being \emph{collectively} stable or optimal. If $R > 1$, these conditions give us very little control of the individual replicas and this construction introduces many degrees of freedom that make an OGP difficult to construct.

The reason the proof in \Cref{thm:weak_hardness} was possible is that this was not a ``true'' OGP and did not depend on the structure of the optimization problem in any significant way. The value of $k$ chosen in that result depends on $F$, which itself is chosen depending on $L$. The central idea of the argument is that if we fix a Lipschitz constant $L$, then for large enough $k$, the distance between near-optimal solutions for two problem instances where $F n$ qubits are resampled becomes too large for $L$-stable algorithms. To make a strong hardness claim like in \Cref{thm:strong_hardness}, we effectively need to show that this holds even when we do not resample at all (i.e., keep the same problem instance).

We now sketch out this limitation in more detail. We first proceed as in \Cref{lem:deterministic_shadows_reduction} and \Cref{lem:ogp_existence}. Then, we argue that there exists a collection of bit strings $\vec{s}^{(q, t, r)}$ for $0 \leq q \leq Q, t \in [T], r \in [R]$ satisfying:

\begin{equation}
\frac{1}{R} \sum_{r=1}^R \dhamming\parens{\vec{s}^{(q-1,t,r)}, \vec{s}^{(q,t,r)}} \leq \beta f + \beta L \norm{\bm{X}^{(q-1,t)} - \bm{X}^{(q,t)}}_1,
\end{equation}
and
\begin{equation}
\max_{r \in [R]} V \parens{\bm{X}^{(q,t)}, \vec{s}^{(q,t,r)}} \geq \frac{1}{3}\parens{\gamma - \delta}E_\ast.
\end{equation}

If we then argue as in \Cref{lem:strong_hardness_existence} (by applying Ramsey theory results on $T$ many replicas), we can say that, with high probability, there exists some $q \in [Q]$ and some subset $\cM \subseteq \binom{[T]}{m}$ (identified with $[m]$) such that
\begin{equation}\label{eq:limitations_collective_stability}
\frac{1}{R n} \sum_{r=1}^R \dhamming\parens{\vec{s}^{(q,t,r)}, \vec{s}^{(q,t',r)}} \in [\nu_0, \nu_1]
\end{equation}
and
\begin{equation}
\max_{r \in [R]} V \parens{\bm{X}^{(q,t)}, \vec{s}^{(q,t,r)}} \geq \frac{1}{3}\parens{\gamma - \delta}E_\ast.
\end{equation}

Now, to demonstrate an algorithmic hardness claim akin to \Cref{thm:strong_hardness}, we must argue that the set of such collections of bit strings $\parens{\vec{s}^{(q,t,r)}}_{t \in [m], r \in [R]}$ is empty with high probability for all $q \in [Q]$. We have the following:
\begin{equation*}
\begin{aligned}
& \Pr_{\bm{X}^{(q,1)}, \dots, \bm{X}^{(q,t)} \sim \Xi_{q/Q,m}} \Bigg[\exists \parens{\vec{s}^{(q,t,r)}} \in \cF\parens{n, m, R, \nu_1} : \\
& \hspace{8em} \max_{r \in [R]} V\parens{\bm{X}^{(q,t)}, \vec{s}^{(q,t,r)}} \geq \alpha \forall t \in [m]\Bigg] \leq \\
\leq{} & \abs{\cF\parens{n, m, R, \nu_1}} R^m \max_{r_0, \dots, r_t \in [R]^m} \\
& \hspace{8em} \Pr_{\bm{X}^{(q,1)}, \dots, \bm{X}^{(q,t)} \sim \Xi_{q/Q,m}} \brackets{\min_{t \in [m]} V\parens{\bm{X}^{(q,t)}, \vec{s}^{(q,t,r_t)}} \geq \alpha} \leq
\end{aligned}
\end{equation*}
\begin{equation}
\begin{aligned}
\leq{} & \abs{\cF\parens{n, m, R, \nu_1}} R^m \max_{r_0, \dots, r_t \in [R]^m} \\
& \hspace{8em} \Pr_{\bm{X}^{(q,1)}, \dots, \bm{X}^{(q,t)} \sim \Xi_{q/Q,m}} \brackets{\sum_{t \in [m]} V\parens{\bm{X}^{(q,t)}, \vec{s}^{(q,t,r_t)}} \geq m \alpha} \leq \\
\leq{} & \exp\left[n\left(R \log 2 + (m - 1) R \entropy(\nu_1) - \rule{0cm}{1cm} \right. \right. \\
& \hspace{8em} \left. \left. \rule{0cm}{1cm} - \frac{\frac{1}{2} m^2 \alpha^2}{m + 2 \sum_{t \neq t' \in [m]} \parens{1 - \frac{q}{Q} - \frac{2}{n} \dhamming\parens{\vec{s}^{(q,t,r_t)}_{> \ceil{q n / Q}}, \vec{s}^{(q,t',r_{t'})}_{> \ceil{q n / Q}}}}^k}\right) \right] \leq \\
\leq{} & \exp\left[n\left(R \log 2 + (m - 1) R \entropy(\nu_1) - \rule{0cm}{1cm} \right. \right. \\
& \hspace{8em} \left. \left. \rule{0cm}{1cm} - \frac{\frac{1}{2} m^2 \alpha^2}{m + 2 \sum_{t \neq t' \in [m]} \parens{1 - \frac{2}{n} \dhamming\parens{\vec{s}^{(q,t,r_t)}, \vec{s}^{(q,t',r_{t'})}}}^k}\right) \right],
\end{aligned}
\end{equation}

so we need to have

\begin{equation}
\alpha^2 \leq \frac{2 R}{m^2} \parens{\log 2 + (m - 1) R \entropy(\nu_1)} \parens{m + 2 \sum_{t \neq t' \in [m]} \parens{1 - \frac{2}{n} \dhamming\parens{\vec{s}^{(q,t,r_t)}, \vec{s}^{(q,t',r_{t'})}}}^k}.
\end{equation}

The issue is that the only tool we have available to bound the covariance terms is the ``collective stability'' condition from \Cref{eq:limitations_collective_stability}, which gives us almost no control over each individual covariance term unless $R = 1$. But $R$ needs to be large to drive down the error in the classical shadows estimator. This points to a fundamental limitation of the Quantum Overlap Gap Property we are able to prove in our weak hardness result: while classical shadows allow us to extend the theory of the Overlap Gap Property to study the hardness of quantum problems for quantum algorithms, they also introduce extraneous degrees of freedom that make it much harder to prove non-existence of forbidden configurations than in the classical case.

\subsection{Conclusion and Future Work}

In this work, we used the theoretical framework of the Quantum Overlap Gap Property to establish algorithmic hardness results for the Quantum Hypergraph Max-Cut problem, allowing us to provably guarantee the failure of stable and local algorithms under specific conditions. Our work builds upon prior research on the Quantum Overlap Gap Property~\cite{Anschuetz2026_glassiness} and allows us to extend hardness results for classical combinatorial optimization problems~\cite{Farhi2020_QAOA_needs_to_see_the_whole_graph} to study the limitations of quantum algorithms for \emph{quantum} optimization problems.

One future research direction is analyzing how these results apply to more of the known quantum algorithms for optimization and investigating algorithms where the stability and locality properties emerge for less obvious reasons. One example is the \emph{Decoded Quantum Interferometry (DQI)} algorithm~\cite{Jordan2025_DQI}. It was recently shown~\cite{Anschuetz2025_DQI} that DQI exhibits stability (and in fact, locality in the sense of \Cref{def:local_alg}) under certain assumptions on the decoding rate, and this was used to demonstrate its failure on unstructured MAX-$k$-XOR-SAT instances. However, DQI cannot be directly applied to the Quantum Hypergraph Max-Cut problem. Recent work~\cite{Schmidhuber2025_Hamiltonian_DQI} proposed the \emph{Hamiltonian DQI}, and it may be possible to study the limitations of this algorithm using the tools we have developed in this work.

The most significant open problem that remains is the question of whether or not it is possible to extend the strong hardness result to all stable algorithms. The Overlap Gap Property is a rigorous formalization of a phenomenon known in the physics literature as \emph{shattering}. Recent work~\cite{Zlokapa2025_glassiness} described a quantum version of shattering and its implication on average-case hardness results for quantum algorithms. Other recent work~\cite{Anschuetz2024combinatorialnlts} has connected the existence of the Overlap Gap Property to NLTS-like statements in certain constructed local quantum spin systems. However, this previous work either rely on certain non-rigorous arguments (such as the ``replica trick'') that are difficult to eliminate, or look at highly constructed quantum systems. Future research must focus on developing a stronger form of the QOGP that avoids these issues, possibly by investigating the geometry of the space of quantum states directly instead of using classical shadows estimators.

\section{AI Disclosure}

We used Google Gemini for assistance in some routine mathematical tasks such the calculations and bounds in \Cref{sec:weak_hardness_QOGP} and \Cref{app:ghz}. The tool was also used to generate \LaTeX{} code for the figures in this paper. AI tools were not used in a way that materially affected any part of the research methodology, and the authors verified the correctness and originality of all content including references.

\begingroup
\emergencystretch=1.5em
\printbibliography[heading=bibintoc]
\endgroup

\appendix
\crefalias{section}{appsec}

\section{The GHZ Hypergraph Hamiltonian}\label{app:ghz}

\subsection{Problem Setup}

In this appendix, we discuss a problem related to Quantum Max-Cut but for which we believe no QOGP exists and for which the hardness results established in this work do not apply. While the Quantum Max-Cut Hamiltonian with $k = 2$ uses projectors onto the $\ket{\Psi^-}$ state, the EPR Hamiltonian~\cite{King2023_QMC} uses projectors onto the $\ket{\Phi^+} = \frac{1}{\sqrt{2}}\parens{\ket{0} + \ket{1}}$. One way of extending this to $k$-uniform hypergraphs is to consider projectors onto $k$-qubit GHZ states:

\begin{equation}
\bm{H}_\GHZ^{(n,k)} \parens{\bm{X}} \coloneq \frac{\sqrt{k}}{n} \sum_{\substack{A \in \binom{[n]}{k}}} X_A \ket{\GHZ}_A \bra{\GHZ}_A \otimes \Id_{[n] \setminus A},
\end{equation}
where we select each $X_A$ as before, that is, each $X_A = S_A J_A$ where \\ $S_A \sim \Bern\parens{p / \binom{n - 1}{k - 1}}$, $J_A \sim \{-1, 1\}$ i.i.d., and
\begin{equation}
\ket{\GHZ} = \frac{1}{\sqrt{2}} (\ket{0 \dots 0} + \ket{1 \dots 1}).
\end{equation}

Now, the stabilizers of the GHZ state are generated by $\Id^{\otimes k}, \parens{\sigmax}^{\otimes k}$, and \\ $\Id^{\otimes l} \otimes \sigmaz \otimes \sigmaz \otimes \Id^{\otimes k - 2 - l}$ for $0 \leq l \leq k - 2$. So, we can write
\begin{equation}
\begin{aligned}
\ket{\GHZ}\bra{\GHZ} ={}& \frac{1}{2^k} (\Id^{\otimes k} + \parens{\sigmax}^{\otimes k}) \prod_{l=0}^{k-2} \parens{\Id^{\otimes k} + \Id^{\otimes l} \otimes \sigmaz \otimes \sigmaz \otimes \Id^{\otimes k - 2 - l}} = \\
={}& \frac{1}{2^k} (\Id^{\otimes k} + \parens{\sigmax}^{\otimes k}) \sum_{\substack{A' \subseteq [k] \\ \abs{A'} \text{even}}} \sigmaz_{A'};
\end{aligned}
\end{equation}
and so,
\begin{equation}
\bm{H}_\GHZ^{(n,k)} \parens{\bm{X}} = \frac{\sqrt{k}}{2^k n} \sum_{\substack{A \in \binom{[n]}{k}}} X_A (1 + \sigmax_A) \sum_{\substack{A' \subseteq A \\ \abs{A'} \text{ even}}} \sigmaz_{A'}.
\end{equation}

Now, we consider the classical shadows estimator using the Pauli basis states. Then, we are interested in the problem of finding $\vec{b}, \vec{s}$ that maximize the energy
\begin{equation}
V_\GHZ^{(n,k)} \parens{\bm{X}, \vec{b}, \vec{s}} \coloneq \bra{\vec{b}; \vec{s}} \bm{H} \ket{\vec{b}; \vec{s}}.
\end{equation}
Note that since each term in the Hamiltonian does not consist of a single basis, we cannot just let all $b_v$ be the same in the shadows estimator.

We can write
\begin{equation}\label{eq:V_bs}
\begin{aligned}
& V_\GHZ^{(n,k)} \parens{\bm{X}, \vec{b}, \vec{s}} = \\
={} & \frac{\sqrt{k}}{2^k n} \sum_{\substack{A \in \binom{[n]}{k}}} X_A
\sum_{\substack{A' \subseteq A \\ \abs{A'} \text{ even}}}
\parens{
\bra{\vec{b}; \vec{s}} \sigmaz_{A'} \ket{\vec{b}; \vec{s}} +
\bra{\vec{b}; \vec{s}} \sigmax_A \sigmaz_{A'} \ket{\vec{b}; \vec{s}}
} = \\
={}& \frac{\sqrt{k}}{2^k n} \sum_{\substack{A \in \binom{[n]}{k}}} X_A
\sum_{\substack{A' \subseteq A \\ \abs{A'} \text{ even}}}
\parens{
\bra{\vec{b}; \vec{s}} \sigmaz_{A'} \ket{\vec{b}; \vec{s}} +
(-\iunit)^{\abs{A'}} \bra{\vec{b}; \vec{s}} \sigmay_{A'} \otimes \sigmax_{A \setminus A'} \ket{\vec{b}; \vec{s}}
} = \\
={}& \frac{\sqrt{k}}{2^k n} \sum_{\substack{A \in \binom{[n]}{k}}} X_A
\sum_{\substack{A' \subseteq A \\ \abs{A'} \text{ even}}}
\Bigg(
\prod_{v \in A'} \delta_{b_v, 3} (-1)^{s_v} + \\
& \hspace{10em} + (-1)^\frac{\abs{A'}}{2} \prod_{v \in A'} \delta_{b_v, 2} (-1)^{s_v} \prod_{v \in A \setminus A'} \delta_{b_v, 1} (-1)^{s_v}
\Bigg) = \\
={}& \frac{\sqrt{k}}{2^k n} \brackets{
\sum_{\substack{A \in \binom{[n]}{k}}} X_A
\sum_{\substack{A' \subseteq A \cap \vec{b}[3] \\ \abs{A'} \text{ even}}}
\prod_{v \in A'} (-1)^{s_v} +
\sum_{\substack{A \in \binom{[n]}{k} \\ A \cap \vec{b}[3] = \varnothing \\ A \cap \vec{b}[2] \text{ even}}} X_A (-1)^\frac{\abs{A \cap \vec{b}[2]}}{2} \prod_{v \in A} (-1)^{s_v}
}.
\end{aligned}
\end{equation}

In the above, we write $\vec{b}[b]$ to represent the set of vertex indices $v$ for which $\vec{b}_v = b$. Consider the first term in the above expression. If $A \cap \vec{b}[3] = \varnothing$, the contribution from $A$ is $1$. Otherwise, if $s$ is constant on $A \cap \vec{b}[3]$, the contribution is $2^{\abs{A \cap \vec{b}[3]} - 1}$ since the number of even-size subsets of a non-empty set equals the number of odd-size subsets, and the contribution is $0$ otherwise. So, we have that

\begin{equation}
\begin{aligned}
& V_\GHZ^{(n,k)} \parens{\bm{X}, \vec{b}, \vec{s}} = \\
={} & \frac{\sqrt{k}}{2^k n} \left[
\sum_{\substack{A \in \binom{[n]}{k} \\ A \cap \vec{b}[3] \neq \varnothing \\ s \text{ constant on } A \cap \vec{b}[3]}} X_A 2^{\abs{A \cap \vec{b}[3]} - 1} +  \right. \\
& \left.
\phantom{\sum_{\substack{A \in \binom{[n]}{k} \\ A \cap \vec{b}[3] \neq \varnothing \\ s \text{ constant on } A \cap \vec{b}[3]}}} 
+ \sum_{\substack{A \in \binom{[n]}{k} \\ A \cap \vec{b}[3] = \varnothing}} X_A \parens{1 + \mathds{1}_{\abs{A \cap \vec{b}[2]} \text{ even}} (-1)^\frac{\abs{A \cap \vec{b}[2]}}{2} \prod_{v \in A} (-1)^{s_v}
}\right].
\end{aligned}
\end{equation}

\subsection{Bounds on the Maximal Energy}\label{sec:ghz_ground_state_bounds}

\subsubsection{Upper Bound}\label{sec:ghz_upper_bound}

\begin{proposition}
For any $\epsilon > 0$, we have that
\begin{equation}
\Pr\brackets{\max_{\vec{b},\vec{s}} V_\GHZ^{(n,k)} \parens{\bm{X}, \vec{b}, \vec{s}} \geq \alpha_U + \epsilon} = o_n(1),
\end{equation}
where
\begin{equation}
\begin{aligned}
& \alpha_U \coloneq \\
& \max_{0 \leq w \leq z \leq 1} \sqrt{\frac{(\entropy(z) + z \entropy(w / z))\brackets{(4w - z + 1)^k + (3z - 4w + 1)^k + 16 (1-z)^k} p}{2 \cdot 4^k}}.
\end{aligned}
\end{equation}
\end{proposition}

\begin{proof}
We know by the universality result \cite[Theorem 6]{Anschuetz2025_bounds} that we can equivalently consider taking $X_A \sim \cN\parens{0, p / \binom{n - 1}{k - 1}}$ instead of choosing from $\{-1, 1\}$ since the first and second moments of these distributions are the same, and the Rademacher disorder is bounded. 

Take some fixed $\vec{b},\vec{s}$. Observe that when taking the expectation over the choices of $X_A$, we have that
\begin{equation}
\EE[V_\GHZ^{(n,k)} \parens{\bm{X}, \vec{b}, \vec{s}}] = 0
\end{equation}
since negating every $X_A$ corresponds to negating the total energy.

Now, consider the set of states
\begin{equation}
\cS_{z,w} \coloneq \{(\vec{b},\vec{s}) \in \{1,2,3\}^n \times \{0, 1\}^n : \abs{\vec{b}[3]} = z n, \abs{\vec{b}[3] \cap \vec{s}[1]} = w n\}
\end{equation}
for any $0 \leq w \leq z \leq 1$ such that $z n, w n \in \NN$.

Now, if we take some $(\vec{b},\vec{s}) \in \cS_{z,w}$, we can write
\begin{equation}
\begin{aligned}
& \EE\brackets{V_\GHZ^{(n,k)} \parens{\bm{X}, \vec{b}, \vec{s}}^2} = \\
={} & \frac{k}{4^k n^2} \frac{p}{\binom{n - 1}{k - 1}} \left[\sum_{\substack{A \in \binom{[n]}{k} \\ A \cap b[3] \neq \varnothing \\ s \text{ constant on } A \cap b[3]}} 2^{2 \abs{A \cap b[3]} - 2} + \right. \\
& \left. \phantom{\sum_{\substack{A \in \binom{[n]}{k} \\ A \cap b[3] \neq \varnothing \\ s \text{ constant on } A \cap b[3]}}}
+ \sum_{\substack{A \in \binom{[n]}{k} \\ A \cap b[3] = \varnothing}}
\parens{1 + \mathds{1}_{\abs{A \cap b[2]} \text{ even}} (-1)^\frac{\abs{A \cap b[2]}}{2} \prod_{v \in A} (-1)^{s_v}
}^2 \right] \leq \\
\end{aligned}
\end{equation}
\begin{equation}
\begin{aligned}
\leq{} & \frac{p}{4^k n \binom{n}{k}} \brackets{\sum_{a=0}^k 2^{2a - 2} \parens{\binom{w n}{a} + \binom{(z - w) n}{a}} \binom{(1 - z) n}{k - a} + 4\binom{(1 - z) n}{k}} \approx \\
\approx{} & \frac{p k!}{4^k n \cdot n^k} \brackets{\sum_{a=0}^k 4^{a - 1} \frac{n^k}{a! (k - a)!} \parens{w^a + (z - w)^a} (1 - z)^{k - a} + \frac{4 n^k (1 - z)^k}{k!}} = \\
={} & \frac{p}{4^k n} \brackets{\sum_{a=0}^k 4^{a - 1} \binom{k}{a} \parens{w^a + (z - w)^a} (1 - z)^{k - a} + 4 (1 - z)^k} = \\
={} & \frac{p}{4^k n} \brackets{\frac{1}{4}(4 w + 1 - z)^k + \frac{1}{4}(4(z - w) + 1 - z)^k + 4 (1 - z)^k} = \\
={} & \frac{p}{4^k n} \brackets{\frac{1}{4}(4 w + 1 - z)^k + \frac{1}{4}(3z + 1 - 4w)^k + 4 (1 - z)^k}.
\end{aligned}
\end{equation}

Now, by a Gaussian tail bound, we have that
\begin{equation}
\begin{aligned}
& \Pr\brackets{V_\GHZ^{(n,k)} \parens{\bm{X}, \vec{b}, \vec{s}} \geq \alpha} \leq \\
\leq{} & \exp\brackets{-\frac{\frac{1}{2} \alpha^2}{\frac{p}{4^k n} \brackets{\frac{1}{4}(4 w + 1 - z)^k + \frac{1}{4}(3z + 1 - 4w)^k + 4 (1 - z)^k}}} = \\
={} & \exp\brackets{-\frac{2 \alpha^2 \cdot 4^k n / p}{(4 w + 1 - z)^k + (3z + 1 - 4w)^k + 16 (1 - z)^k}}.
\end{aligned}
\end{equation}

Now, using a union bound, we have
\begin{equation}
\begin{aligned}
& \Pr\brackets{\max_{(\vec{b},\vec{s}) \in \cS_{z,w}} V_\GHZ^{(n,k)} \parens{\bm{X}, \vec{b}, \vec{s}} \geq \alpha} \leq \\
={} & \binom{n}{z n} \binom{z n}{w n} \exp\brackets{-\frac{2 \alpha^2 \cdot 4^k n / p}{\frac{1}{2}(4 w + 1 - z)^k + \frac{1}{2}(3z + 1 - 4w)^k + 8 (1 - z)^k}} = \\
={} & \exp\parens{n \entropy(z) + z n \entropy(w / z) - \frac{2 \alpha^2 \cdot 4^k n}{(4w - z + 1)^k + (3z - 4w + 1)^k + 16 (1-z)^k}},
\end{aligned}
\end{equation}
where $\entropy(x) = -x \log x - (1 - x) \log (1 - x)$ is the binary entropy function. Now, we have that
\begin{equation}
\Pr\brackets{\max_{(\vec{b},\vec{s}) \in \cS_{z,w}} V_\GHZ^{(n,k)} \parens{\bm{X}, \vec{b}, \vec{s}} \geq \alpha_U} = e^{-\Omega(n)}.
\end{equation}
Now, performing a union bound over all choices of $z, w$, of which there are $O(n^2)$, we have
\begin{equation}
\Pr\brackets{\max_{\vec{b},\vec{s}} V_\GHZ^{(n,k)} \parens{\bm{X}, \vec{b}, \vec{s}} \geq \alpha_U} = O(n^2) e^{-\Omega(n)} = o_n(1).
\end{equation}
\end{proof}

Note that if we restrict to the case $z = 1$ (only looking at classical states), the bound becomes
\begin{equation}
\alpha_{U,\mathrm{cl}} \coloneq \max_{0 \leq w \leq 1} \sqrt{\frac{1}{2} \entropy(w) (w^k + (1 - w)^k) p}.
\end{equation}
Since the $16(1 - z)^k$ term becomes subleading for large $k$, the upper bound $\alpha_U$ approaches $\alpha_{U,\mathrm{cl}}$. Furthermore, observe that for large $k$, this is maximized when $w$ approaches $0$ or $1$.

\subsubsection{Lower Bound}\label{sec:ghz_lower_bound}

Now, to find a lower bound, we will use the second moment method.

\begin{proposition}
For any $\epsilon > 0$, we have that
\begin{equation}
\Pr\brackets{\max_{\vec{b},\vec{s}} V_\GHZ^{(n,k)} \parens{\bm{X}, \vec{b}, \vec{s}} \leq \alpha_L - \epsilon} = o_n (1),
\end{equation}
where
\begin{equation}
\alpha_L \coloneq
\max_{1/2 < w \leq 1} \alpha_L(w) \coloneq
\max_{1/2 < w \leq 1} \sqrt{\frac{1}{2} \entropy(w) w^k p (1 - o_k (1))}.
\end{equation}
\end{proposition}

\begin{proof}
Let $Y_{\vec{b},\vec{s}} = [V_\GHZ^{(n,k)} \parens{\bm{X}, \vec{b}, \vec{s}} \geq \alpha]$ for some threshold $\alpha$ to be determined. Let
\begin{equation}
Y = \sum_{\vec{b},\vec{s}} Y_{\vec{b},\vec{s}},
\end{equation}
that is, the number of configurations with energy at least $\alpha$. By the Paley-Zygmund inequality, we have that
\begin{equation}
\Pr[\max_{\vec{b},\vec{s}} V_\GHZ^{(n,k)} \parens{\bm{X}, \vec{b}, \vec{s}} \geq \alpha] = \Pr[Y > 0] \geq
\frac{\EE[Y]^2}{\EE[Y^2]} =
\frac{\sum_{\vec{b},\vec{s},\vec{b}',\vec{s}'} \EE[Y_{\vec{b},\vec{s}}] \EE[Y_{\vec{b}',\vec{s}'}]}{\sum_{\vec{b},\vec{s},\vec{b}',\vec{s}'} \EE[Y_{\vec{b},\vec{s}} Y_{\vec{b}',\vec{s}'}]}.
\end{equation}

By the same argument as before, we can equivalently consider taking \\ $X_A \sim \cN\parens{0, p / \binom{n - 1}{k - 1}}$. So, we can view $V_\GHZ^{(n,k)} \parens{\bm{X}, \vec{b}, \vec{s}}$ as a Gaussian random variable with mean $0$ and variance $\sigma_{\vec{b},\vec{s}}^2$. We then have from standard Gaussian tail bounds \cite{Gamarnik2023_shattering} that
\begin{equation}
\EE[Y_{\vec{b},\vec{s}}] = \Pr[V_\GHZ^{(n,k)} \parens{\bm{X}, \vec{b}, \vec{s}} \geq \alpha] = \frac{\sigma_{\vec{b},\vec{s}}}{\alpha \sqrt{2 \pi}} e^{-\frac{\alpha^2}{\sigma^2}} (1 + o_n(1))
\end{equation}
as long as $\alpha / \sigma_{\vec{b},\vec{s}} \geq \omega_n (1)$.

Now, to find this lower bound, we will just consider classical states with a majority of $1$s: take $z = 1$, $w > 1 / 2$. For such states, we can calculate:

\begin{equation}
\sigma_{\vec{b},\vec{s}}^2 = \frac{1}{4^k n \binom{n}{k}} \sum_{\substack{A \in \binom{[n]}{k} \\ s \text{ constant on } A}} 2^{2 k - 2} \approx \frac{w^k p}{4 n} \eqcolon \sigma^2.
\end{equation}

Now, consider two states defined by $s, s'$ that share $1$s at a $u < w$ fraction of locations. Then, we have
\begin{equation}
\begin{aligned}
& \EE[V_\GHZ^{(n,k)} \parens{\bm{X}, \vec{b}, \vec{s}} V_\GHZ^{(n,k)} \parens{\bm{X}, \vec{b}, \vec{s}'}] = \\
={} & \frac{p}{4^k n \binom{n}{k}} \brackets{\binom{u n}{k} + 2 \binom{(w-u) n}{k} + \binom{(1 - 2 w + u) n}{k}} \frac{1}{4} 4^k \approx \\
\approx{} & \frac{p}{4 n} \parens{u^k + 2 (w - u)^k + (1 - 2 w + u)^k}.
\end{aligned}
\end{equation}

So now, we have the correlation coefficient between $V_\GHZ^{(n,k)} \parens{\bm{X}, \vec{b}, \vec{s}}$ and \\ $V_\GHZ^{(n,k)} \parens{\bm{X}, \vec{b}, \vec{s}'}$:
\begin{equation}
\rho = \frac{\EE[V_\GHZ^{(n,k)} \parens{\bm{X}, \vec{b}, \vec{s}} V_\GHZ^{(n,k)} \parens{\bm{X}, \vec{b}, \vec{s}'}]}{\sigma_{\vec{b},\vec{s}} \sigma_{\vec{b},\vec{s}'}} \approx \frac{1}{w^k} \parens{u^k + 2 (w - u)^k + (1 - 2 w + u)^k}.
\end{equation}

Using a bivariate Gaussian tail bound \cite{Gamarnik2023_shattering}, we can then establish that
\begin{equation}
\begin{aligned}
& \frac{\EE[Y_{\vec{b},\vec{s}} Y_{\vec{b},\vec{s}'}]}{\EE[Y_{\vec{b},\vec{s}}] \EE[Y_{\vec{b},\vec{s}'}]} \leq
\frac{
\frac{(1 + \rho)^2 \sigma^2}{2 \pi \alpha^2 \sqrt{1 - \rho^2}} \exp\parens{-\frac{\alpha^2}{(1 + \rho)\sigma^2}}
}{
\frac{\sigma^2}{2\pi \alpha^2} \exp\parens{-\frac{\alpha^2}{\sigma^2}} (1 + o_n(1))
} \leq \\
\leq{} & \frac{(1 + \rho)^2}{\sqrt{1 - \rho^2}} \exp\parens{\frac{\alpha^2}{\sigma^2}\parens{1 - \frac{1}{1 + \rho}}} =
\frac{(1 + \rho)^2}{\sqrt{1 - \rho^2}} \exp\parens{\frac{\alpha^2}{\sigma^2 (1 + \rho^{-1})}} = \\
={} & \frac{(1 + \rho)^2}{\sqrt{1 - \rho^2}} \exp\parens{\frac{4 n \alpha^2}{w^k p \parens{1 + \frac{w^k}{u^k + 2 (w - u)^k + (1 - 2w + u)^k}}}}
= \\
={} & \exp\parens{\frac{4 n \alpha^2}{w^k p \parens{1 + \frac{w^k}{u^k + 2 (w - u)^k + (1 - 2w + u)^k}}} (1 + o_n(1))}.
\end{aligned}
\end{equation}

So now, we can write our lower bound (ignoring $o_n (1)$ terms):
\begin{equation}
\begin{aligned}
& \Pr[\max_{\vec{b},\vec{s}} V_\GHZ^{(n,k)} \parens{\bm{X}, \vec{b}, \vec{s}} \geq \alpha] \geq \\
\geq{} & \frac{\sum_{1/2 < u \leq 1 : u n \in \NN}{\binom{n}{u n} \binom{(1 - u)n}{(w - u)n} \binom{(1 - w)n}{(w - u)n}}}
{\sum_{1/2 < u \leq 1 : u n \in \NN}\binom{n}{u n} \binom{(1 - u)n}{(w - u)n} \binom{(1 - w)n}{(w - u)n} \exp\parens{\frac{4 n \alpha^2}{w^k p \parens{1 + \frac{w^k}{u^k + 2 (w - u)^k + (1 - 2w + u)^k}}}}}.
\end{aligned}
\end{equation}

We can write
\begin{equation}
\binom{n}{u n} \binom{(1 - u)n}{(w - u)n} \binom{(1 - w)n}{(w - u)n} \approx e^{n \Phi(u, w)},
\end{equation}
where
\begin{equation}
\begin{aligned}
\Phi(u, w) = \entropy(u) + (1 - u) H\parens{\frac{w - u}{1 - u}} + (1 - w) H\parens{\frac{w - u}{1 - w}} = \\
= - u \log u - 2 (w - u) \log(w - u) - (1 - 2 w + u) \log (1 - 2 w + u).
\end{aligned}
\end{equation}

For a fixed $w$, we can see that
\begin{equation}
\frac{\diff}{\diff u} \Phi(u, w) =
-\log u - 1 + 2\log(w - u) + 2 - \log(1 - 2w + u) - 1 =
\log\parens{\frac{(w - u)^2}{u(1 - 2w + u)}}
\end{equation}
\begin{equation}
\begin{aligned}
& \frac{\diff}{\diff u} \Phi(u, w) = 0 \implies
(w - u)^2 = u(1 - 2w + u) \implies \\
\implies{} & u^2 - 2wu + w^2 = u - 2wu + u^2 \implies u = w^2.
\end{aligned}
\end{equation}

So, $\Phi(u, w)$ is maximized at $u = w^2$, with a value of
\begin{equation}
\begin{aligned}
& -2 \log w - \parens{2 w (1 - w) + 2 (1 - w)^2} \log(1 - w) = \\
={} & -2 w \log w - 2 (1 - w) \log (1 - w) = 2 \entropy(w).
\end{aligned}
\end{equation}

Now, if we let
\begin{equation}
\Psi(u, w) = \frac{4 n \alpha^2}{w^k p \parens{1 + \frac{w^k}{u^k + 2 (w - u)^k + (1 - 2w + u)^k}}},
\end{equation}

we can write
\begin{equation}
\Pr[\max_{\vec{b},\vec{s}} V_\GHZ^{(n,k)} \parens{\bm{X}, \vec{b}, \vec{s}} \geq \alpha] \geq
\frac{\sum_{1/2 < u \leq 1 : u n \in \NN} e^{n \Phi(u, w)}}
{\sum_{1/2 < u \leq 1 : u n \in \NN} e^{n \Phi(u, w) + n \Psi(u, w)}}.
\end{equation}

The numerator can be approximated as $e^{2 n \entropy(w)}$. Now, the denominator is maximized when $u = w - o_k (1)$. In the regime $u \approx w > 1/2$, this will become approximately
\begin{equation}
\begin{aligned}
& e^{n (\Phi(w, w) + \Psi(w, w))} \approx \\
\approx{} & \exp\brackets{n \entropy(w) + \frac{4 n \alpha^2}{w^k p \parens{1 + \frac{w^k}{w^k + (1 - w)^k}}}} \approx \\
\approx{} & \exp\parens{n \entropy(w) + \frac{2 n \alpha^2}{w^k p}}
\end{aligned}
\end{equation}

So, overall, we have that
\begin{equation}
\Pr[\max_{\vec{b},\vec{s}} V_\GHZ^{(n,k)} \parens{\bm{X}, \vec{b}, \vec{s}} \geq \alpha] \geq
\exp\brackets{n\parens{\entropy(w) - \frac{2 \alpha^2}{w^k p} }}.
\end{equation}

Suppose for the sake of contradiction that for any sufficiently small $\epsilon > 0$ we have
\begin{equation}
\EE[\max_{\vec{b},\vec{s}} V_\GHZ^{(n,k)} \parens{\bm{X}, \vec{b}, \vec{s}}] \leq \alpha_L(w) - \frac{2 \epsilon}{3}.
\end{equation}

Observe that by~\citet[Proposition 2.7.6]{Vershynin_2018_HDP}, we have the following concentration bounds:
\begin{equation}
\begin{aligned}
\Pr\brackets{\max_{\vec{b},\vec{s}} V_\GHZ^{(n,k)} \parens{\bm{X}, \vec{b}, \vec{s}} \geq \EE[\max_{\vec{b},\vec{s}} V_\GHZ^{(n,k)} \parens{\bm{X}, \vec{b}, \vec{s}}] + \epsilon} \leq \exp\parens{-\frac{2 n \epsilon^2}{w^k p}} \\
\Pr\brackets{\max_{\vec{b},\vec{s}} V_\GHZ^{(n,k)} \parens{\bm{X}, \vec{b}, \vec{s}} \leq \EE[\max_{\vec{b},\vec{s}} V_\GHZ^{(n,k)} \parens{\bm{X}, \vec{b}, \vec{s}}] - \epsilon} \leq \exp\parens{-\frac{2 n \epsilon^2}{w^k p}}.
\end{aligned}
\end{equation}

So, from the above, we must have that

\begin{equation}
\begin{aligned}
& \Pr\brackets{\max_{\vec{b},\vec{s}} V_\GHZ^{(n,k)} \parens{\bm{X}, \vec{b}, \vec{s}} \geq \alpha_L(w) - \frac{\epsilon}{3}} \leq \exp\parens{-\frac{2 n \epsilon^2}{w^k p}} \leq \\
\leq{} & \Pr\brackets{\max_{\vec{b},\vec{s}} V_\GHZ^{(n,k)} \parens{\bm{X}, \vec{b}, \vec{s}} \geq \EE[\max_{\vec{b},\vec{s}} V_\GHZ^{(n,k)} \parens{\bm{X}, \vec{b}, \vec{s}}] + \frac{\epsilon}{3}} \leq \exp\parens{-\frac{2 n \epsilon^2}{9 w^k p}}.
\end{aligned}
\end{equation}

But we also have that
\begin{equation}
\begin{aligned}
& \Pr\brackets{\max_{\vec{b},\vec{s}} V_\GHZ^{(n,k)} \parens{\bm{X}, \vec{b}, \vec{s}} \geq \alpha_L(w) - \frac{\epsilon}{3}} \geq \\
\geq{} & \exp\brackets{n\parens{\entropy(w) - \frac{4 (\alpha_L(w) - \epsilon/3)^2}{w^k p}}} = \\
={} & \exp\parens{\frac{4 n}{w^k p} \parens{\frac{2 \alpha_L (w) \epsilon}{3} - \frac{\epsilon^2}{9}}}.
\end{aligned}
\end{equation}

But observe that if we set $\epsilon$ to be sufficiently small, then

\begin{equation}
\frac{4 n}{w^k p} \parens{\frac{2 \alpha_L (w) \epsilon}{3} - \frac{\epsilon^2}{9}} - \parens{-\frac{2 n \epsilon^2}{9 w^k p}} > 0,
\end{equation}
which means that there is a contradiction between the lower bound and the concentration bound. Thus,

\begin{equation}
\EE[\max_{\vec{b},\vec{s}} V_\GHZ^{(n,k)} \parens{\bm{X}, \vec{b}, \vec{s}}] \geq \alpha_L (w) - \frac{2 \epsilon}{3},
\end{equation}

which means that for any $w$,
\begin{equation}
\Pr\brackets{\max_{\vec{b},\vec{s}} V_\GHZ^{(n,k)} \parens{\bm{X}, \vec{b}, \vec{s}} \leq \alpha_L (w) - \epsilon} \leq \exp\parens{-\frac{2 n \epsilon^2}{9 w^k p}} = o_n (1),
\end{equation}
so in particular,
\begin{equation}
\Pr\brackets{\max_{\vec{b},\vec{s}} V_\GHZ^{(n,k)} \parens{\bm{X}, \vec{b}, \vec{s}} \leq \alpha_L - \epsilon} = o_n (1).
\end{equation}
\end{proof}

\subsection{Evidence Against the QOGP}

Observe that for large values of $k$, the upper and lower bounds match and are attained when $w$ is very close to $0$ or $1$. Thus, the space of near-optimal solutions consists of very low-entropy states, which prevents the clustering phenomenon from occurring. While we do not formally prove the non-existence of the QOGP for this problem, the known techniques for establishing it fail and it is possible that there exist average-case efficient quantum algorithms for the GHZ Hypergraph Hamiltonian for large values of $k$.

\section{Results from Ramsey Theory}\label{sec:ramsey_theory_results}

Here, we reproduce results from Ramsey Theory that are used to construct a forbidden configuration in \Cref{lem:strong_hardness_existence} in the strong hardness proof. This follows a similar argument to \citet{Gamarnik2021_algorithmic_obstructions}, but our version is simpler and yields a better bound, independently strengthening the algorithmic hardness results of \citet{Gamarnik2021_algorithmic_obstructions,9996948,pmlr-v195-gamarnik23a,Gamarnik2023_shattering}.

\begin{definition}[Ramsey numbers]
For positive integers $m_1, \dots, m_Q \geq 2$, the \emph{Ramsey number} $R(m_1, \dots, m_Q)$ is the smallest value of $T$ such that for any $Q$-coloring of the edges of $K_T$ (the complete graph with $T$ vertices), the graph contains an $m_j$-clique of color $j$ for some $1 \leq j \leq Q$. Also define $R_Q (m) \coloneq R(m_1, \dots, m_Q)$ when $m_j = m$ for all $1 \leq j \leq Q$.
\end{definition}

\begin{proposition}[Version of the argument of~\citet{Erdos1935_Ramsey} for many colors]\label{prop:multicolor_ramsey_general}
\begin{equation}
R(m_1, \dots, m_Q) \leq Q^{\sum_{j=1}^Q m_j}.
\end{equation}
\end{proposition}
\begin{proof}
We proceed by induction on the number of colors $Q$ and on the values $m_1, \dots, m_Q$. First, suppose that the claim holds for $Q - 1$. If $m_j = 2$ for some $j$, then if the graph does not contain a $2$-clique of color $j$ it does not use this color, so \begin{equation}
R(m_1, \dots, m_Q) = R(m_1, \dots, m_{j-1}, m_{j+1}, \dots, m_Q) \leq (Q - 1)^{\sum_{j=1}^Q m_j - 2} \leq Q^{\sum_{j=1}^Q m_j}.
\end{equation}
In the case $Q = 2$, we have that $R(2, 2) = 2 \leq 2^4$. Now, we can fix $Q$ and induct on $m_1, \dots, m_Q$. Suppose that the claim holds for $R(m_1, \dots, m_j - 1, \dots, m_Q)$ for all $1 \leq j \leq Q$. Let
\begin{equation}
n = \sum_{j=1}^Q R(m_1, \dots, m_j - 1, \dots, m_Q),
\end{equation}
and consider any $Q$-coloring of the edges of $K_n$. Then, for any vertex $v$, by the pigeonhole principle, there must exist some $1 \leq j \leq Q$ such that $v$ has at least $R(m_1, \dots, m_j - 1, \dots, m_Q)$ adjacent edges of color $j$. Consider the neighbors of vertex $v$ connected to it by color $j$. Either they contain an $m_{j'}$-clique for $j' \neq j$ or they contain an $m_j - 1$-clique of color $j$, creating an $m_j$-clique of color $j$ when taken together with $v$. Thus,
\begin{equation}
\begin{aligned}
& R(m_1, \dots, m_Q) \leq \sum_{j=1}^Q R(m_1, \dots, m_j - 1, \dots, m_Q) \leq \\
\leq{} & Q \cdot Q^{\parens{\sum_{j=1}^Q m_j} - 1} = Q^{\sum_{j=1}^Q m_j}.
\end{aligned}
\end{equation}
\end{proof}

\begin{corollary}\label{cor:multicolor_ramsey_diag}
\begin{equation}
R_Q (m) \leq Q^{Q m}.
\end{equation}
\end{corollary}

\begin{proof}
This follows immediately from \Cref{prop:multicolor_ramsey_general}.
\end{proof}

\begin{corollary}[{Improved version of~\citet[Proposition 6.12]{Gamarnik2021_algorithmic_obstructions}}]\label{cor:ramsey_m_clique}
Let $T = (Q + 1)^{(Q + 1) m}$. Consider a graph $G$ with $T$ vertices such that any subgraph consisting of $m$ vertices contains at least one edge, and all edges in $G$ are assigned one of $Q$ colors. Then, $G$ must contain a monochromatic $m$-clique.
\end{corollary}

\begin{proof}
We can treat the absence of an edge as an additional color $Q + 1$ and apply \Cref{cor:multicolor_ramsey_diag}. We know that any graph of size $T$ must contain an $m$-clique of one of the colors, but since it does not contain any $m$-clique of the color $Q + 1$ by assumption, it must contain an $m$-clique of one of the colors $1, \dots, Q$.
\end{proof}

\end{document}